\documentclass{article}

\usepackage{latexsym}
\usepackage{amsmath}
\usepackage{amsthm} 
\usepackage{mathtools}
\usepackage{forloop}
\usepackage[paper=letterpaper,margin=1in]{geometry}
\usepackage[ruled,vlined,linesnumbered]{algorithm2e}
\usepackage{nicefrac}
\usepackage{paralist}
\usepackage{xcolor} 
\usepackage{tabularx}
\usepackage{xifthen}
\usepackage{makecell}
\usepackage{caption}
\usepackage{subcaption}

\usepackage{tikz}
\usetikzlibrary{arrows.meta, positioning, quotes}

\usepackage{epsfig}
\usepackage{amssymb}
\usepackage[hidelinks]{hyperref}
\usepackage{graphicx}

\usepackage{natbib}
\usepackage{authblk}

\newtheorem{theorem}{Theorem}
\newtheorem{lemma}[theorem]{Lemma}

\newtheorem{corollary}[theorem]{Corollary}

\newtheorem{proposition}[theorem]{Proposition}
\newtheorem{observation}[theorem]{Observation}

\newcommand{\defcal}[1]{\expandafter\newcommand\csname c#1\endcsname{{\mathcal{#1}}}}
\newcommand{\defbb}[1]{\expandafter\newcommand\csname b#1\endcsname{{\mathbb{#1}}}}
\newcounter{calBbCounter}
\forLoop{1}{26}{calBbCounter}{
	\edef\letter{\Alph{calBbCounter}}
	\expandafter\defcal\letter
	\expandafter\defbb\letter
}

\newcommand{\eps}{\varepsilon}
\newcommand{\ie}{{\it i.e.}}

\newcommand{\Reals}{\mathbb{R}}
\newcommand{\nnR}{\Reals_+}
\DeclareMathOperator*{\argmax}{arg\,max}
\newcommand{\characteristic}{{\mathbf{1}}}
\DeclarePairedDelimiter\norm\lVert\rVert

\DeclarePairedDelimiterXPP\Exp[1]{\bE}{\lbrack}{\rbrack}{}{#1}

\newcommand{\bigOsym}{\mathcal{O}}
\newcommand{\bigtOsym}{\tilde{\bigOsym}}
\DeclarePairedDelimiterXPP\bigO[1]{\bigOsym}{\lparen}{\rparen}{}{#1}
\DeclarePairedDelimiterXPP\bigtO[1]{\bigtOsym}{\lparen}{\rparen}{}{#1}
\newcommand{\littleOsym}{o}
\DeclarePairedDelimiterXPP\littleO[1]{\littleOsym}{\lparen}{\rparen}{}{#1}

\newcommand{\gnd}{\cN}
\newcommand{\OPT}{OPT}

\newcommand{\SetF}[3]{{#1_{#2}^{(#3)}}} 
\newcommand{\numiter}{T}
\newcommand{\retsol}{S}
\newcommand{\maxgain}{\Delta_f}
\newcommand{\threshold}{\tau}
\newcommand{\density}{\rho}
\newcommand{\knapevent}{E}
\newcommand{\highdensity}{\cH}
\newcommand{\lowdensity}{\cL}

\newcommand{\mainalg}{\textsc{{SimultaneousGreedys}}\xspace} 
\newcommand{\fastalg}{\textsc{FastSGS}\xspace}
\newcommand{\knapsackalg}{\textsc{{KnapsackSGS}}\xspace}
\newcommand{\densitysearchalg}{\textsc{{DensitySearchSGS}}\xspace}
\newcommand{\samplegreedy}{\textsc{{SampleGreedy}}\xspace}
\newcommand{\repeatedgreedy}{\textsc{{RepeatedGreedy}}\xspace}
\newcommand{\greedy}{\textsc{{Greedy}}\xspace}

\newcommand{\usm}{\textsc{{USM}}\xspace}
\newcommand{\modgreedy}{\textsc{{ModifiedGreedy}}\xspace}
\newcommand{\modrepeatedgreedy}{\textsc{{ModifiedRepeatedGreedy}}\xspace}
\newcommand{\densitysearchRG}{\textsc{{DensitySearchRG}}\xspace}

\newtheorem{manualtheoreminner}{Proposition}
\newenvironment{manualtheorem}[1]{%
	\renewcommand\themanualtheoreminner{#1}%
	\manualtheoreminner
}{\endmanualtheoreminner}

\newcommand{\package}{\texttt{SubmodularGreedy.jl}\xspace}

\title{
	How Do You Want Your Greedy: Simultaneous or Repeated?%
	\thanks{Parts of the repeated greedy analysis and the inapproximability results presented in this paper have previously appeared in a preliminary form in a conference paper that appeared in COLT 2017 \citep{FHK17}.}
}

\author[1]{Moran Feldman}
\author[2]{Christopher Harshaw}
\author[3]{Amin Karbasi}
\affil[1]{University of Haifa, Department of Computer Science}
\affil[2]{Yale University, Department of Computer Science}
\affil[3]{Yale University, Departments of Electrical Engeering, Computer Science, Statsitics \& Data Science}

\date{}

\begin{document}

\maketitle

\begin{abstract}
We present \mainalg, a deterministic algorithm for constrained submodular maximization.
At a high level, the algorithm maintains $\ell$ solutions and greedily updates them in a \emph{simultaneous} fashion, rather than a sequential one.
\mainalg achieves the tightest known approximation guarantees for both $k$-extendible systems and the more general $k$-systems, which are $(k+1)^2/k = k + \bigO{1}$ and $(1 + \sqrt{k+2})^2 = k + \bigO{\sqrt{k}}$, respectively.
This is in contrast to previous algorithms, which are designed to provide tight approximation guarantees in one setting, but not both.
We also improve the analysis of \repeatedgreedy, showing that it achieves an approximation ratio of $k + \bigO{\sqrt{k}}$ for $k$-systems when allowed to run for $\bigO{\sqrt{k}}$ iterations, an improvement in both the runtime and approximation over previous analyses.
Furthermore, the approximation guarantees of both algorithms further improve to $k+1$ when the objective is monotone.
We demonstrate that both algorithms may be modified to run in nearly linear time with an arbitrarily small loss in the approximation.
This leads to the first nearly linear time algorithm for submodular maximization over $k$-extendible systems and $k$-systems.

Both \mainalg and \repeatedgreedy are flexible enough to incorporate the intersection of $m$ additional knapsack constraints, while retaining similar approximation guarantees.
In particular, both algorithms yield an approximation guarantee of roughly $k + 2m + \bigO{\sqrt{k+m}}$ for $k$-systems and \mainalg enjoys an improved approximation guarantee of $k+2m + \bigO{\sqrt{m}}$ for $k$-extendible systems.
To complement our algorithmic contributions, we provide a hardness result which states that no algorithm making polynomially many queries to the value and independence oracles can achieve an approximation better than $k + \nicefrac{1}{2} + \eps$.
We also present \package, a Julia package which implements these algorithms and may be downloaded at \href{https://github.com/crharshaw/SubmodularGreedy.jl}{this URL}.
Finally, we test the effectiveness of these algorithms on real datasets.
\end{abstract}

\section{Introduction} \label{sec:introduction}

Submodular optimization has become widely adopted into the methodology of many areas of science and engineering.
In addition to being a flexible modeling paradigm, submodular functions are defined by a diminishing returns property that naturally appears in a variety of disciplines, from machine learning and information theory to economics and neuroscience.
Submodular optimization has been used in sensor placement \citep{krause05near}, maximum likelihood inference in determinantal point processes \citep{GKT12}, influence maximization \citep{kempe03}, functional neuroimaging \citep{SKSC2017},  data summarization \citep{lin2011class, Mirzasoleiman13}, crowd teaching \citep{singla14}, black-box interpretability \citep{Elenberg17}, decision making \citep{alieva2020learning, chen2015submodular}, and experimental design \citep{bian2017guarantees, harshaw19a}, to name a few examples. For more information on the applications of submodularity in machine learning and signal processing, we refer the interested reader to the recent survey by \cite{tohidi2020submodularity}. 
The simplest constraint class in these optimization problems is a cardinality constraint, which limits the number of elements any feasible solution may contain.
However, as more applications emerge, there is a growing need for the development of fast algorithms that are able to handle more flexible and expressive constraint classes.

In this paper, we study the problem of maximizing a submodular functions subject to two constraint classes: $k$-systems and its (strict) subclass of $k$-extendible systems. 
These constraint classes capture a wide variety of constraints, including cardinality constraints, spanning trees, general matroids, intersection of matroids, graph matchings, scheduling, and even planar subgraphs.
In the literature, there are two main algorithmic approaches for maximizing submodular functions over each of these constraint classes.
The repeated greedy approach was initially proposed for submodular optimization over $k$-systems by \citet{GRST10}, who showed that $\bigO{k}$ repeated iterations suffice to achieve a $3k$ approximation guarantee. 
\citet{MBK16} refined this analysis, improving the approximation guarantee to $2k$.
One contribution of the current work is to further improve the analysis of the repeated greedy technique, showing that $\bigO{\sqrt{k}}$ suffices to achieve a $k + \bigO{\sqrt{k}}$ approximation guarantee.
The subsample greedy approach was proposed by \cite{FHK17} for submodular optimization over a $k$-extendible system, and achieves an improved approximation ratio of $(k+1)^2/k = k + \bigO{1}$. 

One of the main downsides to these current approaches is that they are tailor made for the particular constraint class and do not perform as well otherwise.
As we show in this paper, our analysis of the repeated greedy technique is tight in the sense that the algorithm attains an approximation guarantee of only $k + \Omega(\sqrt{k})$ for the subclass of $k$-extendible systems, regardless of the number of repeated iterations; similarly, the subsample greedy approach is not known to provide any approximation guarantee for the more general $k$-systems.
Moreover, the types of approximation guarantees provided by the two algorithmic approaches differ: subsampling approaches are randomized algorithms, and their approximation guarantees hold in expectation---which may be too weak for certain applications where strong deterministic guarantees are preferable.
Another downside is that while repeated greedy approaches may be modified to handle additional knapsack constraints \citep{MBK16}, we are not aware of any known adaptation of subsampling greedy that allows it to handle such additional constraints.

Our main contribution in this work is \mainalg, a deterministic algorithm for constrained submodular maximization.
The new algorithmic idea is to greedily construct $\ell$ disjoint solutions in a \emph{simultaneous} fashion.
The solutions are all initialized to be empty; and at each iteration, an element is added to a solution in a greedy fashion, maximizing the marginal gain amongst all feasible element-solution pairs.
At the end of the algorithm, the best solution is returned amongst the $\ell$ constructed solutions.
One may interpret this \mainalg as a derandomization of the subsample greedy technique.
Subsample greedy produces a random solution whose objective value is large, in expectation; however, the support of the solution is exponentially sized, and so a na\"{i}ve derandomization is infeasible.
We show that the average objective value of the $\ell$ deterministically constructed solutions in \mainalg is just as large, and in this sense we reduce the support of the distribution from exponential to constant.

Unlike the previous algorithmic techniques which were limited to specific constraint types, we show that \mainalg achieves the best known approximation guarantees of $(1 + \sqrt{k +2})^2 = k + \bigO{\sqrt{k}}$ and $(k+1)^2/k = k + \bigO{1}$ for $k$-systems and $k$-extendible systems, respectively.
In fact, these approximation ratios guaranteed by \mainalg further improve to $k+1$ when the submodular objective function is monotone (in the case of $k$-systems, one needs to modify the value of $\ell$ to get this improvement).

Another contribution of this work is to show that both  \mainalg and \repeatedgreedy may be modified to create several different variants.
First, we show that by employing an approximate greedy search based on a marginal gain thresholding technique \citep{Badanidiyuru14}, both algorithms can be made to require only $\bigtO{n / \eps}$ queries to the value and independence oracles\footnote{Throughout the paper, we use the $\bigtOsym$ notation to suppress poly-logarithmic factors.}  at the cost of a $1 + \eps$ factor increase in the approximation guarantees.
To our knowledge, this is the first nearly linear time algorithm for submodular maximization over a $k$-system.
Next, we show that additional knapsack constraints may be incorporated into both algorithms by incorporating a density threshold technique \citep{MBK16} in the greedy selection procedure.
Not only does this work improve upon the approximation guarantees and efficiency of \cite{MBK16} for submodular maximization subject to a $k$-system constraint and $m$ additional knapsacks, this work is also the first to provide (further improved) approximations when the subclass of $k$-extendible systems are considered.
Even with these nearly linear time and knapsack modifications, the approximation guarantees of \mainalg are still adaptive in the sense that they improve for $k$-extendible systems and they further improve when the objective function is monotone.
For this reason, we consider \mainalg to be like a Swiss Army knife for constrained submodular maximization: it is one main tool (the simultaneous greedy procedure) with several variants (nearly linear run time, density ratio technique) that can be used to produce the best known results for several problems of interest including $k$-systems, $k$-extendible systems, intersection of these with additional knapsacks, and a possibly monotone objective.
For a succinct summary of the comparison to previous work, see Table~\ref{table:alg-summary}.

\begin{table}
	\begin{center}\scalebox{0.9}{\begin{tabular}{c | c | c | c }
		\textbf{Algorithm} & \textbf{Running Time} & \textbf{$k$-system} & \textbf{$k$-extendible system} \\\hline
		\makecell{Repeated Greedy \\ {\scriptsize \citep{GRST10}}} 
			& $\bigO{n^2}$ & $3k$ & {\scriptsize (same as for $k$-system)}  \\[0.15in] 
		\makecell{Sample Greedy \\ {\scriptsize \citep{FHK17}}} 
			& $\bigO{n^2}$& - & \makecell{$k + \bigO{1}$ \\ {\scriptsize (in expectation)}  } \\[0.15in]
		\makecell{Repeated Greedy \\ {\scriptsize \textbf{(this work)}}} 
		& $\bigO{n^2}$ & $k + \bigO{\sqrt{k}}$ & {\scriptsize (same as for $k$-system)}  \\[0.15in] 
		\makecell{\mainalg \\ {\scriptsize \textbf{(this work)}}}  
			& $\bigO{n^2}$ & $k + \bigO{\sqrt{k}}$ & $k + \bigO{1}$\\[0.15in]
		\makecell{\fastalg \\ {\scriptsize\textbf{ (this work)}}}  
			& $\bigtO{n / \eps}$ & $(1 + \eps) k + \bigO{\sqrt{k}}$ & $(1 + \eps)k + \bigO{1}$ \\ \hline
		\makecell{FANTOM \\ {\scriptsize \citep{MBK16}}} 
			& $\bigtO{n^2 / \eps}$ & \makecell{$(1 + \eps)(2k + (2 + 2/k) m)$ \\ ${} + \bigO{1}$} & {\scriptsize (same as for $k$-system)} \\[0.15in]
		\makecell{\densitysearchalg \\ {\scriptsize \textbf{(this work)}}}  
			& $\bigtO{n / \eps}$ & \makecell{$(1+\eps)(k + 2m)$ \\ ${}+ \bigO{\sqrt{k + m}}$} & \makecell{$(1+\eps)(k + 2m)$ \\ ${}+ \bigO{\sqrt{m}}$}
	\end{tabular}}\end{center}
	\caption{
		A comparison with previous works. 
		For the sake of clarity, the dependence on $k$ is suppressed from the running times. 
		The last two rows involve $m$ knapsacks constraints in addition to the independence system constraint.
		Only modifications of the simultaneous greedy approach are shown, while modifications of the repeated greedy approach presented in this paper are suppressed.
	}
	\label{table:alg-summary}
\end{table}

We compliment these algorithmic contributions with a hardness result, showing that no algorithm making polynomially many queries to the value and independence oracles can yield an approximation factor smaller than $k + \nicefrac{1}{2} - \eps$ over a $k$-extendible system.
This hardness result demonstrates that the approximation produced by \mainalg in the setting of $k$-extendible systems is nearly tight, and we prove it using the symmetric gap technique of \cite{V13}. 
Note that because $k$-extendible systems are a subclass of $k$-systems, our hardness also holds for the more general class of $k$-systems; however, whether the additional $\bigO{\sqrt{k}}$ term in the approximation factor is necessary for this class remains an open question. Moreover, an almost as strong hardness of $k - \eps$ was already shown for $k$-systems by~\cite{Badanidiyuru14}.

\paragraph{Organization}
The organization of the remainder of the paper is as follows:
In the remainder of Section~\ref{sec:introduction}, we review the related works.
We present the preliminary definitions and problem statement in Section~\ref{sec:preliminaries}.
In Section~\ref{sec:main-alg}, we present \mainalg and its analysis.
Section~\ref{sec:fast-alg} contains the nearly linear time modification and Section~\ref{sec:knapsack-alg} contains the additional knapsack modification.
Our improved analysis of \repeatedgreedy, including linear-time and knapsack modifications, is contained in Section~\ref{sec:repeated-greedy}.
The hardness results are presented in Section~\ref{sec:hardness}.
Section~\ref{sec:practical} contains practical considerations when implementing these algorithms as well as a description of the \package package.
Section~\ref{sec:experiments} contains experiments on real datasets. 
Finally, we conclude in Section~\ref{sec:conclusion}.

\subsection{Related Work} \label{sec:related-work}

The study of sumodular maximization over $k$-systems goes back to \cite{FNW78} who proved that the natural greedy algorithm obtains $(k+1)$-approximation when the objective function is monotone.
Algorithms for maximizing \emph{non-monotone} submodular functions under a $k$-system constraint, however, were not obtained until much more recently.
\cite{GRST10} proposed a repeated greedy approach for this problem.
At a high level, the algorithm repeatedly performs the following procedure: run the greedy algorithm to obtain a solution $S$, then perform unconstrained maximization on the elements of $S$ to produce a set $S'$, and finally remove the larger set $S$ from the ground set. 
Among all considered solutions $S$ and $S'$, the set with the largest objective value is returned.
\cite{GRST10} proved that when the number of iterations of repeated greedy is set to $k+1$, then the approximation ratio is roughly $3k$.
This analysis was improved by \cite{MBK16} who showed that the same repeated greedy algorithm achieves an approximation ratio of roughly $2k$.
Note that because the repeated greedy based algorithms first consider the set returned by the greedy algorithm, their approximation ratios automatically improve to $k+1$ for monotone objectives.
We also remark that \cite{MBK16} demonstrated that the repeated greedy technique may be modified to incorporate additional knapsack constraints through the use of a density thresholding technique.

An important subclass of the $k$-systems are the $k$-extendible systems, which were defined by \cite{Mestre2006}. For this subclass, \cite{FHK17} introduced a subsampling approach as an alternative to repeated greedy, yielding an algorithm that is faster and also enjoys a somewhat better approximation guarantee.
The idea is to independently subsample elements of the ground set, and then run the greedy algorithm once on the subsample.
This subsampling approach runs in expected time $\bigO{n + nr/ k}$, where $r$ is the rank of the $k$-extendible system, and attains an approximation ratio of $(k+1)^2/k$ in expectation.
The main downside to this approach is that the approximation guarantee holds only in expectation, and thus, repetition is necessary to achieve a good approximation ratio with a high probability. 
The authors also show that even a very small number of repetitions suffices in practice. 
Nevertheless, the inherent uncertainty in the approximation quality of the returned solution may be undesirable in certain scenarios.

The class of $k$-extendible systems includes in its turn other subclasses of interest, including the class of $k$-exchange systems introduced by~\cite{FNSW11} and the well know class of $k$-intersection, which includes constraints that can be represented as the intersection of $k$ matroids. Naturally, the above mentioned subsampling technique of \cite{FHK17} for $k$-extendible systems applies also to constraints from these two subclasses, and is arguably the best approximation ratio that can be achieved for these classes using practical techniques. However, local search approaches have been used to achieve improved approximation ratios for both these subclasses whose time complexity is exponential in both $k$ and some error parameter $\eps > 0$---which makes these improved approximation ratios mostly of theoretical interest (except maybe when $k$ is very small). Specifically, for the intersection of $k \geq 2$ matroids, \cite{LMNS10} proved an approximation ratio of $k + 2 + 1/k + \eps$, which was later improved to $k + 1 + 1/(k + 1) + \eps$ by \cite{LSV10}. The last approximation ratio was later extended also to $k$-exchange systems by \cite{FNSW11}. The case of $k = 1$, in which all the above classes reduce to be the class of matroids, was also studied extensively, and the currently best approximation ratios for this case is roughly $(0.385)^{-1} \approx 2.60$ \citep{Buchbinder2016b}.
For an in-depth discussion of these algorithmic techniques, we refer readers to the survey of \citet{BF18}.


The run time of greedy methods is typically quadratic, as each iteration requires examining all the remaining elements in the ground set.
A heuristic often used to reduce this time complexity is the so-called ``lazy greedy" approach, which uses the submodularity of the objective to avoid examining elements that cannot have the maximal marginal gain in a given iteration \citep{Minoux1978}.
While this method typically yields a substantial improvement in practice, it does not improve the worst-case time complexity.
However, inspired by the lazy greedy approach, \cite{Badanidiyuru14} proposed a technique known as ``marginal gain thresholding", which reduces the run time of the greedy algorithm to $\bigO{n / \eps \cdot \log(n / \eps)}$, while incurring only a small additive $\eps$ factor in the approximation ratio.
Later on, \cite{Mirzasoleiman15} proposed a stochastic approach which further reduces the run time of the greedy algorithm to $\bigO{n \log(1 / \eps)}$, but applies only in the context of the simple cardinality constraint. Additional fast algorithms for submodular maximization were suggested by \citep{Badanidiyuru14,Buchbinder2017,Ene2019,Ene2019b}.
It is also worth mentioning that most of the above algorithms can be further improved, in practice, using a lazy greedy like approach.


Our simultaneous greedy technique is most closely related to a recent work of \cite{Kuhnle2019}, where a similar ``interlaced greedy'' approach is proposed to obtain a $\nicefrac{1}{4} - \eps$ approximation for maximizing a non-monotone submodular function subject to a cardinality constraint.
The proposed idea is similar: simultaneously run the greedy algorithm to construct two disjoint solutions.
In addition to extending to more general settings and subsuming these approximation results, the current work also demonstrates a tighter analysis even for the cardinality constraint presented in \cite{Kuhnle2019}.
Namely, our analysis shows that only one run of the simultaneous greedy technique is required to obtain the $\nicefrac{1}{4} - \eps$ approximation, whereas the analysis of \cite{Kuhnle2019} requires the algorithm to be run twice in order to obtain this approximation.
After an initial preprint of this work appeared online, \cite{han2021power} demonstrated that combining the simultaneous greedy approach with the subsampling approach yields improvements in the running time and the low order terms of the approximation guarantees.

\section{Preliminaries} \label{sec:preliminaries}
In this section, we introduce several preliminary definitions required for the problem we investigate.
In Section~\ref{sec:submodular-functions}, we define submodular set functions, which are the class of objective functions we consider.
In Section~\ref{sec:independent-systems}, we discuss the independence systems that act as constraints in our problem. Finally, Section~\ref{sec:problem-statement} formally defines our problem. 

\subsection{Submodular functions} \label{sec:submodular-functions}
Let $\gnd$ be a finite set of size $n$ which we refer to as the \emph{ground set}.
A real valued set function $f\colon 2^\gnd \rightarrow \Reals$ is \emph{submodular} if for all sets $X, Y \subseteq \gnd$,
\[
f(X \cup Y) + f(X \cap Y) \leq f(X) + f(Y) \enspace.
\]
Given a set $S$ and element $e$, we use the shorthand $S + u$ to denote the union $S \cup \{u\}$. Additionally, we define $f(u \mid S) = f(S + u) - f(S)$, i.e., $f(u \mid S)$ is the marginal gain with respect to $f$ of adding $e$ to the set $S$.\footnote{More generally, we define $f(X \mid Y) = f(X \cup Y) - f(Y)$ for all sets $X, Y \subseteq \gnd$.} An equivalent definition of submodularity is that a function $f$ is submodular if for all subsets $A \subseteq B \subseteq \gnd$ and element $u \notin B$,
\begin{equation} \label{eq:diminishing-returns}
f(u \mid A) \geq f(u \mid B) \enspace.
\end{equation}
Inequality~\eqref{eq:diminishing-returns} is referred to as the diminishing returns property.
Indeed, if $f$ is interpreted as a utility function, then Inequality~\eqref{eq:diminishing-returns} states that the marginal gain of adding an element $e$ to a subset decreases as the subset grows. Throughout the paper, we restrict our attention to non-negative submodular functions, i.e., functions whose value is non-negative for every set. The non-negativity is a necessary condition for obtaining multiplicative approximation guarantees.

A set function $f$ is \emph{modular} (or linear) if Inequality~\eqref{eq:diminishing-returns} always holds for it with equality.
Any modular function can be represented using the form
\[
f(S) = \sum_{u \in S} c_u +b
\] 
for an appropriate choice of a real number $c_u\in \Reals$ for every element $u \in \gnd$ and a fixed bias $b\in \Reals$. Finally, a set function $f$ is \emph{monotone} if adding more elements only increases its value; that is, $f(A) \leq f(B)$ for all subsets $A \subseteq B \subseteq \gnd$.

\subsection{Independence systems} \label{sec:independent-systems}
The feasible sets in the optimization problems that we consider are described by an independence system. 
For $\cI \subseteq 2^\gnd$, the pair $(\cI, \gnd)$ is an \emph{independence system} if $\cI$ is non-empty and satisfies the down-closure property, i.e., if $A \subseteq B$ and $B \in \cI$ then $A \in \cI$.
For notational simplicity, we occasionally refer to the independence system as $ \cI $ when the ground set $\gnd$ is clear from context.
A set $A \subseteq \gnd$ is called \emph{independent} in the independence system $\cI$ if $A \in \cI$. Furthermore, if $A$ is maximal independent set with respect to inclusion among all the subsets of a given set $B \subseteq \gnd$, then $A$ is called a \emph{base} of $B$. A base of the ground set $\gnd$ is also called a base of the independence system. The cardinality of the largest independent set of a given independence system $\cI$ is known as the \emph{rank} of the independence system, and we use $r$ to denote it when the independence system is clear from the context.

There is a wide variety of independence systems which have been studied in the literature, and we review some of them here.
An independence system $(\cI, \gnd)$ is a \emph{$k$-system} if for every set $B \subseteq \gnd$ the ratio between the sizes of any two bases of $B$ is at most $k$.
Any independence system
is a $k$-system for some $k \leq n$; however, we are most interested in settings where $k$ is a constant or otherwise small with respect to the number $n$ of elements in the ground set.
A subclass of independence systems are the $k$-extendible systems.
Intuitively, a $k$-extendible system is an independence system in which adding an element $u$ to any independent set requires removing at most $k$ elements to maintain independence. Formally, this means that an independence system is \emph{$k$-extendible} if for every pair of independent sets $A \subseteq B \in \cI$ and element $u \notin B$ such that $A + u \in \cI$, there exists a set $Y \subseteq B \setminus A$ of size at most $k$ such that $(B \setminus Y) \cup \{u\}$ is independent. 
It is known that every $k$-extendible system is also a $k$-system~\citep{CCPV11}, that the intersection of a $k_1$-extendible system and a $k_2$-extendible system is a $(k_1+k_2)$-extendible system and that the intersection of a $k_1$-system and a $k_2$-system is a $(k_1+k_2)$-system~\citep{HKFK20}. These observations provide a way to build more complex independence systems from simpler ones, allowing for a flexible framework for constraints in the optimization problems we consider.

One of the most well-studied examples of a $k$-extendible systems are the $1$-extendible systems, which are also known as matroids.\footnote{See~\cite{Mestre2006} for a proof of the equivalence of a $1$-extendible system with the traditional definition of matroids.}
Matroids capture a wide variety of set constraints, including independent sets of vectors, cardinality-constrained partitions, spanning forests, graph matchings, and the simple cardinality constraint.
The intersection of $k$-matroids is a $k$-extendible system, but the converse is generally not true for $k \geq 2$.
Indeed, the class of $k$-extendible systems includes systems which are not expressible as the intersection of a few matroids, including the class of $b$-matchings in graphs (which are $2$-extendible) and asymmetric TSP (which is $3$-extendible), as well as certain scheduling formulations \citep{Mestre2006}.
Although the class of $k$-systems is strictly larger than the class of $k$-extendible systems, the majority of interesting examples are $k$-extendible systems.
There are, however, a few exceptions such as the collection of all subsets of edges of a graph which induce a planar subgraph, which is $3$-system \citep{HKFK20}.
The taxonomy  of the independence systems discussed above is depicted below, and all the containments are known to be strict for $k \geq 2$
\[
\text{cardinality constraint} \subset
\text{matroid} \subset
\text{intersection of $k$ matroids} \subset
\text{$k$-extendible system} \subset
\text{$k$-system}
\enspace.
\]

Knapsack constraints are another popular family of constraints that can be represented as independence systems.
Formally, an independence system capturing a knapsack constraint is defined as the collection of sets $S \subseteq \gnd$ obeying $c(S) \leq 1$ for some non-negative modular function $c(S) = \sum_{u \in S} c_u$.
We are often interested below in the intersection of $m$ knapsack constraints, and denote the corresponding modular functions by $c_1, c_2, \dotsc, c_m$.
In this work, and more broadly in the literature, knapsack constraints are considered separately from the main independence system constraint.
Technically, this is not completely necessary because the intersection of $m$ knapsacks is a $k$-extendible system for some $k$. However, this $k$ might be as large as $m \cdot \left( c_{\max} / c_{\min} \right)$, where $c_{\max}$ and $c_{\min}$ are the largest and smallest knapsack coefficients, respectively (i.e., $c_{\max} = \max_{u \in \gnd, i \in [m]} c_i(u)$ and $c_{\min} = \min_{u \in \gnd , i \in [m]} c_i(u)$). In contrast, treating the knapsack constraints as separate from the underlying independence system allows us to aim for approximation ratios that depend only on $m$, and is thus preferable.

\subsection{Problem statement} \label{sec:problem-statement}
In this paper we study the problem of maximizing a non-negative submodular function subject to an independence system and the intersection of $m$ knapsack constraints.
More precisely, we aim to solve the following optimization problem
\begin{equation} \label{eq:opt-problem}
\begin{array}{llll}
& \max
& & f(S) \\
& \text{subject to}
& & S \in \cI \\
&  && c_i(S) \leq 1 \quad \forall \ i=1 \dots m \enspace ,
\end{array}
\end{equation}
where $f$ is non-negative and submodular and $\cI$ is an independence system which is either a $k$-system or a $k$-extendible system. 
For simplicity, we assume throughout the work that the singleton $\{u\}$ is a feasible solution for the above problem for every element $u \in \cN$. 
Clearly, any element violating this assumption can be removed from the ground set without affecting the set of feasible solutions.
We also denote by $\OPT$ an optimal solution to the program.

We evaluate our algorithms by their running times and approximation ratios.
As is standard in the literature, our algorithms access the objective function and the independence system constraint only through value and independence oracles, respectively.
The \emph{value oracle} takes as input a set $S \subseteq \gnd$ and returns $f(S)$---the evaluation of $f$ at $S$.
Similarly, the \emph{independence oracle} takes as input a set $S$ and indicates whether or not $S \in \cI$.
The computational efficiency of algorithms in this model is often judged based on the number of oracle queries they make, and we follow this convention.

\section{Simultaneous Greedys} \label{sec:main-alg}

In this section we present an algorithm named \mainalg for solving Problem~\eqref{eq:opt-problem} in the special case of $m = 0$, i.e., the case in which there are no knapsack constraints. The main idea behind \mainalg is to greedily and \emph{simultaneously} construct $\ell$ disjoint solutions by iteratively adding elements to the solutions in a way that maximizes the momentary marginal gain.
Formally, the algorithm begins by initializing $\ell$ solutions $S^{(1)}, S^{(2)}, \dotsc, S^{(\ell)}$ to be empty sets.
At each iteration, the algorithm considers all the pairs of element $u$ and solution $S^{(j)}$ such that
(1) $u$ does not yet belong to any of the solutions,
(2) $u$ can be added to $S^{(j)}$ without violating independence, and
(3) the addition of $u$ to $S^{(j)}$ increases the objective value of $S^{(j)}$.
The set of such pairs is denoted by $\cA$ in the pseudocode of the algorithm. Among all the considered pairs, the algorithm picks the one for which $f(u \mid S^{(j)})$ is maximal (i.e., the pair for which the addition of $u$ to $S^{(j)}$ yields the maximal increase in the value of the solution), and then adds $u$ to $S^{(j)}$. The algorithm terminates when no further pairs with all the above properties can be found. The pseudocode of \mainalg appears below as Algorithm~\ref{alg:main-alg}.

\begin{algorithm}[H]
	\caption{\mainalg($\gnd, f, \cI, \ell$)}\label{alg:main-alg}
	Initialize $\ell$ solutions, $\SetF{S}{0}{j} \gets \varnothing$ for $j =1, \dots \ell$. \\
	Initialize available ground set $\gnd_0 \gets \gnd$, and iteration counter $i \gets 1$.\\
	Initialize feasible element-solution pairs $\cA_1 = \{ (u, j) : \{u\} \in \cI , f(u \mid \varnothing) > 0, j \in [\ell] \}$. \\
	\While{ $\cA_i$ is nonempty }
	{
		Let $(u_i, j_i) \gets \max_{(u,j) \in \cA_{i}} f( u \mid \SetF{S}{i-1}{j})$ be a feasible element-solution pair maximizing the marginal gain.\\
		Update the solutions as
		$ 
		\SetF{S}{i}{j} \gets
		\left\{
		\begin{array}{lr}
		\SetF{S}{i-1}{j_i} + u_i  &\text{if } j = j_i\\
		\SetF{S}{i-1}{j} &\text{if } j \neq j_i
		\end{array}
		\right.$ \\
		Update the available ground set $\gnd_{i} \gets \gnd_{i-1} - u_i$. \\
		Update the feasible element-solution pairs, $\cA_{i+1} = \{ (u, j) : u \in \gnd_{i}, \SetF{S}{i}{j} + u \in \cI , f(u \mid  \SetF{S}{i}{j}) > 0 \}$.\\
		Update iteration counter $i \gets i+1$.
	}
	\Return{the set $\retsol$ maximizing $f$ among the sets $\{\SetF{S}{i}{j}\}_{j = 1}^\ell$}.
\end{algorithm}

We begin our analysis of \mainalg by providing a bound on the number of oracle calls used by the algorithm.

\begin{observation} \label{obs:main-alg-runtime}
	\mainalg requires at most $\bigO{\ell^2 r n}$ calls to the value and independence oracles.
\end{observation}
\begin{proof}
	In every single iteration, the algorithm examines the possibility of adding each of the $n$ elements to each of the $\ell$ solutions, requiring $\bigO{\ell n}$ calls to the value and independence oracles.
	Since exactly one element is added to some solution at every iteration, the number of iterations is the sum of the cardinalities of the produced solutions, which is at most $\ell r$ because all the solutions are feasible.
	Combining the two above observations, i.e., that there are at most $\ell r$ iterations, each requiring $\bigO{\ell n}$ oracle calls, we get that the total number of oracle calls required by \mainalg is $\bigO{\ell^2 r n}$.
\end{proof}

We now present theorems proving approximation guarantees for \mainalg when $\cI$ is guaranteed to be either a $k$-system or a $k$-extendible system.
To get the tightest approximation guarantees from these theorems, one has to set the number $\ell$ of constructed solutions differently for the two classes of constraints.

\begin{theorem} \label{thm:deterministic_extendible}
	Suppose that $(\gnd, \cI)$ is a $k$-extendible system and that the number of solutions is set to $\ell = k + 1$.
	Then, \mainalg requires $\bigO{k^2 r n}$ oracle calls and produces a solution whose approximation ratio is at most $(k + 1)^2/k = k + \bigO{1}$. 
	Moreover, when $f$ is non-negative monotone submodular, then the approximation ratio improves to $k+1$.
\end{theorem}

\begin{theorem} \label{thm:deterministic_system}
	Suppose that $(\gnd, \cI)$ is $k$-system and that the number of solutions is set to $\ell = \lfloor 2 + \sqrt{k + 2} \rfloor$.
	Then, \mainalg requires $\bigO{k r n}$ oracle calls and produces a solution whose approximation ratio is at most $(1 + \sqrt{k+2})^2 = k + \bigO{\sqrt{k}}$. 
	Moreover, when $f$ is non-negative monotone submodular and the number of solutions is set to $\ell=1$, then the approximation ratio improves to $k+1$.
\end{theorem}

Note that the improved approximation for $k$-extendible systems comes at the higher computational cost of an extra $\bigO{k}$ factor in the running time. 
Moreover, the gain in approximation is only for the non-monotone setting, as the two approximation guarantees are the same for monotone objectives.
In both Theorems~\ref{thm:deterministic_extendible} and \ref{thm:deterministic_system}, the bound on the required number of oracle calls is a direct application of Observation~\ref{obs:main-alg-runtime} and the choice of $\ell$, the number of constructed solutions. The proof of the approximation ratios is more involved. In Section~\ref{sec:approx-meta-analysis}, we provide a unified meta-proof for analyzing \mainalg given a constraint obeying some kinds of parametrized properties. Then, in Sections~\ref{sec:analysis-extendible} and~\ref{sec:analysis-k-system} we show that $k$-extendible systems and $k$-systems have these properties for a proper choice of the parameters, respectively, yielding the the different approximation guarantees of Theorems~\ref{thm:deterministic_extendible} and~\ref{thm:deterministic_system}.

The second part of Theorem~\ref{thm:deterministic_system} considers $\ell = 1$, which recovers the greedy algorithm.
Although it was previously known that the greedy algorithm achieves $(k+1)$-approximation for monotone submodular objectives under a $k$-system, we remark that this result for this special setting is cleanly obtained by our unified analysis.
We also remark that, for monotone objectives, the result of Theorem~\ref{thm:deterministic_extendible} holds for any number of solutions $\ell \leq k+1$; which  further demonstrates that the analysis of the greedy algorithm is handled by our meta-analysis. 
The details for the case of $\ell \leq k+1$ are covered in the proof of Theorem~\ref{thm:deterministic_extendible}.

\SetKwIF{With}{OtherwiseWith}{Otherwise}{with}{do}{otherwise with}{otherwise}{}
 \begin{algorithm}[H]
	\caption{\samplegreedy($\cN, f, \mathcal{I}, k$)} \label{alg:sample_greedy}
	\DontPrintSemicolon
	Let $\gnd' \leftarrow \varnothing$ and $S \leftarrow \varnothing$.\\
	\For{each $u \in \gnd$}{
		\With{probability $(k+1)^{-1}$\label{line:sample}}{Add $u$ to $\gnd'$.}
	}
	\While{there exists $u \in \gnd'$ such that $S + u \in \cI$ and $ f(u \mid S) > 0$}{
		Let $u \in \gnd'$ be the element of this kind maximizing $f(u \mid S)$.\\
		Add $u$ to $S$.\\
	}
	\Return{$S$}.\\
\end{algorithm}

As mentioned in Section~\ref{sec:introduction}, one may interpret \mainalg as a de-randomization of \samplegreedy, the subsampling algorithm of \cite{FHK17} presented here as Algorithm~\ref{alg:sample_greedy}.
\samplegreedy creates a subsample of the ground set by sampling each element independently with probability $p$ and then running the vanilla greedy algorithm.
\cite{FHK17} show that, for $k$-extendible systems, setting the sampling probability to $p = (k+1)^{-1}$ yields an approximation ratio of $\nicefrac{(k+1)^2}{k}$, which improves to $k+1$ for monotone objectives (i.e., the same approximation guarantees of \mainalg for these cases). One of the key step in the analysis of \samplegreedy is an averaging argument over the distribution of solutions it may produce, whose support might be of exponential size. This means that na\"{i}vly trying to de-randomize \samplegreedy requires keeping all the states which it might take, and therefore, yields an exponential algorithm. In the analysis of \mainalg we bypass this hurdle by managing to make the above averaging argument work for a much smaller distribution whose support consists only of the $\ell$ solutions maintained by the algorithm.
We note that this idea of de-randomizing a randomized algorithm by coming up with a polynomial size distribution mimicking the behavior of an exponential size distribution was originally used in the context of submodular maximization by \citet{Buchbinder2016b}, albeit using very different techniques based on linear programming.
Finally, unlike \samplegreedy, \mainalg has the additional benefit of producing approximation guarantees for the more general class of $k$-systems.

\subsection{Meta-analysis for approximation guarantees} \label{sec:approx-meta-analysis}

In this section, we present a unified analysis for obtaining approximation guarantees for \mainalg under general independence system constraints.
Specifically, Proposition~\ref{prop:general_result} reduces the conditions for approximation to simple combinatorial statements relating the constructed solutions to $\OPT$. 
These combinatorial statements are shown to hold for $k$-extendible systems and $k$-systems in Sections~\ref{sec:analysis-extendible} and \ref{sec:analysis-k-system}, respectively.

The main idea of the unified analysis is to keep track of the elements of $\OPT$ which could have been{\textemdash}but were not{\textemdash}added to each of the $\ell$ solutions by the algorithm. 
At the beginning of the algorithm, all solutions are initialized to the empty set and so each element of $\OPT$ could be added to each solution in the first iteration.
However, every time that the algorithm adds an element to one of the solutions, it means that certain elements of $\OPT$ are now no longer able to be added to that solution, due to the independence constraint.
In this sense, these elements of $\OPT$ are ``thrown away'' from the set of possible elements to be added to the solution.
The main technical requirement of the unified approximation analysis is that only a few elements of $\OPT$ are thrown away in this sense at each iteration. 
These conditions are more precisely stated in the hypothesis of Proposition~\ref{prop:general_result}.

Let $\numiter$ be the number of iterations performed by \mainalg, and let $\SetF{U}{i}{j}$ be the singleton set $\{u_i\}$ if $j = j_i$ and the empty set otherwise.

\begin{proposition} \label{prop:general_result}
	Let us define $\SetF{O}{0}{j} =\OPT$ for every solution $1 \leq j \leq \ell$. If there exist a value $p$ and sets $\SetF{O}{i}{j}$ for every iteration $1 \leq i \leq \numiter$ and solution $1 \leq j \leq \ell$ such that
	\begin{compactitem}
		\item $\SetF{S}{i}{j} + u$ is independent for every iteration $0 \leq i \leq r$, solution $1 \leq j \leq \ell$, and element $u \in \SetF{O}{i}{j}$.
		\item $\SetF{O}{i}{j} \subseteq \SetF{O}{i - 1}{j} \cap \gnd_i$ for every iteration $1 \leq i \leq \numiter$ and solution $1 \leq j \leq \ell$.
		\item $(\SetF{S}{\numiter}{j} \setminus \SetF{S}{i}{j}) \cap \OPT \subseteq \SetF{O}{i}{j}$ for every iteration $0 \leq i \leq \numiter$ and solution $1 \leq j \leq \ell$.
		\item $\sum_{i = 1}^\ell |\SetF{O}{i - 1}{j} \setminus (\SetF{O}{i}{j} \cup \SetF{U}{i}{j})| \leq p$ for every iteration $1 \leq i \leq \numiter$.
	\end{compactitem}
	Then, the solution $\retsol$ produced by \mainalg is a $\frac{p + 1}{1 - \ell^{-1}}$-approximation solution.
	Moreover, this approximation ratio improves to $p+1$ when $f$ is monotone.
\end{proposition}

Before proceeding, we would like to provide some intuition for the conditions appearing in Proposition~\ref{prop:general_result}.
Intuitively, the set $\SetF{O}{i}{j}$ contains elements of $\OPT$ which have not already been added to a solution and can still be added to the $j$-th solution at iteration $i$.
Condition 1 formally states this ability to add the elements of $\SetF{O}{i}{j}$ to the $j$-th solution, and Condition 2 formally states that the elements in $\SetF{O}{i}{j}$ do not already appear in a solution.
Condition 3 requires $\SetF{O}{i}{j}$ to include all the elements of $\OPT$ which are eventually (but not yet) included in one of the final solutions.
Finally, Condition 4 is a bound on the number of elements which are removed from these sets at each iteration.
Together, these conditions are strong enough to provide a general approximation guarantee.

The following lemma is the first step towards proving Proposition~\ref{prop:general_result}. Intuitively, this lemma shows that as the iteration $i$ increases, the decrease in the value of $f(\SetF{O}{i}{j} \mid \SetF{S}{i}{j})$ is transferred, at least to some extent, to $\SetF{S}{i}{j}$.

\begin{lemma} \label{lem:greedy_invariant}
	Given the conditions of Proposition~\ref{prop:general_result}, for every iteration $0 \leq i \leq \numiter$,
	\[
	(p + 1) \cdot \sum_{j = 1}^\ell f(\SetF{S}{i}{j}) + \sum_{j = 1}^\ell f(\SetF{O}{i}{j} \mid \SetF{S}{i}{j})
	\geq
	\sum_{j = 1}^\ell f(\OPT \cup \SetF{S}{i}{j})
	\enspace.
	\]
\end{lemma}
\begin{proof}
	We prove the lemma by induction on the iterations $i=0, 1, \dotsc, \numiter$. 
	The base case is the case of $i = 0$, corresponding to the initialization of the algorithm.
	Recall that  the solutions are initialized to be empty, i.e.,  $\SetF{S}{0}{j} = \varnothing$ for every $j \in [\ell]$.
	This, together with non-negativity of $f$, implies
	{\allowdisplaybreaks
	\begin{align*}
	\sum_{j = 1}^\ell f(\OPT \cup \SetF{S}{0}{j})
	&= \sum_{j = 1}^\ell f(\OPT \cup \varnothing) 
		&\text{(by the initialization $\SetF{S}{0}{j} = \varnothing$)} \\
	&= \sum_{j = 1}^\ell f(\varnothing) + \sum_{j = 1}^\ell f(\OPT \mid \varnothing) 
		&\text{(rearranging terms)}\\
	&= (p + 1) \cdot \sum_{j = 1}^\ell f(\varnothing) + \sum_{j = 1}^\ell f(\OPT \mid \varnothing)
		&\mspace{-9mu}\text{($f(\varnothing) \geq 0$ by the non-negativity)} \\
	&\leq (p + 1) \cdot \sum_{j = 1}^\ell f(\SetF{S}{0}{j}) + \sum_{j = 1}^\ell f(\SetF{O}{0}{j} \mid \SetF{S}{0}{j}) \enspace.
		&\text{(by the initialization $\SetF{S}{0}{j} = \varnothing$)}
	\end{align*}
	}%
	
	Assume now that the lemma holds for all iterations between $0$ to $i - 1 $, and let us prove it for iteration $i$.
	Recall that only the solution $\SetF{S}{i}{j_i}$ is modified during iteration $i$.
	Thus, we have that the change in iteration $i$ in the first sum in the guarantee of the lemma is
	\begin{equation} \label{eq:S_difference}
	(p + 1) \cdot \sum_{j = 1}^\ell f(\SetF{S}{i}{j})
	-
	(p + 1) \cdot \sum_{j = 1}^\ell f(\SetF{S}{i - 1}{j})
	=
	(p + 1) \cdot f(u_i \mid \SetF{S}{i - 1}{j_i})
	\enspace.
	\end{equation}
	
	Bounding the change in the second sum in the guarantee is more involved, and is done in two steps. The first step is the following inequality.
	\begin{align} \label{eq:S_difference_condition_on}
	\sum_{j = 1}^\ell f(\SetF{O}{i-1}{j} \mid \SetF{S}{i - 1}{j}) 
	&- \sum_{j = 1}^\ell f(\SetF{O}{i-1}{j} \mid \SetF{S}{i}{j}) \\ \nonumber
	&= f(\SetF{O}{i-1}{j_i} \mid \SetF{S}{i - 1}{j_i}) - f(\SetF{O}{i-1}{j_i} \mid \SetF{S}{i}{j_i})
		&\text{(only $\SetF{S}{i}{j_i}$ is modified)}\\ \nonumber
	&= f(u_i \mid \SetF{S}{i-1}{j_i}) - f(u_i \mid \SetF{O}{i-1}{j_i} \cup \SetF{S}{i - 1}{j_i})
		&\text{(rearranging terms)} \\ \nonumber 
	&\leq f(u_i \mid \SetF{S}{i-1}{j_i}) - f(u_i \mid\OPT \cup \SetF{S}{i - 1}{j_i})
	\enspace,
	\end{align}
	where the inequality may be proved by considering two cases.
	First, suppose that $u_i \in \SetF{O}{i-1}{j_i} \cup \SetF{S}{i-1}{j_i}$.
	In this case, the inequality holds with equality, because $\SetF{O}{i-1}{j_i} \subseteq \OPT$ by assumption.
	Consider now the case in which $u_i \not \in \SetF{O}{i-1}{j_i} \cup \SetF{S}{i-1}{j_i}$. 
	In this case, our assumption that $(\SetF{S}{\numiter}{j_i} \setminus \SetF{S}{i-1}{j_i}) \cap \OPT \subseteq \SetF{O}{i-1}{j_i}$ implies $u_i \not \in (\SetF{S}{\numiter}{j_i} \setminus \SetF{S}{i-1}{j_i}) \cap \OPT$, which implies in its turn $u_i \not \in \OPT$ since $u_i \in \SetF{S}{i}{j_i} \subseteq \SetF{S}{\numiter}{j_i}$ and $u_i \in \gnd_{i - 1} \subseteq \gnd \setminus \SetF{S}{i-1}{j_i}$. 
	Therefore, we get that in this case that Inequality~\eqref{eq:S_difference_condition_on} holds due to the submodularity of $f$ (recall that $\SetF{O}{i-1}{j_i} \subseteq \OPT$ by our assumption).
	
For the second step in the proof of the above mentioned bound, we need to observe that, by the definition of the pair $(u_i, j_i)$, we have $$f(u_i \mid \SetF{S}{i - 1}{j_i}) \geq f(u \mid \SetF{S}{i - 1}{j}) \geq f(u \mid \SetF{S}{i}{j}),$$ for any element $u \in \gnd_{i-1}$ and integer $1 \leq j \leq \ell$ for which $\SetF{S}{i - 1}{j} + u$ is independent---the second inequality follows from submodularity when either $u \neq u_i$ or $j \neq j_i$ and from the non-negativity of $f(u_i \mid \SetF{S}{i - 1}{j_i})$ when $u = u_i$ and $j = j_i$. Since $\SetF{O}{i - 1}{j} \subseteq \gnd_{i - 1}$ and $\SetF{S}{i - 1}{j} + u$ is independent for every $u \in \SetF{O}{i-1}{j}$ by our assumption, the last inequality implies
{\allowdisplaybreaks
	\begin{align} \label{eq:O_difference} 
	\nonumber
	\sum_{j = 1}^\ell & f(\SetF{O}{i - 1}{j} \mid \SetF{S}{i}{j}) \\ \nonumber
	&\leq \sum_{j = 1}^\ell f(\SetF{O}{i}{j} \mid \SetF{S}{i}{j}) 
		+ \sum_{j = 1}^\ell \sum_{u \in \SetF{O}{i-1}{j} \setminus \SetF{O}{i}{j}} \mspace{-18mu} f(u \mid \SetF{S}{i}{j})
		&\text{(submodularity, $\SetF{O}{i}{j} \subseteq \SetF{O}{i - 1}{j}$)}\\ \nonumber
	&= \sum_{j = 1}^\ell f(\SetF{O}{i}{j} \mid \SetF{S}{i}{j}) + \sum_{j = 1}^\ell \sum_{u \in \SetF{O}{i-1}{j} \setminus (\SetF{O}{i}{j} \cup \SetF{U}{i}{j})} \mspace{-36mu} f(u \mid \SetF{S}{i}{j})
		&\text{($\SetF{U}{i}{j} \subseteq \SetF{S}{i}{j}$)} \\ \nonumber 
	&\leq \sum_{j = 1}^\ell f(\SetF{O}{i}{j} \mid \SetF{S}{i}{j}) + \sum_{j = 1}^\ell \sum_{u \in \SetF{O}{i-1}{j} \setminus (\SetF{O}{i}{j} \cup \SetF{U}{i}{j})} \mspace{-36mu} f(u_i \mid \SetF{S}{i - 1}{j_i})
		&\text{(greedy selection of $u_i$)}\\ \nonumber 
	&= \sum_{j = 1}^\ell f(\SetF{O}{i}{j} \mid \SetF{S}{i}{j}) + f(u_i \mid \SetF{S}{i - 1}{j_i}) \cdot \sum_{j = 1}^\ell |\SetF{O}{i-1}{j} \setminus (\SetF{O}{i}{j} \cup \SetF{U}{i}{j})|
		&\text{(rearranging terms)} \\
	&\leq \sum_{j = 1}^\ell f(\SetF{O}{i}{j} \mid \SetF{S}{i}{j}) + p \cdot f(u_i \mid \SetF{S}{i - 1}{j_i})
	\enspace,
	\end{align}
}%
	where the last inequality holds by our assumption that $\sum_{j = 1}^\ell |\SetF{O}{i - 1}{j} \setminus (\SetF{O}{i}{j} \cup \SetF{U}{i}{j})| \leq p$ and the non-negativity of $f(u_i \mid \SetF{S}{i - 1}{j_i})$.
	Combining Inequalities~\eqref{eq:S_difference}, \eqref{eq:S_difference_condition_on} and~\eqref{eq:O_difference}, we get
	\begin{align*}
	(p + 1) \cdot \sum_{j = 1}^\ell &f(\SetF{S}{i}{j}) + \sum_{j = 1}^\ell f(\SetF{O}{i}{j} \mid \SetF{S}{i}{j})\\
	\geq{} &
	\left[(p + 1) \cdot \sum_{j = 1}^\ell f(\SetF{S}{i - 1}{j}) + (p + 1) \cdot f(u_i \mid \SetF{S}{i - 1}{j_i})\right] + \left[\sum_{j = 1}^\ell f(\SetF{O}{i-1}{j} \mid \SetF{S}{i}{j}) - p \cdot f(u_i \mid \SetF{S}{i-1}{j_i})\right]\\
	={} &
	(p + 1) \cdot \sum_{j = 1}^\ell f(\SetF{S}{i - 1}{j})
	+
	\left[\sum_{j = 1}^\ell f(\SetF{O}{i - 1}{j} \mid \SetF{S}{i}{j}) + f(u_i \mid \SetF{S}{i-1}{j_i})\right]\\
	\geq{} &
	(p + 1) \cdot \sum_{j = 1}^\ell f(\SetF{S}{i - 1}{j})
	+
	\sum_{j = 1}^\ell f(\SetF{O}{i - 1}{j} \mid \SetF{S}{i - 1}{j})
	+ f(u_i \mid\OPT \cup \SetF{S}{i - 1}{j_i})\\
	\geq{} &
	\sum_{j = 1}^\ell f(\OPT \cup \SetF{S}{i - 1}{j}) + f(u_i \mid \OPT \cup \SetF{S}{i - 1}{j_i}) \\
	={}&
	\sum_{j = 1}^\ell f(\OPT \cup \SetF{S}{i}{j})
	\enspace,
	\end{align*}
	where the second inequality follows from submodularity and the last inequality follows from the induction hypothesis.
\end{proof}

The following corollary uses the last lemma to prove a lower bound on the sum of the objective values of the $\ell$ final solutions in terms of the optimal solution. 

\begin{corollary} \label{cor:approximation_to_union}
	Given the conditions of Proposition~\ref{prop:general_result},
	\[
	(p + 1) \cdot \sum_{j = 1}^\ell f(\SetF{S}{\numiter}{j})
	\geq
	\sum_{j = 1}^\ell f(\OPT \cup \SetF{S}{\numiter}{j})
	\enspace.
	\]
\end{corollary}
\begin{proof}
	The termination condition of \mainalg implies that $f(u \mid \SetF{S}{\numiter}{j}) \leq 0$ for every element $u \in \gnd_\numiter$ and integer $1 \leq j \leq \ell$ such that $\SetF{S}{\numiter}{j} + u$ is independent. Since $\SetF{O}{\numiter}{j} \subseteq \gnd_\numiter$ and $\SetF{S}{\numiter}{j} + u$ is independent for every $u \in \SetF{O}{\numiter}{j}$ by our assumption, this implies
	\[
	f(\SetF{O}{\numiter}{j} \mid \SetF{S}{\numiter}{j})
	\leq
	\sum_{u \in \SetF{O}{\numiter}{j}} f(u \mid \SetF{S}{\numiter}{j})
	\leq
	0
	\enspace,
	\]
	where the first inequality follows from the submodularity of $f$. Plugging this observation into the guarantee of Lemma~\ref{lem:greedy_invariant} for $i = \numiter$ yields
	\[
	\sum_{j = 1}^\ell f(\OPT \cup \SetF{S}{\numiter}{j})
	\leq
	(p + 1) \cdot \sum_{j = 1}^\ell f(\SetF{S}{\numiter}{j}) + \sum_{j = 1}^\ell f(\SetF{O}{\numiter}{j} \mid \SetF{S}{\numiter}{j})
	\leq
	(p + 1) \cdot \sum_{j = 1}^\ell f(\SetF{S}{\numiter}{j})
	\enspace.
	\qedhere
	\]
\end{proof}

To get an approximation ratio from the guarantee of the last corollary, we need to relate the sum $\sum_{j = 1}^\ell f(OPT \cup \SetF{S}{\numiter}{j})$ to $f(\OPT)$. We do this using the following known lemma.

\begin{lemma}[Lemma~2.2 of \cite{BFNS14}] \label{lem:distribution}
	Let $g\colon 2^\gnd \to \nnR$ be non-negative and submodular, and let $S$ a random subset of $\gnd$ in which each element appears with probability at most $p$ (not necessarily independently). Then, $\bE[g(S)] \geq (1 - p) \cdot g(\varnothing)$.
\end{lemma}

We are now ready to prove Proposition~\ref{prop:general_result}.
\begin{proof}[Proof of Proposition~\ref{prop:general_result}]
	Recall that the set $\retsol$ returned by \mainalg is the one having the largest objective value amongst all of the $\ell$ solutions.
	Thus, by a simple averaging argument together with Corollary~\ref{cor:approximation_to_union}, we obtain the following lower bound on its objective value,
	\begin{equation} \label{eq:first-obj-lowerbound}
	f(\retsol) 
	= \max_{j = 1 \dots \ell} f(\SetF{S}{\numiter}{j})
	\geq \frac{1}{\ell} \cdot \sum_{j = 1}^\ell f(\SetF{S}{\numiter}{\ell})
	\geq \frac{1}{p+1} \left[ \frac{1}{\ell} \cdot \sum_{j=1}^\ell f( \OPT \cup \SetF{S}{\numiter}{\ell}) \right]
	.
	\end{equation}
	Consider now a random set $\bar{S}$ chosen uniformly at random from the $\ell$ constructed solutions $\SetF{S}{\numiter}{1}, \SetF{S}{\numiter}{2}, \dotsc, \SetF{S}{\numiter}{\ell}$. 
	Since the solutions are disjoint by construction, an element can belong to $\bar{S}$ with probability at most $\ell^{-1}$. 
	Hence, by applying Lemma~\ref{lem:distribution} to the submodular function $g(S) = f(\OPT \cup S)$, we get
	\[
	\frac{1}{\ell} \cdot \sum_{j=1}^\ell f( \OPT \cup \SetF{S}{\numiter}{\ell})
	= \bE[ f(OPT \cup \bar{S}) ]
	= \bE[ g(\bar{S})]
	\geq (1 - \ell^{-1}) \cdot g( \varnothing)
	= (1 - \ell^{-1}) \cdot f(\OPT) \enspace.
	\]
	Together with Inequality~\eqref{eq:first-obj-lowerbound}, this shows that the returned solution $\retsol$ is a $(p+1) / (1- \ell^{-1})$-approximation, as desired.
	We remark also that if $f$ is monotone, then for each solution $1 \leq j \leq \ell$ we have that $f( \OPT \cup \SetF{S}{\numiter}{\ell}) \geq f(\OPT)$.
	Applying this directly to Inequality~\eqref{eq:first-obj-lowerbound} yields that the returned set $\retsol$ is a $(p+1)$-approximation when $f$ is monotone.
\end{proof}

\subsection{Analysis for \texorpdfstring{$k$}{k}-Extendible Systems} \label{sec:analysis-extendible}

In this section we use Proposition~\ref{prop:general_result} to prove Theorem~\ref{thm:deterministic_extendible}. 
Throughout this section we assume that $(\gnd, \cI)$ is a $k$-extendible system.
We demonstrate that for any number of solutions $\ell$, the conditions of Proposition~\ref{prop:general_result} hold with the value $p = \max(k, \ell - 1)$.
The proof of Theorem~\ref{thm:deterministic_extendible} follows by setting $\ell = k + 1$.

In order to show that the conditions of Proposition~\ref{prop:general_result} hold, we need to construct a set $\SetF{O}{i}{j}$ for every iteration $0 \leq i \leq \numiter$ and solution $1 \leq j \leq \ell$. Thus, we begin the section by explaining how to construct these sets.
The construction is done in a recursive way, and with the knowledge of the algorithm's execution path.
For $i = 0$, we define $\SetF{O}{0}{j} =\OPT$ for every $1 \leq j \leq \ell$, as is required by Proposition~\ref{prop:general_result}. 
Assume now that the sets $\SetF{O}{i - 1}{1}, \SetF{O}{i - 1}{2}, \dotsc, \SetF{O}{i - 1}{\ell}$ have already been constructed for some iteration $i > 0$, 
then we construct the sets $\SetF{O}{i}{1}, \SetF{O}{i}{2}, \dotsc, \SetF{O}{i}{\ell}$ as follows:
\begin{itemize}
	\item For every solution $1 \leq j \leq \ell$ other than $j_i$, $\SetF{O}{i}{j} = \SetF{O}{i-1}{j} - u_i$.
	\item If $u_i \in \SetF{O}{i-1}{j_i}$, then $\SetF{O}{i}{j_i} = \SetF{O}{i-1}{j_i} - u_i$, else $\SetF{O}{i}{j_i}$ is any maximal subset of $\SetF{O}{i-1}{j_i}$ such that $\SetF{O}{i}{j_i} \cup \SetF{S}{i}{j_i}$ is independent and $(\SetF{S}{\numiter}{j_i} \setminus \SetF{S}{i}{j_i}) \cap\OPT \subseteq \SetF{O}{i}{j_i}$. Notice that such a subset must exist because $[(\SetF{S}{\numiter}{j_i} \setminus \SetF{S}{i}{j_i}) \cap\OPT] \cup \SetF{S}{i}{j_i} \subseteq \SetF{S}{\numiter}{j_i}$ is an independent set and $(\SetF{S}{\numiter}{j_i} \setminus \SetF{S}{i}{j_i}) \cap\OPT \subseteq (\SetF{S}{\numiter}{j_i} \setminus \SetF{S}{i - 1}{j_i}) \cap\OPT \subseteq \SetF{O}{i-1}{j_i}$.
\end{itemize}

\begin{proposition} \label{prop:k-extendible-construction}
	If $(\cI, \gnd)$ is a $k$-extendible system, then the sets $\SetF{O}{i}{j}$ constructed above satisfy the conditions of Proposition~\ref{prop:general_result} with $p = \max(k, \ell - 1)$.
\end{proposition}

The next four lemmata together prove Proposition~\ref{prop:k-extendible-construction} by verifying each of the conditions in Proposition~\ref{prop:general_result}. 
\begin{lemma} \label{lem:extendible_condition_1}
	For every iteration $0 \leq i \leq \numiter$ and solution $1 \leq j \leq \ell$, $\SetF{O}{i}{j} \cup \SetF{S}{i}{j}$ is independent, and thus, $\SetF{S}{i}{j} + u$ is independent for every $u \in \SetF{O}{i}{j}$.
\end{lemma}
\begin{proof}
	We prove the lemma by induction on the iteration $i$. 
	For $i = 0$, the lemma holds since $$\SetF{O}{i}{j} \cup \SetF{S}{i}{j} =\OPT \cup \varnothing =\OPT.$$ 
	Assume now that the lemma holds for all iterations up to and including $i - 1 \geq 0$, and let us prove it for iteration $i$. 
	For solutions which were not updated at this iteration (that is, $j \neq j_i$), the lemma follows from the induction hypothesis since
	\[
	\SetF{O}{i}{j} \cup \SetF{S}{i}{j}
	=
	[\SetF{O}{i-1}{j} - u_i] \cup \SetF{S}{i - 1}{j}
	\subseteq
	\SetF{O}{i-1}{j} \cup \SetF{S}{i - 1}{j}
	\enspace.
	\]
	
	It remains to prove the lemma for the solution $j = j_i$ which was updated. 
	If $u_i \not \in \SetF{O}{i-1}{j_i}$, then $\SetF{O}{i}{j} \cup \SetF{S}{i}{j}$ is independent by the construction of $\SetF{O}{i}{j}$. 
	Otherwise, $\SetF{O}{i}{j} \cup \SetF{S}{i}{j}$ is independent by the induction hypothesis since
	\[
	\SetF{O}{i}{j} \cup \SetF{S}{i}{j}
	=
	[\SetF{O}{i-1}{j} - u_i] \cup [\SetF{S}{i - 1}{j} + u_i]
	=
	\SetF{O}{i-1}{j} \cup \SetF{S}{i - 1}{j}
	\enspace.
	\qedhere
	\]
\end{proof}

\begin{lemma} \label{lem:extendible_condition_2}
	For every iteration $1 \leq i \leq \numiter$ and solution $0 \leq j \leq \ell$, $\SetF{O}{i}{j} \subseteq \SetF{O}{i - 1}{j} \cap \gnd_i$. Moreover, for $i = 0$ we have $\SetF{O}{i}{j} \subseteq \gnd_i$ for every solution $0 \leq j \leq \ell$.
\end{lemma}
\begin{proof}
	We prove the lemma by induction on iterations $i$. 
	For $i = 0$, the lemma trivially holds since $\gnd_0 = \gnd$. 
	Assume now that the lemma holds for iterations up to and including $i - 1 \geq 0$, and let us prove it for iteration $i$. 
	By the construction of $\SetF{O}{i}{j}$, it is a subset of $\SetF{O}{i - 1}{j}$, an thus, to prove the lemma it suffices to show that $\SetF{O}{i}{j} \subseteq \gnd_i = \gnd_{i-1} - u_i$.
	
	The last inclusion follows from combining the next two observations:
	By the induction hypothesis, $\SetF{O}{i - 1}{j}$ is a subset of $\gnd_{i-1}$, and therefore, so must be $\SetF{O}{i}{j}$.
	If $u_i \not \in \SetF{O}{i - 1}{j}$, then $u_i$ cannot belong to $\SetF{O}{i}{j}$ because the last set is a subset of $\SetF{O}{i - 1}{j}$. Otherwise, we get by construction $\SetF{O}{i}{j} = \SetF{O}{i - 1}{j} - u_i$, which guarantees again that $u_i$ does not belong to $\SetF{O}{i}{j}$.
\end{proof}

\begin{lemma} \label{lem:extendible_condition_3}
	For every iteration $0 \leq i \leq \numiter$ and solution $1 \leq j \leq \ell$, $(\SetF{S}{\numiter}{j} \setminus \SetF{S}{i}{j}) \cap\OPT \subseteq \SetF{O}{i}{j}$.
\end{lemma}
\begin{proof}
	We prove the lemma by induction on the iterations $i$. 
	For $i = 0$, the lemma holds since
	\[
	(\SetF{S}{\numiter}{j} \setminus \SetF{S}{i}{j}) \cap \OPT
	\subseteq
	\OPT
	=
	\SetF{O}{i}{j}
	\enspace.
	\]
	Assume now that the lemma holds for all iterations up to and including $i - 1 \geq 0$, and let us prove it for iteration $i$. 
	There are two cases to consider. 
	If $\SetF{O}{i}{j} = \SetF{O}{i-1}{j} - u_i$, then by the induction hypothesis, since $u_i$ is the sole element of $\SetF{S}{i}{j}$ that does not appear in $\SetF{S}{i-1}{j}$ (if there is such an element at all),
	\[
	(\SetF{S}{\numiter}{j} \setminus \SetF{S}{i}{j}) \cap \OPT
	=
	(\SetF{S}{\numiter}{j} \setminus \SetF{S}{i - 1}{j}) \cap \OPT - u_i
	\subseteq
	\SetF{O}{i-1}{j} - u_i
	=
	\SetF{O}{i}{j}
	\enspace.
	\]
	
	It remains to consider the case in which $\SetF{O}{i}{j} \neq \SetF{O}{i-1}{j} - u_i$. 
	However, there is only one case in the construction of $\SetF{O}{i}{j}$ in which this might happen, and in this case $\SetF{O}{i}{j}$ is chosen as a set including $(\SetF{S}{\numiter}{j} \setminus \SetF{S}{i}{j}) \cap\OPT$, so there is nothing to prove.
\end{proof}

\begin{lemma} \label{lem:extendible_condition_4}
	For every iteration $1 \leq i \leq \numiter$, $\sum_{j = 1}^{\ell} |\SetF{O}{i - 1}{j} \setminus (\SetF{O}{i}{j} \cup \SetF{U}{i}{j})| \leq \max(k, \ell - 1)$.
\end{lemma}
\begin{proof}
	There are two cases to consider. 
	If $u_i \in\OPT$, then Lemma~\ref{lem:extendible_condition_3} guarantees that $\SetF{O}{i-1}{j_i}$ contains $u_i$, and thus, by construction, $\SetF{O}{i}{j} = \SetF{O}{i-1}{j} - u_i$ for every solution $1 \leq j \leq \ell$. 
	Thus,
	\[
	\sum_{i = 1}^{\ell} |\SetF{O}{i - 1}{j} \setminus (\SetF{O}{i}{j} \cup \SetF{U}{i}{j})|
	\leq
	\sum_{i = 1}^{\ell} |\{u_i\} \setminus \SetF{U}{i}{j}|
	=
	\ell - 1
	\enspace,
	\]
	where the equality holds since, by definition, $\SetF{U}{i}{j}$ is equal to $\{u_i\}$ for $j = j_i$ and to $\varnothing$ for every other $j$.
	
	Consider now the case of $u_i \not \in\OPT$. 
	In this case $u_i$ does not belong to $\SetF{O}{i-1}{j}$ for any $j$ because a repeated application of Lemma~\ref{lem:extendible_condition_2} can show that $\SetF{O}{i-1}{j}$ is a subset of $\OPT$. 
	Since $\SetF{O}{i}{j} = \SetF{O}{i-1}{j} - u_i$ for every solution $j \neq j_i$, we get for every such solution $j$,
	\[
	\SetF{O}{i - 1}{j} \setminus (\SetF{O}{i}{j} \cup \SetF{U}{i}{j})
	=
	\varnothing
	\enspace.
	\]
	To understand the set $\SetF{O}{i - 1}{j_i} \setminus (\SetF{O}{i}{j_i} \cup \SetF{U}{i}{j_i})$, we need to make a few observations. First, we recall that by Lemma~\ref{lem:extendible_condition_1}, $\SetF{O}{i-1}{j_i} \cup \SetF{S}{i-1}{j_i}$ is independent. Second, $$(\SetF{S}{\numiter}{j_i} \setminus \SetF{S}{i}{j_i}) \cap \OPT \subseteq \SetF{O}{i}{j_i} \subseteq \SetF{O}{i-1}{j_i}$$ by Lemmata~\ref{lem:extendible_condition_2} and~\ref{lem:extendible_condition_3}, and finally, $$\SetF{S}{i - 1}{j_i} \cup [(\SetF{S}{\numiter}{j_i} \setminus \SetF{S}{i}{j_i}) \cap\OPT] + u_i \subseteq \SetF{S}{\numiter}{j_i}$$ is also independent. Since $(\gnd, \cI)$ is $k$-extendible, these three observations imply together that there must exist a set $Y$ of size at most $k$ such that $(\SetF{O}{i-1}{j_i} \setminus Y) \cup \SetF{S}{i}{j_i}$ is independent, and $Y$ does not include elements of $(\SetF{S}{\numiter}{j_i} \setminus \SetF{S}{i}{j_i}) \cap \OPT$. 
	One can now observe that $\SetF{O}{i-1}{j_i} \setminus Y$ obeys all the conditions to be $\SetF{O}{i}{j_i}$ according to the construction of this set in the case of $u_i \not \in \SetF{O}{i-1}{j_i}$, and thus, since the construction selects a maximal set obeying these conditions as $\SetF{O}{i}{j_i}$, we get
	\[
	|\SetF{O}{i}{j_i}|
	\geq
	|\SetF{O}{i-1}{j_i} \setminus Y|
	\geq
	|\SetF{O}{i-1}{j_i}| - |Y|
	\geq
	|\SetF{O}{i-1}{j_i}| - k
	\enspace.
	\]
	Since $\SetF{O}{i}{j_i}$ is a subset of $\SetF{O}{i - 1}{j_i}$, this implies
	\[
	|\SetF{O}{i - 1}{j_i} \setminus (\SetF{O}{i}{j_i} \cup \SetF{U}{i}{j_i})|
	\leq
	|\SetF{O}{i - 1}{j_i} \setminus \SetF{O}{i}{j_i}|
	=
	|\SetF{O}{i - 1}{j_i}| - |\SetF{O}{i}{j_i}|
	\leq
	k
	\enspace.
	\]
	Combining everything that we have proved for the case of $u_i \not \in \OPT$, we get that in this case
	\[
	\sum_{i = 1}^{\ell} |\SetF{O}{i - 1}{j} \setminus (\SetF{O}{i}{j} \cup \SetF{U}{i}{j})|
	=
	(\ell - 1) \cdot |\varnothing| + |\SetF{O}{i - 1}{j_i} \setminus (\SetF{O}{i}{j_i} \cup \SetF{U}{i}{j_i})|
	\leq
	k
	\enspace.
	\]
	The two cases together yields that the sum in question is at most $p = \max(k, \ell - 1)$.
\end{proof}

We are now ready to prove Theorem~\ref{thm:deterministic_extendible}.
\begin{proof}[Proof of Theorem~\ref{thm:deterministic_extendible}]
	Lemmata~\ref{lem:extendible_condition_1}, \ref{lem:extendible_condition_2}, \ref{lem:extendible_condition_3} and~\ref{lem:extendible_condition_4} together prove Proposition~\ref{prop:k-extendible-construction}, which states that the sets we have constructed obey the conditions of Proposition~\ref{prop:general_result} for $p = \max(k, \ell - 1)$. 
	This implies that the approximation ratio of \mainalg for $k$-extendible systems is at most 
	\[
	\frac{p+1}{1 - \ell^{-1}}
	= \frac{\max(k, \ell - 1) + 1}{1 - \ell^{-1}}
	= \frac{\max( k +1, \ell)}{1 - \ell^{-1}}
	\]
	Choosing the number of solutions to be $\ell = k + 1$ optimizes this approximation factor and yields
	\[
	\frac{\max( k +1, \ell)}{1 - \ell^{-1}}
	=
	\frac{k + 1}{1 - (k + 1)^{-1}}
	=
	\frac{(k + 1)^2}{(k + 1) - 1}
	=
	\frac{(k + 1)^2}{k}
	\enspace.
	\]
	Now further suppose that $f$ is monotone in addition to being submodular and non-negative. 
	In this case, Proposition~\ref{prop:general_result} guarantees an approximation factor of at most 
	\[
	p + 1 = \max(k, \ell - 1) + 1 = \max(k+1, \ell)
	\enspace,
	\]
	which demonstrates that the approximation factor improves to $k+1$ for any number of solutions $\ell \leq k+1$.
\end{proof}

\subsection{Analysis for \texorpdfstring{$k$}{k}-Systems} \label{sec:analysis-k-system}

In this section we use Proposition~\ref{prop:general_result} to prove Theorem~\ref{thm:deterministic_system}. 
Throughout this section we assume that $(\gnd, \cI)$ is a $k$-system.
We demonstrate that for any number of solutions $\ell$, the conditions of Proposition~\ref{prop:general_result} hold with the value $p = k + \ell - 1$.
Then, to prove Theorem~\ref{thm:deterministic_system}, we choose $\ell = \lfloor 2 + \sqrt{k + 2} \rfloor$.

To use Proposition~\ref{prop:general_result}, we need to construct a set $\SetF{O}{i}{j}$ for every iteration $0 \leq i \leq \numiter$ and solution $1 \leq j \leq \ell$. 
As in Section~\ref{sec:analysis-extendible}, these sets are constructed recursively and with knowledge of the deterministic algorithm's execution path; however, for the case of $k$-systems, the construction of these sets starts at the final iteration and works backwards to the first iteration.
We begin by constructing related sets $\SetF{\tilde{O}}{i}{j}$ using the following recursive rule.
\begin{itemize}
	\item For the final iteration $i = \numiter$, $\SetF{\tilde{O}}{i}{j}$ contains all the elements of $\OPT \setminus \SetF{S}{i}{j}$ that can be added to $\SetF{S}{i}{j}$ without violating independence. In other words, $\SetF{\tilde{O}}{i}{j} = \{u \in\OPT \setminus \SetF{S}{i}{j} \mid \SetF{S}{i}{j} + u \in \cI\}$.
	\item For earlier iterations $i < \numiter$, if the solution $j$ is unaffected at this iteration (that is, $j \neq j_{i + 1}$) then we simply set $\SetF{\tilde{O}}{i}{j} = \SetF{\tilde{O}}{i + 1}{j}$. 
	Otherwise, let $\SetF{B}{i}{j_{i + 1}}$ be the set of elements of $\OPT \setminus (\SetF{S}{i + 1}{j_{i + 1}} \cup \SetF{\tilde{O}}{i+1}{j_{i + 1}})$ that can be added to $\SetF{S}{i}{j_{i + 1}}$ without violating independence. 
	In other words, $$\SetF{B}{i}{j_{i + 1}} = \{u \in \OPT \setminus (\SetF{S}{i + 1}{j_{i + 1}} \cup \SetF{\tilde{O}}{i + 1}{j_{i + 1}}) \mid \SetF{S}{i}{j_{i + 1}} + u \in \cI\}.$$ 
	We also denote by $\SetF{\tilde{B}}{i}{j_{i + 1}}$ an arbitrary subset of $\SetF{B}{i}{j_{i + 1}}$ of size $\min\{|\SetF{B}{i}{j_{i + 1}}|, k\}$. 
	Using this notation, we can now define $$\SetF{\tilde{O}}{i}{j_{i + 1}} = \SetF{\tilde{O}}{i + 1}{j_{i + 1}} \cup \SetF{\tilde{B}}{i}{j_{i + 1}} \cup (\OPT \cap \{u_{i + 1}\}).$$
\end{itemize}
Using the sets $\SetF{\tilde{O}}{i}{j}$ defined by the above recursive rule, we can now define the sets $\SetF{O}{i}{j}$ using the following formula. 
For every iteration $0 \leq i \leq r$ and solution $1 \leq j \leq \ell$, let $\SetF{O}{i}{j} = \SetF{\tilde{O}}{i}{j} \cap \gnd_i$.

\begin{proposition} \label{prop:k-system-construction}
	If $(\cI, \gnd)$ is a $k$-system, then the sets $\SetF{O}{i}{j}$ constructed above satisfy the conditions of Proposition~\ref{prop:general_result} with $p = k + \ell - 1$.
\end{proposition}

The following lemmata together prove Proposition~\ref{prop:k-system-construction} by verifying each of the conditions of Proposition~\ref{prop:general_result}.
Unlike the case in Section~\ref{sec:analysis-extendible}, here it is not clear from the construction that $\SetF{O}{0}{j} =\OPT$ for each of the solutions $1 \leq j \leq \ell$. 
The next lemma proves that this is indeed the case.
\begin{lemma}
	For every solution $1 \leq j \leq \ell$, $\SetF{\tilde{O}}{0}{j} =\OPT$, and thus, $\SetF{O}{0}{j} = \SetF{\tilde{O}}{0}{j} \cap \gnd_0 =\OPT$ because $\gnd_0 = \gnd$.
\end{lemma}
\begin{proof}
	By reverse induction over the iterations, we prove the stronger claim: that for every iteration $0 \leq i \leq \numiter$ and solution $1 \leq j \leq \ell$, 
	$$| \OPT \setminus (\SetF{\tilde{O}}{i}{j} \cup \SetF{S}{i}{j})| \leq k \cdot |\SetF{S}{i}{j}|.$$ 
	Notice that this claim indeed implies the lemma since $\SetF{\tilde{O}}{i}{j}$ contains only elements of $\OPT$ and $\SetF{S}{0}{j} = \varnothing$.
	
	We begin the proof by induction by showing that the claim holds for at the final iteration $i = \numiter$. 
	By the definition of $\SetF{\tilde{O}}{\numiter}{j}$, no element of $\OPT \setminus (\SetF{\tilde{O}}{\numiter}{j} \cup \SetF{S}{\numiter}{j})$ can be added to $\SetF{S}{\numiter}{j}$ without violating independence, and thus, $\SetF{S}{\numiter}{j}$ is a base of $(\OPT \setminus \SetF{\tilde{O}}{\numiter}{j}) \cup \SetF{S}{\numiter}{j}$. 
	In contrast, $\OPT \setminus (\SetF{\tilde{O}}{\numiter}{j} \cup \SetF{S}{\numiter}{j})$ is an independent subset of $(\OPT \setminus \SetF{\tilde{O}}{\numiter}{j}) \cup \SetF{S}{\numiter}{j}$ because it is also a subset of the independent set $\OPT$. 
	Thus, since $(\gnd, \cI)$ is a $k$-system,
	\[
	|OPT \setminus (\SetF{\tilde{O}}{\numiter}{j} \cup \SetF{S}{\numiter}{j})|
	\leq
	k \cdot |\SetF{S}{\numiter}{j}|
	\enspace,
	\]
	which is the claim that we wanted to prove.
	
	Assume now that the claim holds for all iterations $i+1, i+2, \dotsc, \numiter$, and let us prove it for iteration $i$. 
	There are three cases to consider.
	 If the solution $j$ was not updated during this iteration ($j \neq j_{i + 1}$), then $\SetF{\tilde{O}}{i}{j} = \SetF{\tilde{O}}{i + 1}{j}$ and $\SetF{S}{i}{j} = \SetF{S}{i + 1}{j}$, and therefore, by the induction hypothesis,
	\[
	| \OPT \setminus (\SetF{\tilde{O}}{i}{j} \cup \SetF{S}{i}{j})|
	=
	| \OPT \setminus (\SetF{\tilde{O}}{i + 1}{j} \cup \SetF{S}{i + 1}{j})|
	\leq
	k \cdot |\SetF{S}{i + 1}{j}|
	=
	k \cdot |\SetF{S}{i}{j}|
	\enspace.
	\]
	The second case is when $j = j_{i + 1}$ and $|\SetF{\tilde{B}}{i}{j}| = k$. 
	In this case,
	\begin{align*}
	| \OPT \setminus (\SetF{\tilde{O}}{i}{j} \cup \SetF{S}{i}{j})|
	={} &
	| \OPT \setminus (\SetF{\tilde{O}}{i + 1}{j} \cup \SetF{S}{i + 1}{j})| - |\SetF{\tilde{B}}{i}{j}|\\
	={} &
	| \OPT \setminus (\SetF{\tilde{O}}{i + 1}{j} \cup \SetF{S}{i + 1}{j})| - k
	\leq
	k \cdot |\SetF{S}{i + 1}{j}| - k
	=
	k \cdot |\SetF{S}{i}{j}|
	\enspace,
	\end{align*}
	where the inequality holds by the induction hypothesis, and the first equality holds since $\SetF{\tilde{O}}{i}{j} \setminus \SetF{\tilde{O}}{i + 1}{j} = \SetF{\tilde{B}}{i}{j} \cup (\OPT \cap \{u_{i + 1}\})$, the elements of $\SetF{\tilde{B}}{i}{j}$ belong to $\OPT \setminus \SetF{S}{i + 1}{j}$ and the element $u_{i + 1}$ does not belong to this set.
	
	The last case we need to consider is when $j = j_{i + 1}$ and $|\SetF{\tilde{B}}{i}{j}| < k$. 
	In this case $\SetF{\tilde{B}}{i}{j} = \SetF{B}{i}{j}$, which implies that no element of $\OPT \setminus (\SetF{\tilde{O}}{i}{j} \cup \SetF{S}{i}{j})$ can be added to $\SetF{S}{i}{j}$ without violating independence, and thus, $\SetF{S}{i}{j}$ is a base of $(OPT \setminus \SetF{\tilde{O}}{i}{j}) \cup \SetF{S}{i}{j}$. 
	This allows us to prove the claim in the same way in which this is done in the base case. Specifically, observe that  $\OPT \setminus (\SetF{\tilde{O}}{i}{j} \cup \SetF{S}{i}{j})$ is an independent subset of $(\OPT \setminus \SetF{\tilde{O}}{i}{j}) \cup \SetF{S}{i}{j}$ because it is also a subset of the independent set $\OPT$. 
	Thus, since $(\gnd, \cI)$ is a $k$-system,
	\[
	|OPT \setminus (\SetF{\tilde{O}}{i}{j} \cup \SetF{S}{i}{j})|
	\leq
	k \cdot |\SetF{S}{i}{j}|
	\enspace,
	\]
	which is the claim that we wanted to prove.
\end{proof}

We now proceed to proving the explicit conditions of Proposition~\ref{prop:general_result}.

\begin{lemma} \label{lem:system_condition_1}
	For every iteration $0 \leq i \leq \numiter$ and solution $1 \leq j \leq \ell$, $\SetF{S}{i}{j} + u$ is independent for every $u \in \SetF{O}{i}{j}$.
\end{lemma}
\begin{proof}
	We prove by a reverse induction the stronger claim that for every iteration $0 \leq i \leq \numiter$ and solution $1 \leq j \leq \ell$, the set $\SetF{S}{i}{j} + u$ is independent for every $u \in \SetF{\tilde{O}}{i}{j}$.
	Note that this claim implies the lemma because $\SetF{O}{i}{j}$ is a subset of  $\SetF{\tilde{O}}{i}{j}$.
	
	At the last iteration $i = \numiter$, the claim is an immediate consequence of the definition of $\SetF{\tilde{O}}{r}{j}$. 
	Assume now that the claim holds for iterations $i + 1 , i+2, \dots \numiter$, and let us prove it for iteration $i$. 
	If the solution $j$ was not updated at this iteration ($j \neq j_{i + 1}$), then $\SetF{S}{i}{j} = \SetF{S}{i + 1}{j}$ and $\SetF{\tilde{O}}{i}{j} = \SetF{\tilde{O}}{i + 1}{j}$, and so the claims follows immediately from the induction hypothesis. 
	Thus, it remains to consider only the case in which the solution is updated, i.e., $j = j_{i + 1}$.
	In this case, $$\SetF{\tilde{O}}{i}{j} = \SetF{\tilde{O}}{i + 1}{j} \cup \SetF{\tilde{B}}{i}{j} \cup (\OPT \cap \{u_{i + 1
	}\}).$$ 
For every $u \in \SetF{\tilde{O}}{i + 1}{j}$, we have $\SetF{S}{i}{j} + u$ by the induction hypothesis since $\SetF{S}{i}{j}$ is a subset of $\SetF{S}{i + 1}{j}$. 
For every $u \in \SetF{\tilde{B}}{i}{j}$, we have $\SetF{S}{i}{j} + u$ by the definition of $\SetF{B}{i}{j}$. Finally, for $u = u_{i + 1}$, we have $\SetF{S}{i}{j} + u = \SetF{S}{i + 1}{j} \in \cI$.
\end{proof}

\begin{lemma} \label{lem:system_condition_2}
	For every iteration $1 \leq i \leq \numiter$ and solution $1 \leq j \leq \ell$, $\SetF{O}{i}{j} \subseteq \SetF{O}{i - 1}{j} \cap \gnd_i$.
\end{lemma}
\begin{proof}
	We first observe that $\SetF{\tilde{O}}{i}{j} \subseteq \SetF{\tilde{O}}{i - 1}{j}$ by construction, and $\gnd_i = \gnd_{i - 1} - u_i \subseteq \gnd_{i - 1}$. Thus,
	\[
	\SetF{O}{i}{j}
	=
	\SetF{\tilde{O}}{i}{j} \cap \gnd_i
	\subseteq
	\SetF{\tilde{O}}{i - 1}{j} \cap \gnd_{i - 1}
	=
	\SetF{O}{i - 1}{j}
	\enspace.
	\qedhere
	\]
\end{proof}

\begin{lemma} \label{lem:system_condition_3}
	For every iteration $0 \leq i \leq \numiter$ and solution $1 \leq j \leq \ell$, $(\SetF{S}{\numiter}{j} \setminus \SetF{S}{i}{j}) \cap \OPT \subseteq \SetF{O}{i}{j}$.
\end{lemma}
\begin{proof}
	We prove the lemma by reverse induction on the iterations. 
	At the final iteration $i = \numiter$, the claim that we need to prove is trivial since $\SetF{S}{\numiter}{j} \setminus \SetF{S}{i}{j} = \varnothing$. 
	Assume now that the lemma holds for iterations $i + 1, i+2, \dots \numiter$, and let us prove it for iteration $i$. 
	If the solution set is not updated ($j \neq j_{i + 1}$), then $\SetF{S}{i}{j} = \SetF{S}{i + 1}{j}$, which implies
	\[
	(\SetF{S}{\numiter}{j} \setminus \SetF{S}{i}{j}) \cap \OPT
	=
	(\SetF{S}{\numiter}{j} \setminus \SetF{S}{i + 1}{j}) \cap \OPT
	\subseteq
	\SetF{O}{i + 1}{j}
	\subseteq
	\SetF{O}{i}{j}
	\enspace,
	\]
	where the first inclusion holds by the induction hypothesis, and second inclusion by Lemma~\ref{lem:system_condition_2}. 
	Thus, it remains to consider only the case in which $j = j_{i + 1}$.
	
	In this case
	\[
	[(\SetF{S}{\numiter}{j} \setminus \SetF{S}{i}{j}) \cap \OPT] \setminus [(\SetF{S}{\numiter}{j} \setminus \SetF{S}{i + 1}{j}) \cap \OPT]
	=
	\OPT \cap \{u_{i + 1}\}
	\subseteq
	\SetF{\tilde{O}}{i}{j} \cap \gnd_i
	=
	\SetF{O}{i}{j}
	\enspace,
	\]
	where the inclusion follows from the definition of $\SetF{\tilde{O}}{i}{j}$ and the fact that $u_i$ is chosen as an element from $\gnd_i$. 
	Using the induction hypothesis, we now get
	\[
	(\SetF{S}{\numiter}{j} \setminus \SetF{S}{i}{j}) \cap \OPT
	=
	[(\SetF{S}{\numiter}{j} \setminus \SetF{S}{i + 1}{j}) \cap \OPT] \cup [ \OPT \cap \{u_{i + 1}\}]
	\subseteq
	\SetF{O}{i + 1}{j} \cup \SetF{O}{i}{j}
	=
	\SetF{O}{i}{j}
	\enspace,
	\]
	where the final equality follows again from Lemma~\ref{lem:system_condition_2}.
\end{proof}

\begin{lemma} \label{lem:system_condition_4}
	For every iteration $1 \leq i \leq \numiter$, $\sum_{j = 1}^\ell |\SetF{O}{i - 1}{j} \setminus (\SetF{O}{i}{j} \cup \SetF{U}{i}{j})| \leq k + \ell - 1$.
\end{lemma}
\begin{proof}
	For every solution $1 \leq j \leq \ell$ other than $j_i$, we have by definition $\SetF{U}{i}{j} = \varnothing$ and
	\[
	\SetF{O}{i-1}{j}
	=
	\SetF{\tilde{O}}{i-1}{j} \cap \gnd_{i-1}
	=
	\SetF{\tilde{O}}{i}{j} \cap (\gnd_i + u_i)
	\subseteq
	\SetF{\tilde{O}}{i}{j} \cap \gnd_i + u_i
	=
	\SetF{O}{i}{j} + u_i
	\enspace.
	\]
	Therefore,
	\[
	|\SetF{O}{i-1}{j} \setminus (\SetF{O}{i}{j} \cup \SetF{U}{i}{j})|
	\leq
	1
	\enspace.
	\]
	
	Additionally,
	\[
	\SetF{\tilde{O}}{i - 1}{j_i}
	=
	\SetF{\tilde{O}}{i}{j_i} \cup \SetF{\tilde{B}}{i - 1}{j_i} \cup (\OPT \cap \{u_i\})
	=
	\SetF{\tilde{O}}{i}{j_i} \cup \SetF{\tilde{B}}{i - 1}{j_i} \cup (\OPT \cap \SetF{U}{i}{j})
	\subseteq
	\SetF{\tilde{O}}{i}{j_i} \cup \SetF{\tilde{B}}{i - 1}{j_i} \cup \SetF{U}{i}{j}
	\enspace,
	\]
	and
	\[
	\gnd_{i-1}
	=
	\gnd_i + u_i
	=
	\gnd_i \cup \SetF{U}{i}{j}
	\enspace.
	\]
	These two observations imply together
	\begin{align*}
	\SetF{O}{i - 1}{j_i}
	={} &
	\SetF{\tilde{O}}{i - 1}{j_i} \cap \gnd_{i - 1}\\
	\subseteq{}
	&
	[\SetF{\tilde{O}}{i}{j_i} \cup \SetF{\tilde{B}}{i - 1}{j_i} \cup \SetF{U}{i}{j}] \cap [\gnd_i \cup \SetF{U}{i}{j}]
	\\
	\subseteq{} &
	[\SetF{\tilde{O}}{i}{j_i} \cap \gnd_i] \cup  \SetF{\tilde{B}}{i - 1}{j_i} \cup \SetF{U}{i}{j}\\
	={}&
	\SetF{O}{i}{j_i} \cup \SetF{\tilde{B}}{i - 1}{j_i} \cup \SetF{U}{i}{j}
	\enspace,
	\end{align*}
	and therefore, also
	\[
	|\SetF{O}{i - 1}{j_i} \setminus (\SetF{O}{i}{j_i} \cup \SetF{U}{i}{j_i})|
	\leq
	|\SetF{\tilde{B}}{i - 1}{j_i}|
	\leq
	k
	\enspace,
	\]
	where the last inequality follows from the definition of $\SetF{\tilde{B}}{i - 1}{j_i}$.
	
	Combining all the above results, we get
	\[
	\sum_{j = 1}^\ell |\SetF{O}{i - 1}{j} \setminus (\SetF{O}{i}{j} \cup \SetF{U}{i}{j})|
	\leq
	(\ell - 1) \cdot 1 + k
	=
	k + \ell - 1
	\enspace.
	\qedhere
	\]
\end{proof}

We are now ready to prove Theorem~\ref{thm:deterministic_system}.
\begin{proof}[Proof of Theorem~\ref{thm:deterministic_system}]
	Lemmata~\ref{lem:system_condition_1}, \ref{lem:system_condition_2}, \ref{lem:system_condition_3} and~\ref{lem:system_condition_4} prove together Proposition~\ref{prop:k-system-construction}, which states that the sets we have constructed obey the conditions of Proposition~\ref{prop:general_result} with $p = k + \ell - 1$.
	Thus, the last proposition implies that the approximation ratio of \mainalg for $k$-systems and $\ell = \lfloor 2 + \sqrt{k + 2} \rfloor$ is at most
	\[
	\frac{p + 1}{1 - \ell^{-1}}
	=
	\frac{k + \ell }{1 - \ell^{-1}}
	=
	\frac{k + \lfloor 2 + \sqrt{k + 2} \rfloor}{1 - 1/\lfloor 2 + \sqrt{k + 2} \rfloor}
	\leq
	\frac{k + 2 + \sqrt{k + 2}}{1 - 1/( 1 + \sqrt{k + 2})}
	\enspace.
	\]
	To simplify some calculations, let $\alpha = k + 2$.
	By substituting $\alpha$, rearranging terms, and re-substituting $\alpha$ we obtain that the right hand side of the last inequality may be expressed as
	\[
	\frac{\alpha + \sqrt{\alpha}}{1 - 1/( 1 + \sqrt{\alpha})}
	=
	\frac{(1 + \sqrt{\alpha}) \cdot (\alpha +  \sqrt{\alpha})}{\sqrt{\alpha} }
	= 
	(1 + \sqrt{\alpha})(1 + \sqrt{\alpha}) 
	= (1 + \sqrt{\alpha})^2
	= (1 + \sqrt{k+2})^2
	\enspace .
	\]
	Thus, the approximation ratio is at most $(1 + \sqrt{k+2})^2$.
	Suppose that $f$ is monotone so that Proposition~\ref{prop:general_result} guarantees the returned solution is a $(p+1)$-approximation. 
	Setting the number of solutions to $\ell=1$ yields $p = k + \ell - 1 = k$, so that the returned set is a $(k+1)$-approximation. 
	This demonstrates that our unified analysis recovers the guarantees of the greedy algorithm for monotone submodular objectives under a $k$-system constraint.
\end{proof}

\section{A Nearly Linear Time Implementation} \label{sec:fast-alg}
In this section, we present \fastalg, a nearly linear-time variant of \mainalg.
Recall that \mainalg greedily constructs $\ell$ candidate solutions in a simultaneous fashion. 
Because the algorithm uses an exact greedy search for the feasible element-solution pair with the largest marginal gain, the overall runtime is $\bigO{\ell^2 r n}$. 
Although we consider $\ell$ to be a constant (as it scales with $k$), the size of the largest base $r$ could be as large as $\bigO{n}$.
This means that, like other exact greedy approaches, \mainalg has a quadratic runtime. 
In this section, we show that \mainalg may be modified to run in nearly linear time by using the thresholding technique of \citet{Badanidiyuru14} for faster approximate greedy search.

The key idea of \fastalg is to replace the exact greedy search with an approximate greedy search via the use of a marginal gain acceptance threshold: if an element-solution pair is feasible and has a marginal gain which exceeds the threshold, then the update is made without considering other possible pairs.
By appropriately initializing and iteratively lowering this marginal gain threshold, we can ensure that the algorithm runs much quicker at the cost of only a small loss in the approximation.
The allowed loss in approximation is given as an input parameter $\eps \in (0, 1/2)$ to the algorithm.
A formal description of \fastalg appears as Algorithm~\ref{alg:fast-alg}.
It begins by initializing the $\ell$ solutions $\SetF{S}{0}{1}, \SetF{S}{0}{2}, \dotsc, \SetF{S}{0}{\ell}$ to be empty sets;
and the acceptance threshold, denoted by $\threshold$, is initially set to be the largest objective value of any element.
During each iteration of the while loop, the algorithm iterates once through the set of feasible element-solution pairs.
If a feasible element-solution pair is found whose gain exceeds the threshold, then the element is added to that solution.
After the completion of each iteration through all the feasible element-solution pairs, the acceptance threshold is reduced by a multiplicative factor of $1 -  \eps$, and the algorithm terminates when this threshold becomes sufficiently low.

\begin{algorithm}[H]
	\caption{\fastalg($\gnd, f, \cI, \ell,  \eps$)}\label{alg:fast-alg}
	Initialize $\ell$ solutions, $\SetF{S}{0}{j} \gets \varnothing$ for every $j =1, \dotsc, \ell$. \\
	Initialize ground set $\gnd_0 \gets \gnd$, and iteration counter $i \gets 1$.\\
	Let $\maxgain = \max_{u \in \gnd} f(u)$, and initialize threshold $\threshold = \maxgain$. \\
	\While{$\threshold > (  \eps / n)  \cdot \maxgain $}
	{
		\For{every element-solution pairs $(u,j)$ with $u \in \gnd_{i-1}$ and $1 \leq j \leq \ell$}{
			\If{$\SetF{S}{i-1}{j} + u \in \cI$ and $f(u \mid \SetF{S}{i-1}{j}) \geq \threshold$}{
				Let $u_i \gets u$ and $j_i \gets j$.\\
				Update the solutions as
				$ 
				\SetF{S}{i}{j} \gets
				\left\{
				\begin{array}{lr}
				\SetF{S}{i-1}{j_i} + u_i  &\text{if } j = j_i \enspace,\\
				\SetF{S}{i-1}{j} &\text{if } j \neq j_i \enspace.
				\end{array}
				\right.$ \\
				Update the available ground set $\gnd_{i} \gets \gnd_{i-1} - u_i$. \\
				Update the iteration counter $i \gets i+1$.
				}	
		}
		Update the marginal gain $\threshold \gets (1 -  \eps) \cdot \threshold$.
	}
	\Return{the set $\retsol$ maximizing $f$ among the sets $\{\SetF{S}{i}{j}\}_{j = 1}^\ell$}.
\end{algorithm}

We note that the iteration counter $i$ of \fastalg is used to index the state of the solutions, and does not necessarily correspond to the iterations of any specific loop.
We begin our analysis of \fastalg by proving that the number of oracle queries it uses is nearly linear in the number of elements in the ground set.

\begin{observation} \label{obs:fast-alg-runtime}
	\fastalg requires at most $\bigtO{\ell n / \eps}$ calls to the value and independence oracles.
\end{observation}
\begin{proof}
	In every iteration of the while loop, each element-solution pair is considered once, requiring one value query and one independence query.
	Thus, $\bigO{\ell n}$ oracle queries are made at each iteration of the while loop.
	Next, we seek to bound the number of iterations of the while loop.
	Note that the threshold is initially set to $\threshold = \maxgain$, and is decreased by a multiplicative factor of $1 -  \eps$ at each iteration of the while loop.
	Since the while loop ends once the threshold is below $( \eps / n) \cdot \maxgain$, 
	the number of iteration of the while loop is the smallest integer $a$ such that
	$(1- \eps)^a \cdot \maxgain \leq (  \eps  / n) \cdot \maxgain$.
	Dividing by $\maxgain$ and taking the $\log_{1 - \eps}$ of both sides, we get that $a$ is the smallest integer such that
	\[
	a \geq \log_{1 - \eps}( \eps / n) 
	= \frac{\log(\eps / n)}{\log(1 - \eps)} 
	= \frac{\log(n / \eps)}{-\log(1 - \eps)}
	\geq \frac{1 - \eps}{\eps} \log(n / \eps)
	\geq \frac{1}{2\eps} \log(n / \eps) \enspace,
	\]
	where the penultimate inequality uses $-\eps / (1 - \eps) \leq \log(1 - \eps) < 0$, which holds for $\eps \in (0, 1)$, and the last inequality follows from our assumption that $\eps < 1/2$.
	Thus, the number of iterations of the while loops is $\bigO{ 1 / \eps \cdot \log( n / \eps )}$ so that the total number of oracle queries is  $\bigO{\ell n / \eps \cdot \log \left( n /  \eps \right) }$.
	Using the $\bigtOsym$ notation that suppresses log factors, the total number of oracle queries becomes $\bigtO{\ell n / \eps}$.
\end{proof}

Now, we turn to the approximation guarantees of \fastalg. 
The two theorems below show that \fastalg achieves the same approximation guarantees as \mainalg, but with a multiplicative increase that depends on the error term $\eps$.
In this sense, the parameter $\eps$ controls the trade-off between the computational cost of the oracle queries and the approximation guarantee.

\begin{theorem} \label{thm:fast-deterministic_extendible}
	Suppose that $(\gnd, \cI)$ is a $k$-extendible system and that the number of solutions is set to $\ell = k + 1$.
	Then, \fastalg requires $\bigtO{k^2 n / \eps}$ oracle calls and produces a solution whose approximation ratio is at most $(1 -  2 \eps)^{-2} \cdot (k + 1)^2/k  $. 
	Moreover, when $f$ is non-negative monotone submodular and the number of solutions is chosen so $\ell \leq k+1$, then the approximation ratio improves to $(1 -  \eps)^{-2} \cdot (k+1)$.
\end{theorem}

\begin{theorem} \label{thm:fast-deterministic_system}
	Suppose that $(\gnd, \cI)$ is $k$-system and that the number of solutions is set to $\ell = \lfloor 2 + \sqrt{k + 2} \rfloor$.
	Then, \fastalg requires $\bigtO{k n / \eps}$ oracle calls and produces a solution whose approximation ratio is at most $(1 -  2 \eps)^{-2} \cdot (1 + \sqrt{k +2})^2$. 
	Moreover, when $f$ is non-negative monotone submodular and the number of solutions is set to $\ell=1$, then the approximation ratio improves to $(1 -  \eps)^{-2} \cdot (k+1)$.
\end{theorem}

The proofs of Theorems~\ref{thm:fast-deterministic_extendible} and \ref{thm:fast-deterministic_system} are very similar to their counterparts in Section~\ref{sec:main-alg}.
In particular, the same style of unified meta-proof may be used for analyzing \fastalg.
There are, however, two key differences in the analysis when we use the thresholding technique rather than an exact greedy search.
The first difference is that rather than a feasible element-solution pair whose marginal gain is maximal, we choose in each iteration a feasible element-solution pair whose marginal gain is within a $(1 - \eps)$ multiplicative factor of the largest marginal gain.
This $(1 - \eps)$ factor carries throughout the analysis.
The second difference is that, at the end of the algorithm, there may be elements which are feasible to add to solutions and have positive marginal gain; however, by the termination conditions, the marginal gain of each of these elements is at most $ ( \eps / n) \cdot \maxgain$.
Using submodularity, we can ensure that leaving these elements behind does not incur a significant loss in the objective value.
Formally, one can prove Theorems~\ref{thm:fast-deterministic_extendible} and \ref{thm:fast-deterministic_system} by observing that they follow from Proposition~\ref{prop:knapsack_general_result} (that appears in the next section) by plugging in $m = \rho = 0$ in the same way that Theorems~\ref{thm:knapsack-deterministic_extendible} and \ref{thm:knapsack-deterministic_system} follow from Proposition~\ref{prop:density-search-alg-result}.

\section{Incorporating Knapsack Constraints} \label{sec:knapsack-alg}

In this section, we consider the general form of Problem~\eqref{eq:opt-problem}, where the constraint is the intersection of an independence system $\cI$ with $m$ knapsack constraints.
We present $\knapsackalg$, an algorithm which extends the simultaneous greedy technique (Section~\ref{sec:main-alg}) and the faster thresholding variant (Section~\ref{sec:fast-alg}) to handle knapsack constraints by incorporating a density threshold technique.
The density threshold technique we consider was first introduced by \cite{MBK16} in the context of a repeated-greedy style algorithm for maximizing a submodular function over the intersection of a $k$-system and $m$ knapsack constraints.
By incorporating this density threshold technique into the \mainalg framework, we obtain a nearly linear time algorithm which improves both the approximation guarantees and the runtimes of previous methods.

The main idea behind the density threshold technique is to consider adding an element $u$ to a solution $S$ only if the marginal gain is larger than a fixed multiple $\density$ of the sum of its knapsack weights, i.e., $f(u \mid S) \geq \density \cdot \sum_{r=1}^m c_r(u)$.
Here, the quantity $f(u \mid S) / \sum_{r=1}^m c_r(u)$ is referred to as the \emph{density} of an element $u$ with respect to a set $S$. 
The density threshold technique received its name because it only adds an element to a solution if the density of the element is larger than the threshold $\density$.
For convenience, given a set $S$, an element $u$ is said to have high density if its density is larger than (or equal to) $\rho$ and low density if its density is less than $\rho$.

The algorithm \knapsackalg is presented below as Algorithm~\ref{alg:knapsack-alg}.
As before, the algorithm begins by initializing $\ell$ solutions $\SetF{S}{0}{1}, \SetF{S}{0}{2}, \dotsc, \SetF{S}{0}{\ell}$ to be empty sets.
Furthermore, a fast approximate greedy search is again achieved by using a marginal threshold $\threshold$ which is initially set to $\maxgain$ and then iteratively decreased by a multiplicative factor of $(1 - \epsilon)$.
The key difference here, compared to $\fastalg$, is that in order to add an element $u$ to a set $\SetF{S}{i-1}{j}$, we additionally require that the density ratio $f(u \mid S) / \sum_{r=1}^m c_r(u)$  is larger than a fixed threshold $\density$ and also that the updated set that we are considering, $\SetF{S}{i-1}{j} + u$, satisfies all the knapsack constraints.
For the purposes of analysis, we break these two conditions into separate lines, where the knapsack feasibility condition is checked on its own in Line~\ref{line:knapsack-check}.

\begin{algorithm}[H]
	\caption{\knapsackalg($\gnd, f, \cI, \ell, \density, \varepsilon$)}\label{alg:knapsack-alg}
	Initialize $\ell$ solutions, $\SetF{S}{0}{j} \gets \varnothing$ for every $j =1, \dotsc, \ell$. \\
	Initialize ground set $\gnd_0 \gets \gnd$, and iteration counter $i \gets 1$.\\
	Let $\maxgain = \max_{u \in \gnd} f(u)$, and initialize threshold $\threshold = \maxgain$. \\
	\While{$\threshold > (  \varepsilon / n)  \cdot \maxgain $}
	{
		\For{every element-solution pairs $(u,j)$ with $u \in \gnd_{i-1}$ and $1 \leq j \leq \ell$}{
			\If{$\SetF{S}{i-1}{j} + u \in \cI$ and $f(u \mid \SetF{S}{i-1}{j}) \geq \max \left( \threshold, \density \cdot \sum_{r=1}^m c_r(u) \right)$ \label{line:feasibility-check}}{
				\If{$c_r(\SetF{S}{i-1}{j} + u) \leq 1$ for all $1 \leq r \leq m$}{ \label{line:knapsack-check}
					Let $u_i \gets u$ and $j_i \gets j$.\\
					Update the solutions as
					$ 
					\SetF{S}{i}{j} \gets
					\left\{
					\begin{array}{lr}
					\SetF{S}{i-1}{j_i} + u_i  &\text{if } j = j_i \enspace,\\
					\SetF{S}{i-1}{j} &\text{if } j \neq j_i \enspace.
					\end{array}
					\right.$ \\
					Update the available ground set $\gnd_{i} \gets \gnd_{i-1} - u_i$. \\
					Update the iteration counter $i \gets i+1$.
				}
			}	
		}
		Update marginal gain $\threshold \gets (1 -  \varepsilon) \cdot \threshold$.
	}
	\Return{the set $\retsol$ maximizing $f$ among the sets $\{\SetF{S}{i}{j}\}_{j = 1}^\ell$} and the singletons $\{u\}_{u \in \cN}$.
\end{algorithm}

We begin the study of \knapsackalg by analyzing its running time.
In addition to analyzing the number of calls made to the value and independence oracles, we also analyze the number of arithmetic operations required by \knapsackalg that arise when working with the knapsack constraints.
In many practical scenarios, however, the computational burden of even a few calls to the value oracle is much greater than the total cost of all arithmetic operations required by the knapsack constraints; and thus, the bound on the number of such operations is of less significance.

\begin{observation} \label{obs:knapsack-alg-runtime}
	\knapsackalg requires at most $\bigtO{\ell n / \eps }$ calls to the value and independence oracles and $\bigtO{m \ell n / \eps}$ arithmetic operations.
\end{observation}
\begin{proof}
As shown in Observation~\ref{obs:fast-alg-runtime}, there are $\bigtO{1/\eps}$ iterations of the while loop.
At each iteration of the while loop, each of the $\bigO{\ell n}$ element-solution pairs are considered and checking feasibility of each pair requires a single call to the value and independence oracles.
Thus, the number of oracle calls is $\bigtO{\ell n / \eps}$.

The arithmetic operations are required for handling the knapsack constraints.
Note that for each element $u \in \gnd$, the term $\density \cdot \sum_{r=1}^m c_r(u)$ can be computed at the beginning of the algorithm using $m$ additions and $1$ multiplication. 
Thus, each of the $n$ terms may be computed using $\bigO{mn}$ arithmetic operations.
If each of the knapsack values $c_r( S^{(i)})$ are maintained for each of the $m$ knapsacks and $\ell$ solutions, then checking the condition in Line~\ref{line:knapsack-check} requires $\bigO{m}$ arithmetic operations.
Since this condition is checked for possibly every element-solution pair, this means $\bigO{m \ell n}$ arithmetic operations per iteration of the while loop.
Moreover, since the number of iterations of the while loop is $\bigtO{1 / \eps}$, we have that a total number of $\bigtO{m \ell n / \eps }$ arithmetic operations is required.
\end{proof}

Next, we present a unified analysis of \knapsackalg which yields approximation guarantees for $k$-systems and $k$-extendible systems.
At a high level, the analysis is similar to that of Proposition~\ref{prop:general_result}.
That is, we analyze the elements of the optimal solution $\OPT$ which must be thrown away as each new element is added.
The key difference here is that we must factor into our analysis the knapsack constraints---which arise via the density threshold criteria in Line~\ref{line:feasibility-check} and the new feasibility condition in Line~\ref{line:knapsack-check}.
It will be beneficial to break up our analysis into two cases based on whether \knapsackalg returns false on any instance of Line~\ref{line:knapsack-check} in its execution.
Towards this goal, let us introduce some new notation.
Let $\knapevent$ be an indicator variable for the event that the knapsack check in Line~\ref{line:knapsack-check} evaluates to false at any point in the algorithm.
That is, if $\knapevent = 0$ then the ``if statement'' in Line~\ref{line:knapsack-check} always evaluates to true; otherwise, $\knapevent = 1$ means that it returned false at some point.

\newcommand{\eqType}{}
\newcommand{\knapsackgeneralresult}[1][]{
	Suppose that there exists sets $\SetF{O}{i}{j}$ for every iteration $0 \leq i \leq \numiter$ and solution $1 \leq j \leq \ell$ and a value $p$ which satisfy the following properties:
	\begin{compactitem}
		\item $\SetF{O}{0}{j} = \OPT$ for every solution $1 \leq j \leq \ell$.
		\item $\SetF{S}{i}{j} + u \in \cI$ for every iteration $0 \leq i \leq \numiter$, solution $1 \leq j \leq \ell$, and element $u \in \SetF{O}{i}{j}$.
		\item $\SetF{O}{i}{j} \subseteq \SetF{O}{i - 1}{j} \cap \gnd_i$ for every iteration $1 \leq i \leq \numiter$ and solution $1 \leq j \leq \ell$.
		\item $(\SetF{S}{\numiter}{j} \setminus \SetF{S}{i}{j}) \cap \OPT \subseteq \SetF{O}{i}{j}$ for every iteration $0 \leq i \leq \numiter$ and solution $1 \leq j \leq \ell$.
		\item $\sum_{i = 1}^\ell |\SetF{O}{i - 1}{j} \setminus (\SetF{O}{i}{j} \cup \SetF{U}{i}{j})| \leq p$ for every iteration $1 \leq i \leq \numiter$.
	\end{compactitem}
	Then, the solution $\retsol$ produced by \knapsackalg satisfies the following approximation guarantees: 
	\ifthenelse{\isempty{#1}}{\renewcommand{\eqType}{equation}}{\renewcommand{\eqType}{equation*}}
	\begin{\eqType} \ifthenelse{\isempty{#1}}{\label{eq:knapsack-case-analysis}}{\tag{\ref{eq:knapsack-case-analysis}}}
	f(S) \geq 
	\left\{
	\begin{array}{lr}
	\frac{1}{2} \density  &\text{ if } \knapevent = 1 \enspace,\\
	\frac{1 - \epsilon}{p+1} \cdot \Big( \left( 1 - \ell^{-1} - \eps \right) f(\OPT) - m \density \Big) &\text{ if } \knapevent = 0 \enspace.
	\end{array}
	\right.
	\end{\eqType}
	Moreover, when $f$ is monotone, these approximation guarantees improve to 
	\begin{\eqType} \ifthenelse{\isempty{#1}}{\label{eq:knapsack-case-analysis-monotone}}{\tag{\ref{eq:knapsack-case-analysis-monotone}}}
	f(S) \geq 
	\left\{
	\begin{array}{lr}
	\frac{1}{2} \density  &\text{ if } \knapevent = 1 \enspace,\\
	\frac{1 - \epsilon}{p+1} \cdot \Big( \left( 1 - \eps \right) f(\OPT) - m \density \Big) &\text{ if } \knapevent = 0 \enspace.
	\end{array}
	\right.
	\end{\eqType}
}
\begin{proposition} \label{prop:knapsack_general_result}
	\knapsackgeneralresult
\end{proposition}

Proposition~\ref{prop:knapsack_general_result} provides a guarantee on the solution produced by \knapsackalg, which depends on the input density threshold $\density$ and also on the question whether Line~\ref{line:knapsack-check} ever evaluates to false during the algorithm's execution.
The conditions of Proposition~\ref{prop:knapsack_general_result} are identical to those in Proposition~\ref{prop:general_result} in Section~\ref{sec:main-alg}, and hence, the previous constructions of these sets $\SetF{O}{i}{j}$ for $k$-systems and $k$-extendible systems can be used here as well.
Note that when $\density = 0$, then every element has high density and we are back in the setting of \mainalg.

The analysis is very similar in spirit to Proposition~\ref{prop:general_result}, except that it features the marginal gain threshold technique for faster approximate greedy search and the density ratio threshold technique for knapsack constraints.
Because the main proof ideas involving the simultaneous greedy technique are presented in Section~\ref{sec:main-alg} and the marginal gain threshold and density ratio threshold techniques already appear in existing works, we defer the proof of Proposition~\ref{prop:knapsack_general_result} to Appendix~\ref{sec:knapsack-proof}.

Now we address the remaining question, which is how to choose a density threshold $\density$ which yields a good approximation.
Note that we always have either $\knapevent = 0$ or $\knapevent=1$, and hence, by taking the minimum of the two lower bounds for the two cases, we obtain the approximation guarantee
\begin{equation} \label{eq:lower_bound}
f(S) \geq \min \left\{ \frac{1}{2} \density, (1 - \eps) \left( \frac{1 - \ell^{-1} - \eps}{p+1} \right) f(\OPT) - \left( \frac{m}{p+1} \right) \density \right\}.
\end{equation}
To maximize this lower bound, we would like to set the density ratio threshold to
$$
\density^* 
= 2(1 - \eps) \left( \frac{1 - \ell^{-1} - \eps}{p+1 + 2m} \right) f(\OPT)
\enspace ,
$$
which would yield an approximation guarantee of
\[
f(S)
\geq (1 - \eps) \left( \frac{1 - \ell^{-1} - \eps}{p+1 + 2m} \right) f(\OPT)
 \enspace.
\]
Unfortunately, we cannot efficiently calculate this optimal choice for density ratio threshold $\density^*$ because it involves $f(\OPT)$.
Nevertheless, by the submodularity and non-negativity of the objective, we know that the optimal value lies within the range 
\[
\maxgain \leq f(\OPT) \leq r \cdot \maxgain \enspace, \]
where we recall that $r$ is the size of the largest independent set.
This interval, which is guaranteed to contain $f(\OPT)$, can be transformed into an interval containing the optimal density threshold $\density^*$.
In particular, we get
\[
\density^* 
= 2(1 - \eps) \left( \frac{1 - \ell^{-1} - \eps}{p+1 + 2m} \right) \maxgain \cdot \alpha
\enspace ,
\]
for some $\alpha \in [1, r]$.
When $r$ is not known exactly, an upper bound may be used here.
One upper bound we can use is that for any base $B \in \cI$,  $r$ is at most $k \cdot |B|$, which follows by the definition of a $k$-system.
Such a base $B$ may be known beforehand or constructed in a greedy fashion using $\bigO{n}$ calls to the independence oracle.
However, for simplicity, we use the somewhat weaker upper bound of $r \leq n$. Using the above mentioned stronger upper bound, or an instance specific upper bound, will reduce the interval in question, and thus, also the runtime. However, the improvement will only be in the logarithmic component of the runtime.

The high level idea is to design an algorithm which calls \knapsackalg several times as a subroutine using various density ratio thresholds in this range and to return the best solution.
\cite{MBK16} propose using a multiplicative grid search over this interval, running the algorithm on each point in the interval.
The multiplicative grid search guarantees that the subroutine algorithm is run with an input density threshold $\density$ which is close to the optimal $\density^*$ in the sense that $(1 - \delta) \density^* \leq \density \leq \density^*$ (for some error parameter $\delta \in (0, 1/2)$).
One may verify that by using this ``approximately-optimal'' density threshold, the approximation ratio obtained by the lower bound~\eqref{eq:lower_bound} is at most a factor $(1 - \delta)^{-1}$ larger than if the optimal threshold $\density^*$ were used.
This multiplicative grid search approach requires running the subroutine on every point in the multiplicative grid, which translates to $\bigO{1 / \delta \cdot \log (n)}$ calls to the subroutine.
Thus, this ``brute force'' multiplicative grid search adds an additional $\bigtO{1 / \delta}$ factor to the running time, which is undesirable, especially for higher accuracy applications where a smaller $\delta$ is preferred.

We propose a binary search method which achieves the same approximation guarantee using exponentially fewer calls to \knapsackalg as a subroutine.
The key to our binary search method is a careful use of the case analysis in Proposition~\ref{prop:knapsack_general_result}.
The algorithm \densitysearchalg is stated formally below as Algorithm~\ref{alg:density-search-alg}.
We consider points on a multiplicatively spaced grid of the interval $[1,n]$, which is given by $ \alpha_k = (1+\delta)^{k}$ for $k = 0, 1, \dotsc, \left\lceil \frac{1}{\delta} \log n \right\rceil$, where $\delta$ is an input parameter that specifies the granularity of the grid.
Another input to the algorithm is $\beta$, which specifies the relation between the points in the grid $[1,n]$ and the density thresholds which are used.
Let us consider the non-monotone case for now, in which case we should set $$\beta = 2(1 - \eps) \left( \frac{1 - \ell^{-1} - \eps}{p+1 + 2m} \right).$$
In this case, note that each point $\alpha_k$ in the $[1,n]$ grid corresponds to the choice of density threshold $$\density_k = \beta \cdot \maxgain \cdot \alpha_k = 2(1 - \eps) \left( \frac{1 - \ell^{-1} - \eps}{p+1 + 2m} \right) \maxgain \cdot \alpha_k.$$ 
The algorithm tries to zoom in on the optimal density threshold using binary search, while using the value of the indicator $E$ for each call to \knapsackalg to make the decision in each iteration of the search (we denote by $E_i$ the value of this indicator for call number $i$). 
While this indicator does not necessarily indicate the relationship between the the current density threshold and $\density^*$, it does give enough of a signal around which we may construct a binary search. 
In particular, if $E_i = 0$, then we get a good approximation as long as our current density threshold is an overestimate of $\density^*$, and thus, in the future we only need to consider higher density thresholds. 
Likewise, if $E_i = 1$, then we get a good approximation as long as our current density threshold is an underestimate  of $\density^*$, and thus, in the future we only need to consider lower density thresholds.


\begin{algorithm}[H]
	\caption{\densitysearchalg($\gnd, f, \cI, \ell, \delta, \varepsilon, \beta$)}\label{alg:density-search-alg}
	Initialize upper and lower bounds $k_{\ell} = 1$, $k_{u} = \lceil \frac{1}{\delta} \log n \rceil$. \\
	Let $\maxgain = \max_{u \in \gnd} f(u)$, and initialize iteration counter $i \gets 1$. \\
	\While{ $| k_u - k_\ell | > 1$}
	{
		Set middle bound $k_i = \left\lceil \frac{k_\ell + k_u }{2} \right\rceil$. \\
		Set density ratio $\density_i \gets \beta \cdot  \maxgain (1 + \delta)^{k_i}$.\\
		Obtain set $S_i \gets \knapsackalg(\gnd, f, \cI, \ell, \density_i, \varepsilon)$. \\
		\uIf{$E_i = 0$}{
			Increase lower bound $k_\ell \gets k_i$. \\
		}
		\Else{
			Decrease upper bound $k_u \gets k_i$. \\	
		}
		Update iteration counter $i \gets i+1$.
	}
	Set density ratio $\density_i \gets \beta \cdot \maxgain (1 + \delta)^{k_\ell}$.\\
	Obtain set $S_i \gets \knapsackalg(\gnd, f, \cI, \ell, \density_i, \varepsilon)$. \\
	\Return{the set $\retsol$ maximizing $f$ among the sets $\{ S_j \}_{j=1}^i$}.
\end{algorithm}

The following proposition bounds the number of calls made to \knapsackalg and provides an approximation guarantee.

\begin{proposition} \label{prop:density-search-alg-result}
	\densitysearchalg makes $\bigtO{1}$ calls to \knapsackalg.
	Additionally assume that the independence system satisfies the conditions in Proposition~\ref{prop:knapsack_general_result} for every execution of $\knapsackalg$.
	If $\beta = 2(1 - \eps) \left( \frac{1 - \ell^{-1} - \eps}{p+1 + 2m} \right)$, then the solution $S$ returned by \densitysearchalg satisfies 
	\[
	f(S) 
	\geq (1 - \delta)(1 - \eps) \left( \frac{1 - \ell^{-1} - \eps}{p+1 + 2m} \right) f(\OPT)
	\geq (1 - \delta) (1 - 2 \eps)^{2} \left( \frac{1 - \ell^{-1}}{p +1 + 2m} \right) f(\OPT)
	\]
	when the number of solutions $\ell$ is at least 2.
	Moreover, if $f$ is monotone and $\beta = 2(1 - \eps) \left( \frac{1 - \eps}{p+1 + 2m} \right)$ then this lower bound further improves to 
	\[
	f(S) \geq 
	(1 - \delta)( 1 - \eps)^2 \left( \frac{1}{p + 1 + 2m} \right) f(\OPT)
	\]
	for any number of solutions $\ell$.
\end{proposition}
\noindent \textbf{Remark:} Observe that Proposition~\ref{prop:density-search-alg-result} requires a different value for $\beta$ in the cases of monotone and non-monotone functions. This is necessary because the $\rho^*$ corresponding to these two cases are different, and the value of $\beta$ is used to adjust the range in which Algorithm~\ref{alg:density-search-alg} searches for a density approximating $\rho^*$ so that this range is guarantee to include $\rho^*$. One can avoid this by slightly increasing the range in which Algorithm~\ref{alg:density-search-alg} searches so that it is guaranteed to include both possible values for $\rho^*$. Since the ratio between the values of $\rho^*$ corresponding to the two cases is only a constant as long as the sum $\ell^{-1} + \eps$ is bounded away from $1$, such an expansion of the search range will have an insignificant effect on the time complexity of the algorithm in most regimes of interest.
\begin{proof}[Proof of Proposition~\ref{prop:density-search-alg-result}]
	We begin by bounding the number of calls that \densitysearchalg makes to \knapsackalg.
	Recall that the number of points in the $\delta$-multiplicative discretization is $\bigO{1 / \delta \cdot \log n}$. 
	At each iteration of the binary search, \knapsackalg is called once.
	It is well known that binary search requires only logarithmically many iterations to terminate.
	Thus, the number of calls to \knapsackalg is 
	$
	\bigO{\log \left( 1 / \delta \cdot \log n \right)}
	= \bigO{\log \left( 1 / \delta \right) + \log \log n }
	= \bigtO{ 1 }$ 
	calls to \knapsackalg.
		
	Now we prove the approximation guarantee of \densitysearchalg using the approximation guarantees of \knapsackalg.
	Let us first consider the general non-monotone case when $\beta = 2(1 - \eps) \left( \frac{1 - \ell^{-1} - \eps}{p+1 + 2m} \right)$.
	We proceed by a case analysis. 
	For the first case, suppose that at some iteration $i$ of \densitysearchalg, it called \knapsackalg with a density threshold $\density_i$ such that $\density_i \leq \density^*$ and the indicator $E_i$ ended up with the value $0$. 
	By Proposition~\ref{prop:knapsack_general_result}, we get in this case
	\begin{align*}
	f(S) &\geq f(S_i)\\
	& \geq  \frac{1 - \eps}{p+1} \cdot \Big( \left( 1 - \ell^{-1} - \eps \right) f(\OPT) - m \density_i \Big) \\
	&\geq \frac{1 - \eps}{p+1} \cdot \Big( \left( 1 - \ell^{-1} - \eps \right) f(\OPT) - m \density^* \Big)\\
	&= (1 - \eps)\left( \frac{1 - \ell^{-1} - \eps}{p+1 + 2m} \right) f(\OPT) \enspace.
	\end{align*}
	The second case is that at some iteration $i$ of \densitysearchalg, it called \knapsackalg with a density threshold $\density_i$ such that $\density_i \geq \density^*$ and the indicator $E_i$ ended up with the value $1$. 
	By Proposition~\ref{prop:knapsack_general_result}, we get in this case
	\[
	f(S)
	\geq
	f(S_i)
	\geq \frac{1}{2} \density_i 
	\geq \frac{1}{2} \density^* 
	= (1 - \eps) \left( \frac{1 - \ell^{-1} - \eps}{p+1 + 2m} \right) f(\OPT)
	\enspace.
	\]
	
	The last case we need to consider is the case that neither of the above cases happens in any iteration. 
	One can observe that in this case the binary search of \densitysearchalg chooses in each iteration the half of its current range that includes $\rho^*$. 
	Thus, we have
	\[
		2(1 - \eps) \left( \frac{1 - \ell^{-1} - \eps}{ p+1 + 2m} \right) \maxgain (1 + \delta)^{k_\ell}
		\leq
		\density^*
		\leq
		2(1 - \eps) \left( \frac{1 - \ell^{-1} - \eps}{ p+1 + 2m} \right) \maxgain (1 + \delta)^{k_u}
		\enspace.
	\]
	Let us now denote the final value of $i$ by $\hat{\imath}$. 
	Since the leftmost side of the last inequality is equal to $\density_{\hat{\imath}}$ and the rightmost side is larger than the leftmost side by at most a factor $1 + \delta$ (since $k_u - k_\ell \leq 1$ when \densitysearchalg terminates), we get $\density_{\hat{\imath}} \leq \density^* \leq (1 + \delta)\density_{\hat{\imath}}$, and by Proposition~\ref{prop:knapsack_general_result},
	\begin{align*}
	f(S)
	\geq{} & f(S_{\hat{\imath}})\\
	\geq{}& \min\left\{\frac{1}{2} \density_{\hat{\imath}}, \frac{1 - \epsilon}{p+1} \cdot \Big( \left( 1 - \ell^{-1} - \eps \right) f(\OPT) - m \density_{\hat{\imath}} \Big)\right\}\\
	\geq{} & \min\left\{\frac{1}{2} \frac{\density^*}{1 + \delta}, \frac{1 - \epsilon}{p+1} \cdot \Big( \left( 1 - \ell^{-1} - \eps \right) f(\OPT) - m \density^* \Big) \right\}\\
	\geq{} & (1 - \delta) \min\left\{\frac{1}{2} \density^*, \frac{1 - \epsilon}{p+1} \cdot \Big( \left( 1 - \ell^{-1} - \eps \right) f(\OPT) - m \density^* \Big)\right\}\\
	={} & (1 - \delta)(1 - \eps) \left( \frac{1 - \ell^{-1} - \eps}{ p+1 + 2m} \right) f(\OPT) \enspace,
	\end{align*}
	which establishes the first inequality in the statement of the proposition.
	We establish the second inequality by observing that for $\ell \geq 2$, the quantity $- \eps + 2 \eps \ell^{-1}$ is non-positive; and
	thus, 
	\[
	1 - \ell^{-1} - \eps
	\geq 
	1 - \ell^{-1} - 2\eps + 2 \eps \ell^{-1}
	= (1 - \ell^{-1}) ( 1 - 2 \eps) \enspace,
	\]
	which establishes the proposition for non-monotone objectives.
	The analysis for monotone objectives follows in an analogous manner.
\end{proof}



We are now ready to present the main approximation results for \densitysearchalg when the independence system is either a $k$-system or a $k$-extendible system.

\begin{theorem} \label{thm:knapsack-deterministic_extendible}
	Suppose that $(\gnd, \cI)$ is a $k$-extendible system, the number of solutions is set to $\ell = M + 1$---where $M = \max \left( \lceil \sqrt{1 + 2 m} \rceil , k \right)$, and the two error terms are set to be equal (\ie, $\eps = \delta \in (0,1/2)$).
	Then, \densitysearchalg requires $\bigtO{M n / \eps}$ oracle calls as well as $\bigtO{M m n / \eps}$ arithmetic operations and produces a solution whose approximation ratio is at most $(1 - 2 \eps)^{-3}$ times
	\[ 
		 \max\left\{k + \frac{2m + 1}{k}, 1 + 2\sqrt{2m + 1} \right\} + 2m + 2 \enspace.
	\]
	Moreover, when $f$ is non-negative monotone submodular and the number of solutions is chosen so that $\ell \leq k+1$, then the approximation ratio improves to $(1 - \eps)^{-3} \cdot (k+ 2m + 1)$.
\end{theorem}

\begin{theorem} \label{thm:knapsack-deterministic_system}
	Suppose that $(\gnd, \cI)$ is $k$-system, the number of solutions is $\ell = \lfloor 2 + \sqrt{k + 2m + 2 } \rfloor$, and the two error terms are set equal as $\eps = \delta \in (0, \nicefrac{1}{2})$.
	Then, $\densitysearchalg$ requires $\bigtO{n\sqrt{k + m} / \eps}$ oracle calls as well as $\bigtO{mn \sqrt{k + m} / \eps}$ arithmetic operations and produces a solution whose approximation ratio is at most $(1 - 2\eps)^{-3} \cdot (1 + \sqrt{k+2m +2})^2$. 
	Moreover, when $f$ is non-negative monotone submodular and the number of solutions is set to $\ell=1$, then the approximation ratio improves to $(1 - \eps)^{-3}\cdot \left( k+ 2m +1 \right)$.
\end{theorem}

As for the previous algorithms, the approximation ratio guaranteed by \densitysearchalg is improved for the subclass of $k$-extendible systems, at the cost of a slightly larger running time.
The algorithm guarantees the same approximation factor for monotone objectives for both $k$-extendible and $k$-systems.
In both cases, the best choice of the number of solutions $\ell$ depends on the number of knapsack constraints $m$.
Moreover, we remark that in the absence of any additional knapsack constraints (that is, $m = 0$), the approximation guarantees of Theorems~\ref{thm:knapsack-deterministic_extendible} and \ref{thm:knapsack-deterministic_system} recover the guarantees of the slower \mainalg, up to the $(1 - \eps)^{-3}$ error terms.
However, unlike \mainalg, the running time of \densitysearchalg is nearly-linear in the size of the ground set.

The $(1 - 2\eps)^{-3}$ multiplicative error terms may seem somewhat non-intuitive at first glance, but it turns out that they can be replaced with $1 + \bigO{\eps}$. In particular, for $\eps \in (0, \nicefrac{1}{4})$ the multiplicative error term  $(1 - 2\eps)^{-3}$ is at most $1 + 28 \eps$.
This follows by the convexity of the function $y(t) = (1 - 2 t)^{-3}$ within the range $[0, 1/2)$.
More specifically, by setting $\lambda = 4 \eps \in (0, 1)$, we get
\begin{align*}
(1 - 2 \eps)^{-3}
&= y(\eps)
= y \left( (1 - \lambda) \cdot 0 + \lambda \cdot \nicefrac{1}{4} \right) 
\leq ( 1 - \lambda) \cdot y(0) + \lambda \cdot y( \nicefrac{1}{4} ) \\
&= (1 - 4 \eps) + 4 \eps \left( 1 - 2 \cdot \frac{1}{4} \right)^{-3}
= (1 - 4 \eps) + 4 \eps \cdot 8
= 1 + 28 \eps \enspace.
\end{align*}
Similarly, one can also show that for all $\eps \in (0, \nicefrac{1}{4})$, the error term $(1 - \eps)^{-3}$, which appears in the approximation for monotone submodular objectives, is at most $(1 + 6 \eps)$.
Furthermore, by scaling $\eps$ one may transfer the constant in front of $\eps$ to the running time, which remains $\bigtO{n / \eps}$. This way, one may consider the multiplicative error term in the approximation factor of the algorithm to be a clean $1 + \eps$.

The proofs of Theorems~\ref{thm:knapsack-deterministic_extendible} and \ref{thm:knapsack-deterministic_system} follow from the unified meta-analysis of Proposition~\ref{prop:density-search-alg-result} in the same way that the Theorems~\ref{thm:deterministic_extendible} and \ref{thm:deterministic_system} follow from the meta-analysis of Proposition~\ref{prop:general_result}.
Namely, the constructions of the sets $\SetF{O}{i}{j}$ in Sections~\ref{sec:analysis-extendible} and \ref{sec:analysis-k-system} demonstrate that the conditions of Proposition~\ref{prop:density-search-alg-result} hold with $p = \max (k, \ell - 1)$ for $k$-extendible systems and $p = k + \ell - 1$ for general $k$-systems.
The final step is then to choose the number of solutions $\ell$ to optimize the resulting approximation ratios.
Although these steps are conceptually similar to the choice of $\ell$ in the analysis of \mainalg, they are somewhat involve, and so we reproduce them here.

\begin{proof}[Proof of Theorem~\ref{thm:knapsack-deterministic_extendible}]
The construction of sets $\SetF{O}{i}{j}$ in Proposition~\ref{prop:k-extendible-construction} demonstrates that the conditions of Proposition~\ref{prop:knapsack_general_result} are satisfied with with $p = \max (k, \ell - 1)$.
Thus, Proposition~\ref{prop:density-search-alg-result} implies that the approximation ratio of \densitysearchalg with $\ell = M + 1$ is at most $(1 - 2 \eps)^{-3}$ times the quantity,
\[
\frac{p + 1 + 2m }{1 - \ell^{-1}}
=
\frac{\max (k, \ell - 1) + 1 + 2m}{1 - \ell^{-1}} \enspace.
\]
Trying to optimize this quantity, we may set $ \ell = M + 1$, where $M = \max \left( \lceil \sqrt{1 + 2 m} \rceil , k \right)$. 
Note that $M$ is at least $k$ and so $\max(k, \ell - 1) = \max(k, M) = M$.
Thus, by the above, we have that the approximation ratio of \densitysearchalg with $\ell = M + 1$ is at most $(1 - 2 \eps)^{-3}$ times the quantity
\begin{align*}
\frac{M + 1 + 2m}{1 - \frac{1}{M +1}}
={} & \frac{(M+1) \left( M+1 + 2m \right) }{M}\\
={}& \frac{(M+1)^2}{M} + \left( \frac{M+1}{M} \right) 2m \\
\leq{} &
	\max\left\{k + 2m + 2 + \frac{2m + 1}{k}, 2m + 3 + 2\sqrt{2m + 1}\right\}\\
	={} &
	2m + 2 + \max\left\{k + \frac{2m + 1}{k}, 1 + 2\sqrt{2m + 1}\right\}
\enspace.
\end{align*}
For monotone submodular objectives, Proposition~\ref{prop:density-search-alg-result} implies that for all number of solutions $\ell \leq k + 1$, the approximation ratio is at most $(1 - \eps)^{-3}$ times the quantity
\[
p + 1 + 2m \leq \max(k, \ell - 1) + 1 + 2m \leq k + 1 + 2m \enspace. 
\qedhere
\]
\end{proof}

\begin{proof}[Proof of Theorem~\ref{thm:knapsack-deterministic_system}]
The construction of sets $\SetF{O}{i}{j}$ in Proposition~\ref{prop:k-system-construction} demonstrates that the conditions of Proposition~\ref{prop:knapsack_general_result} are satisfied with with $p = k + \ell - 1$.
Thus, Proposition~\ref{prop:density-search-alg-result} implies that the approximation ratio of \densitysearchalg with $\ell = \lfloor 2 + \sqrt{k + 2m + 2 } \rfloor$ is at most $(1 - 2 \eps)^{-3}$ times the quantity,
\[
\frac{p + 1 + 2m }{1 - \ell^{-1}}
=
\frac{k + \ell + 2m}{1 - \ell^{-1}}
=
\frac{k + 2m + \lfloor 2 + \sqrt{k + 2m + 2} \rfloor}{1 - 1/\lfloor 2 + \sqrt{k + 2m + 2} \rfloor}
\leq
\frac{k + 2m + 2 + \sqrt{k + 2m + 2}}{1 - 1/( 1 + \sqrt{k + 2m + 2})}.
\]
To simplify some calculations, let $\alpha = k + 2m + 2$.
By substituting $\alpha$ and rearranging terms, we obtain that the right hand side may be expressed as
\[
\frac{\alpha + \sqrt{\alpha}}{1 - 1/( 1 + \sqrt{\alpha})}
=
\frac{(1 + \sqrt{\alpha}) \cdot (\alpha +  \sqrt{\alpha})}{\sqrt{\alpha} }
= 
(1 + \sqrt{\alpha})(1 + \sqrt{\alpha}) 
= (1 + \sqrt{\alpha})^2
\enspace .
\]
Substituting back the value of $\alpha = k + 2m + 2$, we have that the approximation ratio is at most $(1 - 2\eps)^{-3} \cdot (1 + \sqrt{k+2m +2})^2$.
Finally, suppose that $f$ is monotone and $\ell = 1$. Then, by Proposition~\ref{prop:density-search-alg-result}, the approximation factor is at most
\[
(1 - \eps)^{-3} \cdot (p + 2m +1) 
= (1 - \eps)^{-3} \cdot (k + \ell -1 + 2m + 1) 
= (1 - \eps)^{-3} \cdot (k + 2m + 1) 
\enspace.
\qedhere
\]
\end{proof}

\section{Repeated Greedy} \label{sec:repeated-greedy}

In this section, we present and analyze the \repeatedgreedy algorithm for maximizing a submodular function $f$ subject to a $k$-system constraint $(\gnd, \cI)$. 
\repeatedgreedy iteratively executes the following three operations: first, the greedy algorithm is called as a subroutine to produce a feasible set $S_i$, then a subroutine for unconstrained submodular maximization produces a set $S'_i \subset S_i$ with large objective value, and elements of $S_i$ are removed from the remaining ground set.
After the algorithm makes $\ell$ iterations of this kind, the algorithm terminates and outputs the best set among all the sets constructed during its iterations.
The choice of $\ell$ will be determined later to yield the best approximation ratio.
A more formal description of \repeatedgreedy is given as Algorithm~\ref{alg:repeated_greedy}.

\begin{algorithm}
	\DontPrintSemicolon
	\caption{\greedy($\gnd, f, \cI $)}\label{alg:greedy}
	$S \leftarrow \varnothing$ \\
	\While{there exists $u \in \gnd'$ such that $S + u \in \cI$ and $ f(u \mid S) > 0$}{
		Let $u \in \gnd'$ be the element of this kind maximizing $f(u \mid S)$.\\
		Add $u$ to $S$.\\
	}
	\Return{$S$}.\\
\end{algorithm}

\begin{algorithm}
	\DontPrintSemicolon
	\caption{Repeated Greedy($\gnd, f, \cI, \ell$)}\label{alg:repeated_greedy}
	Let $\gnd_1 \gets \gnd$.\\
	\For{$i = 1$ \KwTo $\ell$}
	{
		Run greedy procedure $S_i \gets \greedy(\gnd_i, f, \cI)$ \\
		Filter the greedy solution $S'_i \gets \usm(S_i, f)$ \\
		Update ground set $\gnd_{i + 1} \gets \gnd_i \setminus S_i$.
	}
	\Return the set $S$ maximizing $f$ among the sets $\{S_i, S'_i \}_{i = 1}^\ell$.
\end{algorithm}

\repeatedgreedy calls a subroutine \usm for unconstrained submodular maximization.
Formally, the subroutine $\usm(A,f)$ takes as input a set $A$ and a non-negative submodular function $f$ defined on subsets of $A$ and returns a set $X \subset A$ such that $f(X) \geq \frac{1}{\alpha} f(B)$ for all $B \subset A$.
There are several known algorithms for \usm \citep{FMV07, GV11, Buchbinder2015}.
\cite{FMV07} showed that no algorithm using only polynomially many oracle queries can achieve an approximation ratio smaller than $\alpha = 2$.
For the sake of generality, we remain agnostic to the specific \usm subroutine that is being used and derive an approximation ratio for \repeatedgreedy that depends on $\alpha$; however, in order to obtain nearly-linear time algorithms, we restrict our attention to \usm subroutines which require at most $\bigO{n}$ oracle calls.
In the spirit of proposing deterministic algorithms, we further restrict our attention to deterministic \usm subroutines, although a randomized subroutine may be used here as well, with appropriate probabilistic caveats in the approximation ratio.
At the time of this writing, it is most natural to use the deterministic algorithm of \cite{Buchbinder2015}, which yields an approximation ratio of $\alpha = 3$ and runs in linear time.
We remark there that it is an interesting open problem to construct a deterministic linear time algorithm for \usm which achieves the optimal $\alpha = 2$ approximation, although~\cite{Buchbinder2016b} come close to achieving this goal. Specifically, they designed an algorithm for \usm achieving $(2 + \eps)$-approximation using $O(n/\eps)$ time.

\begin{observation}\label{obs:repeated-greedy-runtime-and-feasibility}
	\repeatedgreedy requires $\bigO{\ell r n}$ oracle calls and its output $S$ is independent.
\end{observation}
\begin{proof}
	We begin the proof by bounding the number of oracle calls used by \repeatedgreedy.
	Observe that \greedy has at most $|S| \leq r$ iterations, during which at most $n$ calls to the value and independence oracle calls are made.
	This means that a single execution of \greedy requires $\bigO{rn}$ oracle calls.
	As discussed above, we only consider implementations of \usm which require $\bigO{n}$ oracle calls, which is negligible compared to the computational requirements of \greedy.
	Finally, because $\ell$ solutions are produced by \repeatedgreedy (and thus, the algorithm makes only $\ell$ iterations), the total number of oracle calls is $\bigO{\ell r n}$.
	
	Next, we prove that the output $S$ is independent.
	For every $1 \leq i \leq \ell$, the set $S_i$ is initialized as independent because the greedy algorithm returns an independent set. 
	Moreover, the output $S'_i$ of \usm satisfies  $S'_i \subseteq S_i$, and thus, $S'_i$ is independent by the down-closed property of $k$-systems. 
	The observation now follows since the output $S$ of \repeatedgreedy is chosen as either $S_i$ or $S'_i$ for one of such $i$.
\end{proof}

We now present the main runtime and approximation guarantees for \repeatedgreedy.

\begin{theorem}\label{thm:repeated-greedy}
	Suppose that $(\gnd, \cI)$ is $k$-system and that the number of solutions is set to $\ell = \lfloor 1 + \sqrt{2(k+1)/ \alpha }  \rfloor $.
	Then, \repeatedgreedy requires $\bigO{\sqrt{k} r n}$ oracle calls and produces a solution whose approximation ratio is at most $k + \left( \sqrt{2 \alpha} \right) \sqrt{k} + (\alpha + 1) + \littleO{1}$. 
	Moreover, when $f$ is a non-negative monotone submodular function and the number of solutions is set to $\ell=1$, then the approximation ratio of \repeatedgreedy improves to $k+1$.
\end{theorem}

Let us compare the approximation ratio of \repeatedgreedy to that of \mainalg.
Before continuing, we remark that while \mainalg was able to achieve an improved $k + \bigO{1}$ approximation for the subclass of $k$-extendible systems, it is not possible to prove a similar approximation guarantee for \repeatedgreedy when the constraint belongs to this subclass. 
This is discussed in more detail in Section~\ref{sec:repeated-greedy-tight-extendible-approx}.

Furthermore, although both \mainalg and \repeatedgreedy achieve similar asymptotic approximation factors of $k + \bigO{\sqrt{k}}$ for the class of $k$-systems, the low order terms in the approximation factor of \repeatedgreedy are larger.
More precisely, \repeatedgreedy achieves an approximation of $k + (\sqrt{2 \alpha}) \sqrt{k} + (1 + \alpha) + \littleO{1}$ while \mainalg achieves an approximation of $(1 + \sqrt{k+2})^2 = k + 2 \sqrt{k} + 3 + \littleO{1}$.
While both algorithms have the same coefficient for the leading $k$ term, \repeatedgreedy has larger coefficients in the low order terms.
In particular, the lower order terms of \repeatedgreedy depend on $\alpha$, the approximation ratio for \usm, which will be at least $2$ by the hardness result of \cite{FMV07}.
Moreover, even the subconstant $\littleO{1}$ term is larger for \repeatedgreedy.
While this term goes to zero as $k$ grows, it may be non-negligible for very small $k$ values.
An explicit form for this term is derived in the proof of Theorem~\ref{thm:repeated-greedy}.
After the proof, we analyze this term more carefully, showing that for $\alpha=3$ and $k=1$ it is $\approx 20.4$ and for $k=10$ the term is $\approx 2.1$.
On the other hand, the $\littleO{1}$ term in the approximation ratio of \mainalg is at most $2/\sqrt{k}$.

We now begin the analysis of the approximation ratio of \repeatedgreedy for $k$-systems. 
As before, we write $\OPT$ to denote an independent set of $(\gnd, \cI)$ maximizing $f$. 
At a high level, our analysis proceeds by showing that (1) by properties of the \greedy and \usm procedures, the value of the output of \repeatedgreedy is proportional to the average value of the union between a set from $\{ S \}_{i=1}^\ell$ and $\OPT$, and then (2) because the sets $\{ S \}_{i=1}^\ell$ are disjoint, this aforementioned average cannot be considerably smaller than the value of $\OPT$.
More concretely, our analysis is based on three lemmata.
The first of these lemmata, presented below, gives lower bounds on the objective values of the two sets $S_i$ and $S'_i$ produced at each iteration.

\begin{lemma} \label{lem:known_approx_results}
	For every $1 \leq i \leq \ell$, $f(S_i) \geq \frac{1}{k+1}f(S_i \cup (\OPT \cap \gnd_i))$ and $f(S'_i) \geq \frac{1}{\alpha} f(S_i \cap \OPT)$.
\end{lemma}
\begin{proof}
	The first inequality is a direct application of Lemma~3.2 of~\cite{GRST10}, which states that a set $S$ obtained by running greedy with a $k$-system constraint must obey $f(S) \geq \frac{1}{k+1} f(S \cup C)$ for all independent sets $C$.
	Notice that the set $S_i$ is the output of the greedy algorithm when executed on the $k$-system obtained by restricting $(\gnd, \cI)$ to the ground set $\gnd_i$ and that $C = \OPT \cap \gnd_i$ is an independent set of this restricted $k$-system.
	This yields that $f(S_i) \geq \frac{1}{k+1} f(S_i \cup (\OPT \cap \gnd_i))$.
	
	Let us now explain why the second inequality of the lemma holds. 
	Observe that $S_i \cap \OPT$ is a subset of $S_i$.
	Thus, by the approximation guarantees of \usm, $f(S_i') \geq \frac{1}{\alpha}  f(S_i \cap \OPT) $.
\end{proof}

The second lemma we need is the following basic fact about submodular functions.

\newcommand{\submodularFactLemma}{
	Suppose $f$ is a non-negative submodular function over ground set $\gnd$. For every three sets $A, B, C \subseteq \gnd$, $f(A \cup (B \cap C)) + f(B \setminus C) \geq f(A \cup B)$.
}
\begin{lemma}\label{lem:submodular_fact1}
	\submodularFactLemma
\end{lemma}
\newcommand{\submodularFactProof}{
	\begin{proof}
		Observe that
		\begin{align*}
			f(A \cup (B \cap C)) + f(B \setminus C)
			\geq{} &
			f(A \cup (B \cap C) \cup (B \setminus C)) + f((A \cup (B \cap C)) \cap (B \setminus C))\\
			\geq{} &
			f(A \cup (B \cap C) \cup (B \setminus C)) \\
			= &	f(A \cup B)
			\enspace,
		\end{align*}
		where the first inequality follows from the submodularity of $f$, and the second inequality follows from its non-negativity.
	\end{proof}
}\submodularFactProof

The third lemma we need is Lemma~\ref{lem:distribution} [Lemma~2.2 of \cite{BFNS14}] which allows us to relate the average value of $f(S_i \cup \OPT)$ to the optimal value $f(\OPT)$.
Together, these allow us to prove Proposition~\ref{prop:repeated-greedy-analysis}, which is a general approximation guarantee for \repeatedgreedy that holds for any number of iterations $\ell \geq 1$.

\begin{proposition}\label{prop:repeated-greedy-analysis}
	If $(\gnd, \cI)$ is a $k$-system, then the solution returned by \repeatedgreedy has an approximation ratio of at most $\frac{ k+1 + \frac{\alpha}{2} (\ell - 1)}{1-1/\ell}$. 
	Moreover, this approximation improves to $k + 1 + \frac{\alpha}{2}(\ell - 1)$ for monotone submodular objectives.
\end{proposition}
\begin{proof}
	Observe that, for every $1 \leq i \leq \ell$, we have
	\begin{equation} \label{eq:set_equalities}
		\OPT \setminus \gnd_i = \OPT \cap ( \gnd \setminus \gnd_i) = \OPT \cap \left( \cup_{j=1}^{i-1} S_i \right) = \cup_{j=1}^{i-1} \left( \OPT \cap S_j \right)
	\end{equation}
	where the first equality holds because $\OPT \subseteq \gnd$, and the second equality follows from the removal of $S_i$ from the ground set in each iteration of \repeatedgreedy.
	Using the previous lemmata and this observation, we can obtain a lower bound on the objective value of the returned solution $S$ in terms of the  average value of $f(S_i \cup \OPT)$ as
	{\allowdisplaybreaks
		\begin{align*}
			\frac{1}{\ell} \sum \limits_{i=1}^\ell f(S_i \cup \OPT) 
			&\leq \frac{1}{\ell} \sum \limits_{i=1}^\ell f(S_i \cup (\OPT \cap \gnd_i)) + \frac{1}{\ell} \sum \limits_{i=1}^\ell f(\OPT \setminus \gnd_i) &\text{(Lemma~\ref{lem:submodular_fact1})} \\
			&= \frac{1}{\ell} \sum \limits_{i=1}^\ell f(S_i \cup (\OPT \cap \gnd_i)) + \frac{1}{\ell} \sum \limits_{i=1}^\ell f \left(\cup_{j=1}^{i-1}(\OPT \cap S_j) \right) &\text{(Equality~\eqref{eq:set_equalities})} \\
			&\leq \frac{1}{\ell} \sum \limits_{i=1}^\ell f(S_i \cup (\OPT \cap \gnd_i)) + \frac{1}{\ell} \sum \limits_{i=1}^\ell \sum \limits_{j=1}^{i-1} f(\OPT \cap S_j) &\text{(submodularity)} \\
			&\leq \frac{k+1}{\ell} \sum \limits_{i=1}^\ell f(S_i) + \frac{\alpha}{\ell} \sum \limits_{i=1}^\ell \sum \limits_{j=1}^{i-1} f(S'_j) &\text{(Lemma~\ref{lem:known_approx_results})} \\
			&\leq \frac{k+1}{\ell} \sum \limits_{i=1}^\ell f(S) + \frac{\alpha}{\ell} \sum \limits_{i=1}^\ell \sum \limits_{j=1}^{i-1} f(S) &\text{(definition of $S$)} \\
			&= \left[ k+1 + \alpha(\ell-1) /2 \right] f(S)
			\enspace.
		\end{align*}
	}
	Rearranging this inequality yields the following lower bound on the value of the returned solution.
	\begin{equation}\label{eq:rg-lower-bound-average}
		f(S) \geq \frac{1}{k + 1 + \frac{\alpha}{2} (\ell - 1)} \cdot \frac{1}{\ell} \sum_{i=1}^\ell f(S_i \cup \OPT) \enspace.
	\end{equation}
	In order to remove the dependence of the right hand side on the solutions $S_i$, we again use Lemma~\ref{lem:distribution} [Lemma~2.2 of \cite{BFNS14}].
	In particular, consider a set $\bar{S}$ chosen uniformly at random from the $\ell$ constructed solutions $S_1, S_2, \dots S_\ell$. 
	Because these solutions are disjoint by construction, an element can belong to $\bar{S}$ with probability at most $\ell^{-1}$. 
	Hence, applying Lemma~\ref{lem:distribution} to the submodular function $g(S) = f(\OPT \cup S)$, we get
	\begin{equation}\label{eq:buchbinder-lb-on-average}
		\frac{1}{\ell} \sum_{i=1}^\ell f(S_i \cup \OPT)
		= \Exp{f(\OPT \cup \bar{S})}
		= \Exp{g(\bar{S})} 
		\geq (1 - \ell^{-1}) \cdot g(\varnothing)
		= (1 - \ell^{-1}) \cdot f(\OPT) \enspace.
	\end{equation}
	Substituting \eqref{eq:buchbinder-lb-on-average} into the lower bound of \eqref{eq:rg-lower-bound-average} yields the desired approximation.
	
	When $f$ is monotone submodular, we may obtain an improved approximation ratio by applying monotonicity directly to the lower bound \eqref{eq:rg-lower-bound-average}.
	In particular, applying monotonicity yields
	\[
	\frac{1}{\ell} \sum_{i=1}^\ell f(S_i \cup \OPT)
	\geq \frac{1}{\ell} \sum_{i=1}^\ell f( \OPT)
	= f(\OPT) \enspace,
	\]
	which yields the desired approximation in the monotone setting.
\end{proof}

Note that the approximation factor derived in Proposition~\ref{prop:repeated-greedy-analysis} is not a monotone function of the number of iterations $\ell$ in \repeatedgreedy.
In this sense, we can optimize this derived approximation factor by choosing some appropriate value of $\ell$.
On the other hand, the true approximation factor of \repeatedgreedy can only increase as the number of iterations increases, as more solutions are produced.
Thus, the non-monotonicity of our derived approximation factor in $\ell$ should be regarded as an artifact of our analysis and not as the actual behavior of \repeatedgreedy.

Nevertheless, we may derive bounds on the approximation factor of \repeatedgreedy for $k$-systems when the number of iterations is set as $\ell = \bigO{\sqrt{k}}$.
In particular, setting $\ell$ to minimize the approximation factor presented in Proposition~\ref{prop:repeated-greedy-analysis} yields the approximation guarantee of Theorem~\ref{thm:repeated-greedy}.

\begin{proof}[Proof of Theorem~\ref{thm:repeated-greedy}]
	First, we show that the number of oracle calls is at most $\bigO{\sqrt{k} r n}$.
	By Observation~\ref{obs:repeated-greedy-runtime-and-feasibility}, the number of oracle calls is at most $\bigO{\ell r n}$ and so the result follows from our choice of the number of solutions, $\ell = \lfloor 1 + \sqrt{2(k+1)/ \alpha }  \rfloor = \bigO{\sqrt{k}}$.
	
	Next, we show that setting the number of solutions to $\ell = \lfloor 1 + \sqrt{2(k+1)/ \alpha }  \rfloor $ yields an approximation factor of at most $k + \left( \sqrt{2 \alpha} \right) \sqrt{k} + (\alpha + 1) + \littleO{1}$.
	We begin by substituting this value of $\ell$ into the approximation guarantee of Proposition~\ref{prop:repeated-greedy-analysis}.
	This gives us that the approximation factor is
	\begin{align*}
		\frac{ k+1 + \frac{\alpha}{2} (\ell - 1)}{1-1/\ell}
		&=	\frac{ k+1 + \frac{\alpha}{2} (\lfloor 1 + \sqrt{2(k+1)/ \alpha }  \rfloor - 1)}{1-\frac{1}{\lfloor 1 + \sqrt{2(k+1)/ \alpha }  \rfloor}} \\
		&\leq \frac{ k+1 + \frac{\alpha}{2} ( 1 + \sqrt{2(k+1)/ \alpha }  - 1)}{1-\frac{1}{ \sqrt{2(k+1)/ \alpha } }} \\
		&= \frac{\alpha}{2} \left( \frac{ 2(k+1)/ \alpha + \sqrt{2(k+1)/ \alpha }}{1 - \frac{1}{\sqrt{2(k+1)/ \alpha }}} \right)
		\enspace.
	\end{align*}
	By defining $\gamma = 2(k+1)/\alpha$, we can simplify the term inside the parenthesis, and write it as
	\[
	\frac{\gamma + \sqrt{\gamma}}{1 - 1/ \sqrt{\gamma}}
	= (1 + \sqrt{\gamma})^2 + 1 + \frac{4}{\gamma - 1} + \frac{2}{1 + \sqrt{\gamma}}
	\enspace.
	\]
	Substituting this back into the calculation above, we have that the approximation factor is at most
	\begin{align*}
		\frac{\alpha}{2} 
		&
		\left( \left( 1 + \sqrt{\frac{2}{\alpha} (k+1)}\right)^2 +1 + \frac{4}{\frac{2}{\alpha} (k+1) - 1} + \frac{2}{1 + \sqrt{\frac{2}{\alpha} (k+1)}}  \right) \\
		&= 
		\frac{\alpha}{2} 
		\left( \frac{2}{\alpha} (k+1) + 2 \sqrt{\frac{2}{\alpha} (k+1)} + 2 + \frac{4}{\frac{2}{\alpha} (k+1) - 1} + \frac{2}{1 + \sqrt{\frac{2}{\alpha} (k+1)}}  \right) \\
		&= k + 1 + \sqrt{2 \alpha (k+1)} + \alpha + \frac{\alpha}{\sqrt{\frac{2}{\alpha} (k+1)} - 1}\\
		&= k + \left( \sqrt{2 \alpha} \right) \sqrt{k} + (\alpha + 1) + \eta  \enspace,
	\end{align*}
	where the term $\eta = \littleO{1}$ decreases gradually with $k$ and is given explicitly as
	\[
	\eta = \sqrt{2 \alpha} \left( \sqrt{k+1} - \sqrt{k} \right) + \frac{\alpha}{\sqrt{\frac{2}{\alpha} (k+1)} - 1} \enspace.
	\]
	
	Finally, we demonstrate the our analysis obtains the improved approximation ratio of $k+1$ for monotone submodular functions when running the greedy algorithm.
	For monotone submodular functions, Proposition~\ref{prop:repeated-greedy-analysis} yields an approximation factor of $k + 1 + \frac{\alpha}{2}(\ell - 1)$.
	In this case, constructing a single solution ($\ell = 1$) yields the approximation factor $k+1$ in the monotone case.
\end{proof}

Although the term $\eta$ goes to zero for large $k$, it may be non-negligible for very small $k$.
The magnitude of this $\eta$ term also depends on the \usm approximation ratio, $\alpha$.
Roughly speaking, $\eta = \bigO{\sqrt{\alpha^3 / k}}$ so that a decrease in $\alpha$ can yield a significant decrease in this sub-constant term.
For example, for $\alpha = 3$ and $k=1$, we have that $\eta \approx 20.4$, and if $k$ increases to $10$ then $\eta \approx 2.1$.
On the other hand, if $\alpha = 2$ and $k = 1$, then we have that $\eta \approx 5.7$, while when $k$ increases to $10$, we get $\eta \approx 1.2$.

\subsection{Tight approximation analysis for $k$-extendible systems } \label{sec:repeated-greedy-tight-extendible-approx}

Earlier in Section~\ref{sec:repeated-greedy}, we proved that \repeatedgreedy achieves a $k + \bigO{\sqrt{k}}$ approximation under a $k$-system if allowed to run for $\bigO{\sqrt{k}}$ iterations.
A natural question is whether, like \mainalg, the approximation factor of \repeatedgreedy may improve for the subclass of $k$-extendible systems.
In this section, we answer this question in the negative by showing that \repeatedgreedy achieves an approximation factor of $k + \Omega(\sqrt{k})$ for the subclass of $k$-extendible systems.
In particular, we prove the following theorem:

\begin{theorem} \label{th:repeated_bad_example}
	The approximation ratio of \repeatedgreedy for the problem of maximizing a non-negative submodular function subject to a $k$-extendible constraint is $k + \Omega(\sqrt{k})$ regardless of the number of iterations used by the algorithm.
\end{theorem}

We remark that the proof of Theorem~\ref{th:repeated_bad_example} works even when the constraint is restricted to be a $k$-matchoid, which is a special case of a $k$-extendible constraint.

To prove Theorem~\ref{th:repeated_bad_example}, we construct a family of instances on which \repeatedgreedy performs poorly. Specifically, we construct a bad instance, denoted by $I_k$, for every integer $k \geq 1$ such that $\sqrt{k}$ is also an integer. We begin the construction by defining for every integer $1 \leq i \leq \sqrt{k}$ the sets
\[
O_i = \{o_{i, j} \mid 1 \leq j \leq 1 + \sqrt{k}\} \enspace, \quad  B_i = \{b_i\} \enspace, \quad \text{and} \quad  D_i = \{d_{i, j} \mid 1 \leq j \leq \sqrt{k}\} \enspace.
\]
Then, the ground set of the instance $I_k$ is the set $\cN_k = \bigcup_{i = 1}^{\sqrt{k}} (O_i \cup D_i \cup B_i)$. To define the objective function $f_k\colon 2^{\cN_k} \to \nnR$ of $I_k$, we first define for every integer $1 \leq i \leq \sqrt{k}$ an auxiliary function $g_i \colon 2^{O_i \cup D_i \cup B_i} \to \nnR$. For every set $S \subseteq O_i \cup D_i \cup B_i$,
\[
g_i(S)
=
\begin{cases}
	|\{j \mid S \cap \{o_{i, j}, d_{i, j}\} \neq \varnothing\}| + \frac{2|S \cap D_i| + |S \cap O_i|}{4\sqrt{k}} & \text{if $b_i \not \in S$} \enspace,\\
	\frac{3 + 4\sqrt{k} + |S \cap \{o_{i, 1 + \sqrt{k}}\}| - |S \cap (O_i \cup D_i - o_{i, 1 + \sqrt{k}})|}{4\sqrt{k}} & \text{if $b_i \in S$} \enspace.
\end{cases}
\]

\begin{observation}
	For every integer $1 \leq i \leq \sqrt{k}$, the function $g_i$ is non-negative and submodular.
\end{observation}
\begin{proof}
	The non-negativity of $g_i$ follows immediately from its definition since the size of the set $O_i \cup D_i - o_{i, k + 1}$ is $2\sqrt{k}$. 
	Therefore, we focus on proving that $g_i$ is submodular. 
	Recall that the function $g_i$ is defined on the ground set $O_i \cup D_i \cup B_i$.
	Thus, $g_i$ is submodular if $g_i(u \mid S)$ is a down-monotone function of $S$ for every element $u \in O_i \cup D_i \cup B_i$, where $S \subseteq O_i \cup D_i \cup B_i \setminus \{u\}$. 
	We do that by considering a few cases. 
	Consider first the case in which $u = b_i$. In this case
	\begin{align*}
		g_i(b_i \mid S)
		={} &
		\frac{3 + 4\sqrt{k} + |S \cap \{o_{i, 1 + \sqrt{k}}\}| - |S \cap (O_i \cup D_i - o_{i, 1 + \sqrt{k}})|}{4\sqrt{k}} \\&- |\{j \mid S \cap \{o_{i, j}, d_{i, j}\} \neq \varnothing\}| - \frac{2|S \cap D_i| + |S \cap O_i|}{4\sqrt{k}} \\
		={} &
		\frac{3 + 4\sqrt{k} - (4\sqrt{k} - 1) \cdot |S \cap \{o_{i, 1 + \sqrt{k}}\}| - |S \cap (O_i \cup D_i - o_{i, 1 + \sqrt{k}})|}{4\sqrt{k}} \\&- |\{j \mid S \cap \{o_{i, j}, d_{i, j}\} \neq \varnothing, j \neq 1 + \sqrt{k}\}| - \frac{2|S \cap D_i| + |S \cap O_i|}{4\sqrt{k}}
		\enspace,
	\end{align*}
	which is clearly a down-monotone function of $S$ (the second equality holds since $d_{i, k + 1}$ does not belong to the ground set $\cN_k$, and therefore, cannot appear in $S$). Consider now the case in which $u = o_{i, j}$ for some integer $1 \leq j \leq \sqrt{k}$. In this case
	\[
	g_i(o_{i, j} \mid S)
	=
	\begin{cases}
		\frac{1}{4\sqrt{k}} + 1 - |S \cap \{d_{i, j}\}| & \text{if $b_i \not \in S$} \enspace,\\
		- \frac{1}{4\sqrt{k}} & \text{if $b_i \in S$} \enspace,
	\end{cases}
	\]
	which is a down-monotone function of $S$ since the expression for the case $b_i \not \in S$ is always non-negative. The next case is when $u = o_{i, 1 + \sqrt{k}}$, which yields
	\[
	g_i(o_{i, 1 + \sqrt{k}} \mid S)
	=
	\begin{cases}
		1 + \frac{1}{4\sqrt{k}} & \text{if $b_i \not \in S$} \enspace,\\
		\frac{1}{4\sqrt{k}} & \text{if $b_i \in S$} \enspace,
	\end{cases}
	\]
	which is down-monotone. The last case to consider is the case of $u = d_{i, j}$ for some integer $1 \leq j \leq \sqrt{k}$. In this case
	\[
	g_i(d_{i, j} \mid S)
	=
	\begin{cases}
		\frac{1}{2\sqrt{k}} + 1 - |S \cap \{o_{i, j}\}| & \text{if $b_i \not \in S$} \enspace,\\
		- \frac{1}{4\sqrt{k}} & \text{if $b_i \in S$} \enspace,
	\end{cases}
	\]
	which is down-monotone since the expression for the case $b_i \not \in S$ is again always non-negative.
\end{proof}

The objective function $f_k$ of the instance $I_k$ can now be defined for every set $S \subseteq \cN_k$ by $f_k(S) = \sum_{i = 1}^{\sqrt{k}} g_i(S \cap (O_i \cup D_i \cup B_i))$. One can note that $f_k$ is non-negative and submodular since it is the sum of $\sqrt{k}$ functions having these properties. To complete the description of the instance $I_k$, we still need to define its constraint. To do that, let us associate each element of $I_k$ with up to $k$ colors from the list $\{\bot\} \cup \{(i, j) \mid 1 \leq i \leq \sqrt{k}, 1 \leq j \leq 1 + \sqrt{k}\}$. A set is feasible under our constraint if no two elements in it share a color (this constraint is a $k$-extendible set system because it can be represented as a $k$-matchoid by having one matroid for each color whose role is to allow at most a single element with that color in a feasible set). The colors of the different elements are as follows.
\begin{itemize}
	\item An element $o_{i, j} \in \cN_k$ has $(i, j)$ as its single color.
	\item An element $d_{i, j} \in \cN_k$ has all the colors in $\{\bot\} \cup \{(i', j') \mid 1 \leq i' \leq \sqrt{k}, 1 \leq j' \leq \sqrt{k}\} \setminus \{(i, j)\}$, which is $1 + (k - 1) = k$ different colors.
	\item An element $b_i \in \cN_k$ has all the colors in $\{\bot\} \cup\{(i', j') \mid 1 \leq i' \leq \sqrt{k}, i' \neq i, 1 \leq j' \leq 1 + \sqrt{k}\}$, which is $1 + (\sqrt{k} - 1)(1 + \sqrt{k}) = k$ different colors.
\end{itemize}

The following observation shows that the optimal solution for the instance $I_k$ has a lot of value.
\begin{observation} \label{obs:optimal_value}
	The value of the optimal solution for $I_k$ is at least $k + \frac{5\sqrt{k}}{4}$.
\end{observation}
\begin{proof}
	Note that the set $\bigcup_{i = 1}^{\sqrt{k}} O_i$ is a feasible solution. The value of this set according to $f_k$ is
	\[
	f_k \left( \bigcup_{i = 1}^{\sqrt{k}} O_i \right)
	=
	\sum_{i = 1}^{\sqrt{k}} g_i(O_i)
	=
	\sum_{i = 1}^{\sqrt{k}} \left(1 + \frac{1}{4\sqrt{k}}\right) \cdot |O_i|
	=
	\sqrt{k} \cdot \left(1 + \frac{1}{4\sqrt{k}}\right) \cdot (1 + \sqrt{k})
	\geq
	k + \frac{5\sqrt{k}}{4}
	\enspace.
	\qedhere
	\]
\end{proof}

Our next objective is to analyze the performance of \repeatedgreedy given the input $I_k$. We do that using the following two lemmata.
\begin{lemma}
	Let $L$ be a strict subset of $\{1, 2, \dotsc, \sqrt{k}\}$, and assume that the greedy algorithm is applied to the instance $I_k$ restricted to the ground set $\cN_k \setminus \{b_i, o_{i, 1 + \sqrt{k}} \mid i \in L\}$. Then, it outputs the set $\{b_i, o_{i, 1 + \sqrt{k}}\}$ for some $i \not \in L$.
\end{lemma}
\begin{proof}
	We begin the proof by considering the marginal contribution of every element of $\cN_k$ with respect to $\varnothing$.
	\begin{itemize}
		\item For every two integers $1 \leq i \leq \sqrt{k}$ and $1 \leq j \leq 1 + \sqrt{k}$, $f_k(o_{i, j} \mid \varnothing) = 1 + \frac{1}{4\sqrt{k}}$.
		\item For every two integers $1 \leq i \leq \sqrt{k}$ and $1 \leq j \leq \sqrt{k}$, $f_k(d_{i, j} \mid \varnothing) = 1 + \frac{2}{4\sqrt{k}}$.
		\item For every integer $1 \leq i \leq \sqrt{k}$, $f_k(b_{i} \mid \varnothing) = \frac{3 + 4\sqrt{k}}{4\sqrt{k}} = 1 + \frac{3}{4\sqrt{k}}$.
	\end{itemize}
	One can observe that the marginal contribution calculated above for the $b_i$ elements is larger than the marginal contributions calculated for the other elements (and is positive), and therefore, the first element that the greedy algorithm will add to its solution will be one such element that is available in the ground set.
	
	Assume therefore that the greedy algorithm has picked so far into its solution only the element $b_i$ for some $i \not \in L$. Due to the constraint, the only elements that can still be added to the solution once $b_i$ is in it are the elements of $O_i$. Their marginal contribution with respect to $\{b_i\}$ is
	\begin{itemize}
		\item For every integer $1 \leq j \leq \sqrt{k}$, $f_k(o_{i, j} \mid \{b_i\}) = -\frac{1}{4\sqrt{k}}$.
		\item $f_k(o_{i, 1 + \sqrt{k}} \mid \{b_i\}) = \frac{1}{4\sqrt{k}}$.
	\end{itemize}
	Since the marginal contribution calculated for $o_{i, 1 + \sqrt{k}}$ is the largest (and is positive), the greedy algorithm picks $o_{i, 1 + \sqrt{k}}$ as the next element to add to its solution. Furthermore, since the marginal contributions of the remaining elements of $O_i$ are already negative at this stage (and hence, will be negative in the future as well), the greedy algorithm does not pick any of them. Thus, its output set is $\{b_i, o_{i, 1 +\sqrt{k}}\}$, as promised.
\end{proof}

\begin{lemma}
	Let $L$ be a strict subset of $\{(i, j) \mid 1 \leq i, j \leq \sqrt{k}\}$, and assume that the greedy algorithm is applied to the instance $I_k$ restricted to the ground set $\cN_k \setminus (\{o_{i, j}, d_{i, j} \mid (i, j) \in L\} \cup \{b_i, o_{i, 1 + \sqrt{k}} \mid 1 \leq i \leq \sqrt{k} \}$. Then, it outputs the set $\{o_{i, j}, d_{i, j}\}$ for some $(i, j) \not \in L$.
\end{lemma}
\begin{proof}
	We begin the proof by considering the marginal contribution of every element of $\cN_k \setminus \{b_i, o_{i, 1 + \sqrt{k}} \mid 1 \leq i \leq \sqrt{k} \}$ with respect to $\varnothing$.
	\begin{itemize}
		\item For every two integers $1 \leq i \leq \sqrt{k}$ and $1 \leq j \leq \sqrt{k}$, $f_k(o_{i, j} \mid \varnothing) = 1 + \frac{1}{4\sqrt{k}}$.
		\item For every two integers $1 \leq i \leq \sqrt{k}$ and $1 \leq j \leq \sqrt{k}$, $f_k(d_{i, j} \mid \varnothing) = 1 + \frac{2}{4\sqrt{k}}$.
	\end{itemize}
	One can observe that the marginal contribution calculated above for the $d_{i, j}$ elements is larger than the marginal contribution calculated for the other elements (and is positive), and therefore, the first element that the greedy algorithm adds to its solution is one such element that is available in the ground set.
	
	Assume therefore that the greedy algorithm has picked so far into its solution only the element $d_{i, j}$ for some $(i, j) \not \in L$. Due to the constraint, the only element that can still be added to the solution once $d_{i, j}$ is in it is $o_{i, j}$, whose marginal with respect to $\{d_{i, j}\}$ is $f_k(o_{i, j} \mid \{d_{i, j}\}) = 1 / \sqrt{4k}$, which is positive. Hence, the greedy algorithm selects at this point $o_{i, j}$, and outputs the set $\{o_{i, j}, d_{i, j}\}$, as promised.
\end{proof}

Combining the two last lemmata, we get the following corollary.
\begin{corollary} \label{cor:possible_outputs}
	Regardless of number of iterations of \repeatedgreedy used, the only sets it can output are either subsets of $\{b_i, o_{i, 1 + \sqrt{k}}\}$ for some integer $1 \leq i \leq \sqrt{k}$ or subsets of $\{o_{i, j}, d_{i, j}\}$ for some integers $1 \leq i, j \leq \sqrt{k}$.
\end{corollary}

We can now upper bound the value of the output of \repeatedgreedy.
\begin{lemma} \label{lem:repeated_greedy_output_value}
	The value of the output set of \repeatedgreedy given $I_k$ is at most $1 + \frac{1}{\sqrt{k}}$.
\end{lemma}
\begin{proof}
	To prove the lemma, we need to show that every set that \repeatedgreedy might output according to Corollary~\ref{cor:possible_outputs} has a value of at most $1 + \frac{1}{\sqrt{k}}$. We do that by considering every possible type of such sets.
	\begin{itemize}
		\item $f_k(\varnothing) = 0$.
		\item For every integer $1 \leq i \leq \sqrt{k}$, $f_k(\{b_i\}) = 1 + \frac{3}{4\sqrt{k}} < 1 + \frac{1}{\sqrt{k}}$.
		\item For every integer $1 \leq i \leq \sqrt{k}$, $f_k(\{b_i, o_{i, 1 + \sqrt{k}}\}) = 1 + \frac{4}{4\sqrt{k}} = 1 + \frac{1}{\sqrt{k}}$.
		\item For every two integers $1 \leq i \leq \sqrt{k}$ and $1 \leq j \leq 1 + \sqrt{k}$, $f_k(\{o_{i, j}\}) = 1 + \frac{1}{4\sqrt{k}} < 1 + \frac{1}{\sqrt{k}}$.
		\item For every two integers $1 \leq i, j \leq \sqrt{k}$, $f_k(\{d_{i, j}\}) = 1 + \frac{2}{4\sqrt{k}} < 1 + \frac{1}{\sqrt{k}}$.
		\item For every two integers $1 \leq i, j \leq \sqrt{k}$, $f_k(\{o_{i, j}, d_{i, j}\}) = 1 + \frac{3}{4\sqrt{k}} < 1 + \frac{1}{\sqrt{k}}$. \qedhere
	\end{itemize} 
\end{proof}

To complete the proof of Theorem~\ref{th:repeated_bad_example}, it remains to observe that the ratio between the value of the optimal solution of $I_k$ (lower bounded by Observation~\ref{obs:optimal_value}) and the maximum value of a set that \repeatedgreedy can output (upper bounded by Lemma~\ref{lem:repeated_greedy_output_value}) is at least
\[
\frac{k + \frac{5\sqrt{k}}{4}}{1 + \frac{1}{\sqrt{k}}}
=
\frac{4k\sqrt{k} + 5k}{4\sqrt{k} + 4}
=
k + \frac{k}{4\sqrt{k} + 4}
=
k + \Omega(\sqrt{k})
\enspace.
\]

\subsection{Nearly Linear Time  with Knapsack Constraints}

In this section, we demonstrate how \repeatedgreedy may be modified to run in nearly linear time and achieve approximations for the more general problem \eqref{eq:opt-problem}, where there are $m$ additional knapsack constraints.
As before, we use the marginal gain thresholding technique of \cite{Badanidiyuru14} to ensure a nearly linear run time at the cost of an (arbitrarily) small multiplicative increase in the approximation factor.
To handle knapsack constraints, we use the density threshold technique of \cite{MBK16} with our improved binary search analysis.
These modification techniques are identical to those used in Sections~\ref{sec:fast-alg} and \ref{sec:knapsack-alg}, and so our discussion of them in this section is considerably shorter.

These modifications are made primarily in the greedy subroutine, which we present below as \modgreedy.
\modgreedy is similar to \greedy, but differs in two respects: rather than iteratively searching over all elements to find the one with largest marginal gain, the \modgreedy iteratively decreases marginal gain thresholds and accepts any element whose marginal gain is above the threshold \emph{and} whose density is above the density threshold.
The formal details of the implementation are given below as Algorithm~\ref{alg:modified-greedy}.

\begin{algorithm}[th]
	\caption{\modgreedy($\gnd, f, \cI, \density, \varepsilon$)} \label{alg:modified-greedy}
	Initialize solution $S_0 \gets \varnothing$ and iteration counter $i \gets 1$.\\
	Let $\maxgain = \max_{u \in \gnd} f(u)$, and initialize threshold $\threshold = \maxgain$. \\
	\While{$\threshold > (  \varepsilon / n)  \cdot \maxgain $}
	{
		\For{every element $u$ with $u \in \gnd$ such that $S_{i-1} + u \in \cI$}{
			\If{$f(u \mid S_{i-1}) \geq \max \left( \threshold, \density \cdot \sum_{r=1}^m c_r(u) \right)$ \label{line:mg-feasibility-check}}{
				\If{$c_r(S_{i-1} + u) \leq 1$ for all $1 \leq r \leq m$}{ \label{line:mg-knapsack-check}
					Let $u_i \gets u$ .\\
					Update the solution as $S_i \gets S_{i-1} + u$.\\
					Update the iteration counter $i \gets i+1$.
				}
			}	
		}
		Update marginal gain $\threshold \gets (1 -  \varepsilon) \cdot \threshold$.
	}
	Let $u^* \gets \argmax_{u \in \gnd} f(u)$.\\
	\Return{the set $S$ maximizing $f$ among the sets $S_i$ and $\{ u^*\}$}.
\end{algorithm}

The \modrepeatedgreedy algorithm iteratively calls \modgreedy to produce a solution $S_i$, runs an unconstrained submodular maximization (\usm) subroutine on $S_i$ to obtain $S'_i$, and then removes $S_i$ from the ground set.
Finally, \modrepeatedgreedy returns the best solution among the sequence $S_1, S'_1, \dots S_\ell, S'_\ell$ produced.
Note that this is similar to \repeatedgreedy, the only difference being that \modrepeatedgreedy calls \modgreedy rather than the vanilla greedy algorithm. 
We present \modrepeatedgreedy formally below as Algorithm~\ref{alg:repeated_greedy}.

\begin{algorithm}[H]
	\DontPrintSemicolon
	\caption{\modrepeatedgreedy($\gnd, f, \cI, \ell, \density, \varepsilon$)}\label{alg:mod-repeated_greedy}
	Let $\gnd_1 \gets \gnd$.\\
	\For{$i = 1$ \KwTo $\ell$}
	{
		Run modified greedy procedure $S_i \gets \modgreedy(\gnd_i, f, \cI, \density, \varepsilon)$ \\
		Filter the greedy solution $S'_i \gets \usm(S_i)$ \\
		Update ground set $\gnd_{i + 1} \gets \gnd_i \setminus S_i$.
	}
	\Return the set $S$ maximizing $f$ among the sets $\{S_i, S'_i \}_{i = 1}^\ell$.
\end{algorithm}

In the following observation, we bound the running time of \modrepeatedgreedy and prove feasibility of the returned solution.

\begin{observation}\label{obs:mod-repeated-greedy-runtime-and-feasibility}
	\modrepeatedgreedy requires $\bigO{\ell n / \eps}$ oracle calls, $\bigO{\ell n m / \eps}$ arithmetic operations, and its output $S$ is independent in $\cI$ and satisfies the knapsack constraints.
\end{observation}
\begin{proof}
These statements follow largely from analysis of \modgreedy.
Note that each iteration of the while loop of \modgreedy examines each element only once so that $\bigO{n}$ oracle calls and $\bigO{mn}$ arithmetic operations are required during each iteration.
As discussed in the proof of Observation~\ref{obs:fast-alg-runtime}, there are at most $\bigtO{1/ \eps}$ iterations of the while loop when the input error term satisfying $\eps < 1/2$.
Thus, \modgreedy requires a total of $\bigtO{n / \eps}$ oracle queries and $\bigtO{n m / \eps}$ arithmetic operations.
Moreover, by the acceptance criteria, the solution returned by \modgreedy is independent in $\cI$ and satisfies the knapsack constraints.

\modrepeatedgreedy makes $\ell$ calls to \modgreedy and to the \usm subroutine, which is assumed to run in linear time. 
Thus, the algorithm requires $\bigtO{\ell n / \eps}$ oracle calls and $\bigtO{ \ell n m / \eps}$ arithmetic operations.
Finally, the solution returned by \modrepeatedgreedy is feasible with respect to independence and knapsack constraints because all the outputs of \modgreedy and \usm have this property.
\end{proof}

The following proposition provides an approximation guarantee for the solution returned by \modrepeatedgreedy.
As in Section~\ref{sec:knapsack-alg}, we define $\knapevent$ to be an indicator variable which takes the value $1$ if the knapsack check in Line~\ref{line:mg-knapsack-check} of \modgreedy evaluates to false at any point in the execution of \modrepeatedgreedy and $0$ otherwise.

\newcommand{\modRGguarantees}[1][]{
	If $(\gnd, \cI)$ is a $k$-system, then the solution $S$ returned by \modrepeatedgreedy satisfies the following approximation guarantees.
	\ifthenelse{\isempty{#1}}{\renewcommand{\eqType}{equation}}{\renewcommand{\eqType}{equation*}}
	\begin{\eqType} \ifthenelse{\isempty{#1}}{\label{eq:modRG-case-analysis}}{\tag{\ref{eq:modRG-case-analysis}}}
		f(S) \geq 
		\left\{
		\begin{array}{lr}
			\frac{1}{2} \density  &\text{ if } \knapevent = 1 \enspace,\\
			\left( \frac{1 - \eps}{k + 1 + \alpha(\ell-1)/2} \right) \Big( (1 - 1 / \ell - \eps) f(\OPT) - \density m \Big) &\text{ if } \knapevent = 0 \enspace.
		\end{array}
		\right.
	\end{\eqType}
	Moreover, when $f$ is monotone, these approximation guarantees improve to 
	\begin{\eqType} \ifthenelse{\isempty{#1}}{\label{eq:modRG-case-analysis-monotone}}{\tag{\ref{eq:modRG-case-analysis-monotone}}}
		f(S) \geq 
		\left\{
		\begin{array}{lr}
			\frac{1}{2} \density  &\text{ if } \knapevent = 1 \enspace,\\
			\left(\frac{(1 - \eps) }{k + 1 + \alpha(\ell-1)/2}\right)  \Big( (1-\eps) f(\OPT) - \density m \Big) &\text{ if } \knapevent = 0 \enspace.
		\end{array}
		\right.
	\end{\eqType}
}
\begin{proposition} \label{prop:mod-repeated-greedy-guarantees}
	\modRGguarantees
\end{proposition}

The proof of Proposition~\ref{prop:mod-repeated-greedy-guarantees} is similar to that of Proposition~\ref{prop:repeated-greedy-analysis}, except that we cannot use the analysis of \cite{GRST10} for the approximation ratio of \greedy subject to a $k$-system.
Instead, we must use an analysis which takes into account the marginal gain and density thresholding techniques.
Because the main proof ideas involving the repeated greedy technique are presented in Section~\ref{sec:repeated-greedy} and the marginal gain threshold and density ratio threshold techniques already appear in existing works, we defer the proof of Proposition~\ref{prop:mod-repeated-greedy-guarantees} to Appendix~\ref{sec:mod-rg-proof}.

We would like to choose a density parameter $\density$ to maximize the lower bound on the returned objective value in Proposition~\ref{prop:mod-repeated-greedy-guarantees}.
As in Section~\ref{sec:knapsack-alg}, we propose a binary search approach on the density parameter $\density$ using the knapsack rejection indicator $\knapevent$.
We sketch this idea again here for completeness.

Suppose that the $f$ is a general non-monotone submodular function{\textemdash}the monotone case may be handled similarly.
Choosing the density parameter
$
\density^* 
= 2(1 - \eps) \left( \frac{1 - 1/\ell - \eps}{k + 2m + 1 + \alpha(\ell-1)/2} \right) 
\triangleq \beta \cdot f(\OPT)
$
approximately maximizes the lower bound \eqref{eq:modRG-case-analysis} presented in Proposition~\ref{prop:mod-repeated-greedy-guarantees}, which yields an approximation guarantee of
\[
f(S) \geq (1 -\eps) \left( \frac{1 - 1/\ell - \eps}{k+2m + 1 + \alpha(\ell-1)/2} \right) f(\OPT) \enspace.
\]
Although the term $\beta$ is known, we do not know the optimal objective value $f(\OPT)$ and so there is no way for us to know the value of the (approximately) optimal density parameter $\density^*$.
However, the submodularity of $f$ implies that the optimal objective value lies within the interval $\maxgain \leq f(\OPT) \leq r \cdot \maxgain$, and thus, the optimal density parameter $\density^*$ lies in the interval $ \beta \cdot \maxgain \leq \density^* \leq \beta \cdot (r \cdot \maxgain)$.
\cite{MBK16} proposed running a multiplicative grid search over this interval, where the repeated greedy algorithm is run with each $\density$ in this grid.
This grid search technique yields an approximation factor which is only $(1 + \delta)$ times larger than if we had used the optimal $\density^*$ and requires $\bigO{1/\delta}$ calls to the repeated greedy algorithm.
As in Section~\ref{sec:knapsack-alg}, we improve upon this technique by proposing a binary search method which uses the knapsack rejection indicator $\knapevent$.
The same approximation factor is achieved, but only $\bigtO{1}$ calls to the repeated greedy algorithm are required, which is an exponential decrease compared to the ``brute-force'' grid search.

We refer to this binary search routine for calling \modrepeatedgreedy with different values of the density parameter as \densitysearchRG.
This is formally outlined below as Algorithm~\ref{alg:density-search-rg}.

\begin{algorithm}[H]
	\caption{\densitysearchRG($\gnd, f, \cI, \ell, \delta, \varepsilon, \beta$)}\label{alg:density-search-rg}
	Initialize upper and lower bounds $k_{\ell} = 1$, $k_{u} = \lceil \frac{1}{\delta} \log n \rceil$. \\
	Let $\maxgain = \max_{u \in \gnd} f(u)$, and initialize iteration counter $i \gets 1$. \\
	\While{ $| k_u - k_\ell | > 1$}
	{
		Set middle bound $k_i = \left\lceil \frac{k_\ell + k_u }{2} \right\rceil$. \\
		Set density ratio $\density_i \gets \beta \cdot  \maxgain (1 + \delta)^{k_i}$.\\
		Obtain set $S_i \gets \modrepeatedgreedy(\gnd, f, \cI, \ell, \density_i, \varepsilon)$. \\
		\uIf{$E_i = 0$}{
			Increase lower bound $k_\ell \gets k_i$. \\
		}
		\Else{
			Decrease upper bound $k_u \gets k_i$. \\	
		}
		Update iteration counter $i \gets i+1$.
	}
	Set density ratio $\density_i \gets \beta \cdot \maxgain (1 + \delta)^{k_\ell}$.\\
	Obtain set $S_i \gets \modrepeatedgreedy(\gnd, f, \cI, \ell, \density_i, \varepsilon)$. \\
	\Return{the set $\retsol$ maximizing $f$ among the sets $\{ S_1, \dots S_T \}$}.
\end{algorithm}

The approximation guarantees of \densitysearchRG are given below in Proposition~\ref{prop:density-search-rg-guarantees}.
The proof is omitted because it is nearly identical to that of Proposition~\ref{prop:density-search-alg-result}.
In particular, the proof of Proposition~\ref{prop:density-search-rg-guarantees} uses Proposition~\ref{prop:mod-repeated-greedy-guarantees} in the same way that Proposition~\ref{prop:knapsack_general_result} is used in the proof of Proposition~\ref{prop:density-search-alg-result}.

\begin{proposition}\label{prop:density-search-rg-guarantees}
	\densitysearchRG makes $\bigtO{1}$ calls to \modrepeatedgreedy.
	If $(\gnd, \cI)$ is a $k$-system and $\beta = 2(1 - \eps) \left( \frac{1 - 1/\ell - \eps}{k + 2m + 1 + \alpha(\ell-1)/2} \right)$, then the solution $S$ returned by \densitysearchRG satisfies
	\begin{align*}
	f(S) 
	&\geq (1 -\delta) (1 -\eps) \left( \frac{1 - 1/\ell - \eps}{k+2m + 1 + \alpha(\ell-1)/2} \right) f(\OPT) \\
	&\geq (1 - \delta) (1 - 2 \eps)^2 \left( \frac{1 - 1/\ell }{k+2m + 1 + \alpha(\ell-1)/2} \right) f(\OPT) 
	\end{align*}
	when the number of iterations $\ell$ is at least $2$.
	Moreover, if $f$ is monotone and $\beta =   \frac{2(1 - \eps)^2}{k + 2m + 1 + \alpha(\ell-1)/2}$ then this lower bound improves to 
	\[
	f(S) 
	\geq (1-\delta) (1-\eps)^2 \left( \frac{1 }{k+2m + 1 + \alpha(\ell-1)/2} \right) f(\OPT) 
	\]
	for any number of iterations $\ell$.
\end{proposition}

By setting the number of solutions $\ell$ to maximize the lower bounds in Proposition~\ref{prop:density-search-rg-guarantees}, we may obtain the following result.
	
\begin{theorem}\label{thm:density-search-rg-guarantees}
	Suppose that $(\gnd, \cI)$ is a $k$-system, the number of iterations is set to $\ell = \lfloor 1 + \sqrt{2 (k+2m + 1)/\alpha} \rfloor$, and the two error terms are set to be equal (i.e., $\varepsilon = \delta \in (0, 1/2)$).
	Then, \densitysearchRG requires $\bigtO{\sqrt{k+m} \cdot n / \varepsilon}$ oracle calls and $\bigtO{\sqrt{k+m} \cdot m n / \varepsilon}$ arithmetic operations and produces a solution whose approximation ratio is at most
	\[
	(1 - 2\eps)^{-3} \Big[ k+ 2m + (\sqrt{2 \alpha}) \sqrt{k + 2m} + (\alpha + 1) + \littleO{1} \Big] \enspace.
	\]
	Moreover, when $f$ is non-negative monotone submodular and the number of iterations is set to $\ell = 1$, then the approximation ratio improves to $(1-\eps)^{-3} \left( k + 2m + 1 \right)$.
\end{theorem}

As discussed in Section~\ref{sec:knapsack-alg}, the $(1 - 2 \eps)^{-1}$ multiplicative factor can always be made into a $(1 + \eps')$ approximation factor by setting $\eps' = c \cdot \eps$ for some constant $c < 1$.
We omit the proof of Theorem~\ref{thm:density-search-rg-guarantees}, as the analysis is essentially the same as in the proof of Theorem~\ref{thm:repeated-greedy}, except that $k$ is replaced now by $k+2m$.
Indeed, one of the interpretations of Theorem~\ref{thm:density-search-rg-guarantees} is that the \modrepeatedgreedy algorithm achieves approximation guarantees similar to \repeatedgreedy, except that the independence parameter $k$ is replaced with $k+2m$ to account for the additional knapsack constraints, while running in nearly linear time.


\section{Hardness Results} \label{sec:hardness}

In this section, we present hardness results which complement our algorithmic contributions. 
In particular, we study the hardness of maximizing linear functions and monotone submodular functions over $k$-extendible systems.
These hardness results demonstrate that the approximation guarantees of \mainalg for submodular maximization over $k$-extendible systems are nearly optimal (up to low order terms) amongst all polynomial time algorithms.
We emphasize here that the following hardness results are information theoretic, and thus, independent of computational complexity hypotheses such as $P \neq NP$.
The first hardness result regards the approximability of maximizing a linear function over a $k$-extendible system.

\begin{theorem} \label{thm:linear_hardness}
	There is no polynomial time algorithm for maximizing a linear function over a $k$-extendible system that achieves an approximation ratio of $k - \eps$ for any constant $\eps > 0$.
\end{theorem}

The second hardness result regards the approximability of maximizing a monotone submodular function over a $k$-extendible system.

\begin{theorem} \label{thm:submodular_hardness}
	There is no polynomial time algorithm for maximizing a non-negative monotone submodular function over a $k$-extendible system that achieves an approximation ratio of $(1 - e^{-1/k})^{-1} - \eps$ for any constant $\eps > 0$.
\end{theorem}

Recall that \mainalg achieves an approximation ratio of $\nicefrac{(k+1)^2}{k} = k + 2 + \nicefrac{1}{k}$ for maximizing a submodular function over a $k$-extendible system and that this approximation ratio improves to $k + 1$ when the objective is monotone.
The hardness result of Theorem~\ref{thm:submodular_hardness} shows that achieving an approximation ratio better than $(1 - e^{-1/k})^{-1} - \eps \geq k + \nicefrac{1}{2} - \eps$ for monotone objectives requires exponentially many queries to the value and independence oracles.
Hence, the gap between the achieves approximation ratio and the hardness result is a small constant. 
In this sense, the approximation achieved by \mainalg and its variants for maximizing over a $k$-extendible system is near-optimal amongst all algorithms which query the oracles polynomially many times.

\mainalg has an approximation guarantee of $k + \bigO{\sqrt{k}}$ for the more general class of $k$-systems, and so it is natural to wonder whether this approximation is also near-optimal amongst polynomial time algorithms.
Since $k$-extendible systems are a subclass of $k$-systems, the hardness results presented here also apply to $k$-systems; however, the gap between the  $k + \bigO{\sqrt{k}}$ approximation and the $k + \nicefrac{1}{2} - \eps$ hardness is larger in this case.
Indeed, it is an open question whether the additive $\bigO{\sqrt{k}}$ term is necessary for any polynomial time algorithm which maximizes a non-monotone submodular objective over a $k$-system or whether this factor may be improved.

The proof of Theorems~\ref{thm:linear_hardness} and~\ref{thm:submodular_hardness} consists of two steps which are organized into two respective sections. 
In Section~\ref{sec:linear-hardness}, we define two $k$-extendible systems which are indistinguishable in polynomial time. 
The inapproximability result for linear objectives follows from the indistinguishability of these systems and the fact that the sizes of their maximal sets are very different. 
In Section~\ref{sec:submodular-hardness}, we define monotone submodular objective functions for the two $k$-extendible systems. 
Using the symmetry gap technique of~\cite{V13}, we will show that these objective functions are also indistinguishable, despite being different. 
Then, we will use the differences between the objective functions to prove the slightly stronger inapproximability result for monotone submodular objectives.

\subsection{Hardness for linear functions over \texorpdfstring{$k$}{k}-extendible systems} \label{sec:linear-hardness}
In this section, we construct two $k$-extendible systems which, after a random permutation is applied to the ground set, are indistinguishable using polynomially many queries with high probability.
Moreover, the size of the largest base is significantly different between these two systems.

First, we construct a $k$-extendible system $\cM(k, h, m) = (\cN_{k, h, m}, \cI_{k, h, m})$, which is parameterized by three positive integers $k, h$ and $m$ such that $h$ is an integer multiple of $2k$.
The ground set of the system consists of $h$ groups of elements, each of size $km$.
More formally, the ground set is $\cN_{k, h, m} = \cup_{i = 1}^h H_i(k,m)$, where $H_i(k, m) = \{u_{i, j} ~|~ 1 \leq j \leq km\}$. 
A set $S \subseteq \cN_{k, h, m}$ is independent if and only if it obeys the following inequality:
\[
g(|S \cap H_1(k, m)|) + |S \setminus H_1(k, m)|
\leq
m
\enspace,
\]
where the function $g$ is a piece-wise linear function defined by
\[
g(x) 
= \min\left\{x, \frac{2km}{h}\right\} + \max\left\{\frac{x - 2km/h}{k}, 0\right\}
= \left\{
\begin{array}{lr}
x &\text{ if } x \leq \nicefrac{2km}{h}\\
\frac{x}{k} + (1 - \frac{1}{k}) \frac{2km}{h} &\text{ if } x \geq \nicefrac{2km}{h}
\end{array}
\right.
\enspace.
\]
Intuitively, a set is independent if its elements do not take too many ``resources'', where most elements requires a unit of resources, but elements of $H_1(k, m)$ take only $1/k$ unit of resources each once there are enough of them. 
Consequently, the only way to get a large independent set is to pack many $H_1(k, m)$ elements.

For notational clarity, we drop the reference to the underlying parameters $k$, $h$, and $m$ in the definition of the set systems throughout the rest of the section.
That is, we write $\cM$ and $H_i$ instead of the more burdensome $\cM(k,h,m)$ and $H_i(k,m)$. 

\begin{lemma}
	For every choice of $h$ and $m$, $\cM$ is a $k$-extendible system.
\end{lemma}
\begin{proof}
	First, observe that $g(x)$ is a monotone function, and therefore, a subset of an independent set of $\cM$ is also independent. 
	Also, $g(0) = 0$, and therefore, $\varnothing \in \cI$. 
	This proves that $\cM$ is an independence system. 
	In the rest of the proof we show that it is also $k$-extendible.
	
	Consider an arbitrary independent set $C \in \cI$, an independent extension $D$ of $C$ and an element $u \not \in D$ for which $C + u \in \cI$. 
	We need to find a subset $Y \subseteq D \setminus C$ of size at most $k$ such that $D \setminus Y + u \in \cI$.
	If $|D \setminus C| \leq k$, then we can simply pick $Y = D \setminus C$. 
	Thus, we can assume from now on that $|D \setminus C| > k$. Let $$\Sigma(S) = g(|S \cap H_1) + |S \setminus H_1|.$$
	By definition, $\Sigma(D) \leq m$ because $D \in \cI$. 
	Observe that $g(x)$ has the property that for every $x \geq 0$, $$k^{-1} \leq g(x + 1) - g(x) \leq 1.$$ Thus, $\Sigma(S)$ increases by at most $1$ every time that we add an element to $S$, but decreases by at least $1/k$ every time that we remove an element from $S$. 
	Hence, if we let $Y$ be an arbitrary subset of $D \setminus C$ of size $k$, then
	\[
	\Sigma(D \setminus Y + u)
	\leq
	\Sigma(D) - \frac{|Y|}{k} + 1
	=
	\Sigma(D)
	\leq
	m
	\enspace,
	\]
	which implies that $D \setminus Y + u \in \cI$.
\end{proof}\vspace{0pt}


Let us now show that $\cM$ contains a large independent set.

\begin{observation} \label{ob:large_independent_set}
	$\cM$ contains an independent set whose size is $k(m - 2km/h) + 2km/h \geq mk(1-2k/h)$. Moreover, there is such set in which all elements belong to $H_1$.
\end{observation}
\begin{proof}
	Let $s = k(m - 2km/h) + 2km/h$, and consider the set $S = \{u_{1, j} \mid 1 \leq j \leq s\}$. 
	This is a subset of $H_1 \subseteq \cN$ since $s \leq km$. Also,
	\begin{align*}
	g(|S|)
	={} &
	g(s) \\
	={}&
	\min\left\{s, \frac{2km}{h}\right\} + \max\left\{\frac{s - 2km/h}{k}, 0\right\}\\
	\leq{} &
	\frac{2km}{h} + \max\left\{\frac{[k(m - 2km/h) + 2km/h] - 2km/h}{k}, 0\right\}\\
	={} &
	\frac{2km}{h} + \max\left\{m - \frac{2km}{h}, 0\right\}
	=
	m
	\enspace.
	\end{align*}
	Since $S$ contains only elements of $H_1$, its independence follows from the above inequality.
\end{proof}\vspace{0pt}

Let us now define our second $k$-extendible system $\cM' = (\cN, \cI')$. 
The ground set of this system is the same as the ground set of $\cM$, but a set $S \subseteq \cN$ is considered independent in this independence system if and only if its size is at most $m$. 
Clearly, this is a $k$-extendible system (in fact, it is a uniform matroid). 
Moreover, note that the ratio between the sizes of the maximal sets in $\cM$ and $\cM'$ is at least
\[
\frac{mk(1 - 2k/h)}{m}
=
k(1 - 2k/h)
\enspace.
\]
Our plan is to show that it takes exponential time to distinguish between the systems $\cM$ and $\cM'$, and thus, no polynomial time algorithm can provide an approximation ratio better than this ratio for the problem of maximizing the cardinality function (i.e., the function $f(S) = |S|$) subject to a $k$-extendible system constraint.

Consider a polynomial time deterministic algorithm that gets either $\cM$ or $\cM'$ after a random permutation was applied to the ground set. 
We prove below that with high probability the algorithm fails to distinguish between the two possible inputs. 
Notice that by Yao's lemma, this implies that for every random algorithm there exists a permutation for which the algorithms fails with high probability to distinguish between the inputs.

Assuming our deterministic algorithm gets $\cM'$, it checks the independence of a polynomial collection of sets. 
Observe that the sets in this collection do not depend on the permutation because the independence of a set in $\cM'$ depends only on its size, and thus, the algorithm will take the same execution path given every permutation. 
If the same algorithm now gets $\cM$ instead, it will start checking the independence of the same sets until it will either get a different answer for one of the checks (different than what is expected for $\cM'$) or it will finish all the checks. 
Note that in the later case the algorithm must return the same answer that it would have returned had it been given $\cM'$. Thus, it is enough to upper bound the probability that any given check made by the algorithm will result in a different answer given the inputs $\cM$ and $\cM'$.

\begin{lemma} \label{le:independence_different_unlikely}
	Following the application of the random ground set permutation, the probability that a set $S$ is independent in $\cM$ but not in $\cM'$, or vice versa, is at most $e^{-\frac{2km}{h^2}}$.
\end{lemma}
\begin{proof}
	Observe that as long as we consider a single set, applying the permutation to the ground set is equivalent to replacing $S$ with a random set of the same size. 
	So, we are interested in the independence in $\cM$ and $\cM'$ of a random set of size $|S|$. 
	If $|S| > km$, then the set is never independent in either $\cM$ or $\cM'$, and if $|S| \leq m$, then the set is always independent in both $\cM$ and $\cM'$. 
	Thus, the interesting case is when $m < |S| \leq km$.
	
	Let $X = |S \cap H_1|$. 
	Notice that $X$ has a hypergeometric distribution, and $\mathbb{E}[X] = |S|/h$. 
	Thus, using bounds given in \citep{Skala13} (these bounds are based on results of \citep{Chvatal79,Hoeffding1963}), we get
	\[
	\Pr\left[X \geq \frac{2km}{h}\right]
	\leq
	\Pr\left[X \geq \bE[|X|] + \frac{km}{h}\right]
	\leq
	e^{-2\left(\frac{km/h}{|S|}\right)^2 \cdot |S|}
	=
	e^{-\frac{2k^2m^2}{h^2 \cdot |S|}}
	\leq
	e^{-\frac{2km}{h^2}}
	\enspace.
	\]
	The lemma now follows by observing that $X \leq 2km/h$ implies that $S$ is a dependent set under both $\cM$ and $\cM'$.
\end{proof}\vspace{0pt}

We now think of $m$ as going to infinity and of $h$ and $k$ as constants. 
Notice that given this point of view the size of the ground set $\cN$ is $nkh = O(m)$. 
Thus, the last lemma implies, via the union bound, that with high probability an algorithm making a polynomial number (in the size of the ground set) of independence checks will not be able to distinguishes between the cases in which it gets as input $\cM$ or $\cM'$.

Using the above results, we are now ready to prove Theorem~\ref{thm:linear_hardness}.

\begin{proof}[Proof of Theorem~\ref{thm:linear_hardness}]
	Consider an algorithm that needs to maximize the cardinality function over the $k$-extendible system $\cM$ after the random permutation was applied, and let $T$ be its output set. 
	Notice that $T$ must be independent in $\cM$, and thus, its size is always upper bounded by $mk$. 
	Moreover, since the algorithm fails, with high probability, to distinguish between $\cM$ and $\cM'$, $T$ is with high probability also independent in $\cM'$, and thus, has a size of at most $m$. Therefore, the expected size of $T$ cannot be larger than $m + o(1)$ (formally, this $o(1)$ terms represents an expression that goes to $0$ as $m$ increases for any given choice of $k$ and $h$).
	
	On the other hand, Lemma~\ref{ob:large_independent_set} shows that $\cM$ contains an independent set of size at least $mk(1 - 2k/h)$. 
	Thus, the approximation ratio of the algorithm is no better than
	\[
	\frac{mk(1 - 2k/h)}{m + o(1)}
	\geq
	\frac{mk(1 - 2k/h)}{m} - \frac{k}{m} o(1)
	=
	k - 2k^2/h - o(1)
	\enspace.
	\]
	Choosing a large enough $h$ (compared to $k$), we can make this approximation ratio larger than $k - \varepsilon$ for any constant $\varepsilon > 0$.
\end{proof} \vspace{0in}

\subsection{Hardness for submodular functions over \texorpdfstring{$k$}{k}-extendible systems} \label{sec:submodular-hardness}

In this section, we prove Theorem~\ref{thm:submodular_hardness}, which is a stronger inapproximability result for maximizing monotone submodular functions over $k$-extendible systems.
As in the previous section, we construct two problem instances which have different optimal values but{\textemdash}after a permutation of the ground set{\textemdash}are indistinguishable using only polynomially many oracle queries.
We use again the two $k$-extendible systems $\cM$ and $\cM'$, parametrized by $h$, $m$ and $k$, defined in the last section.
The additional technical construction of this section is two submodular functions $f$ and $g$ which take very different optimal values over the two extendible systems.
To construct these two submodular functions, we use the symmetry gap technique developed by \cite{V13}.

The symmetry gap technique is a general method for constructing hard instances for submodular optimization; however, we present a relatively self-contained version of this technique, appealing only to the main technical construction from \cite{V13}.
The rest of this paragraph is a high level roadmap for the construction of the submodular functions $f$ and $g$ based on the symmetry gap technique.
Recall that the common ground set $\cN$ of our independence systems is the union of $h$ disjoint sets of elements, i.e., $\cN = \cup_{i = 1}^h H_i$.
We begin the construction of $f$ and $g$ by defining an initial function $q\colon 2^{[h]} \rightarrow \nnR$, which assigns a non-negative value to each subset $X$ of $[h]$. 
A key aspect of this initial function is that it has a desired symmetry property. 
In our context, that means that the value of $q$ depends only on the cardinality of its input, i.e., $|X|$.
Next, we consider the multilinear extension of the initial function, which is denoted by $Q\colon [0,1]^h \rightarrow \nnR$ and is defined as
\[
Q(x) = \sum_{X \subseteq [h]} q(X) \prod_{i \in X} x_i \prod_{i \notin X} (1 - x_i) \enspace.
\]
Note that $Q$ is a function on vectors of the hypercube $[0, 1]^h$.
The next step is to apply a well-chosen perturbation to this multilinear extension $Q$ to obtain a function  $F\colon [0,1]^h \rightarrow \nnR$ and then symmetrize $F$ to obtain a second function $G\colon [0,1]^h \rightarrow \nnR$.
At this point, we construct the desired set functions $f$ and $g$ on the original ground set $\cN$ by mapping sets $S \in \cN$ to vectors $x \in [0,1]^h$.
This mapping depends on the number of elements of $S$ that are in each of the partitions $H_1 \dots H_h$ of the ground set.
Specifically, we define the mapping as
\[
x(S) = \left( \frac{| S \cap H_1 |}{|H_1|}, \frac{|S \cap H_2|}{|H_2|} , \dots \frac{|S \cap H_h|}{|H_h|}  \right) \enspace,
\]
and the set functions $f$ and $g$ are defined as $f(S) = F(x(S))$ and $g(S) = G(x(S))$, respectively.
A diagram that summarizes this construction appears in Figure~\ref{fig:hardness-construction}.
The key technical lemma of \cite{V13} shows that, by picking an appropriate perturbation and symmetrization method, we can guarantee that the set functions $f$ and $g$ are both monotone submodular, $g$ depends only on the cardinality of its input, and yet the two functions have similar values on most inputs.

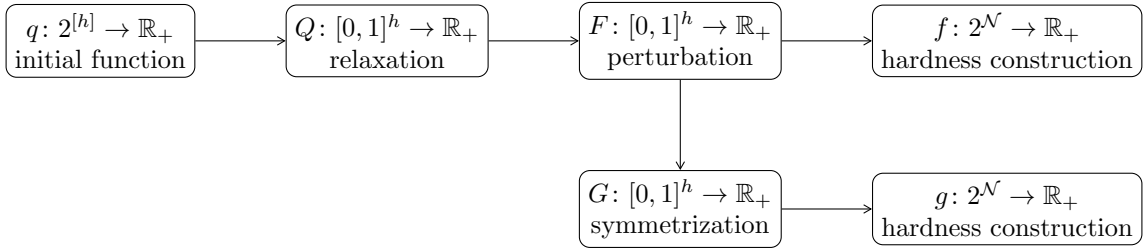
\begin{figure} 
	\centering
	\begin{tikzpicture}[
	node distance = 12mm and 12mm, 
	box/.style = {draw, rounded corners, 
		minimum width=22mm, minimum height=5mm, align=center},
	> = {Straight Barb[angle=60:2pt 3]},
	bend angle = 15,
	auto = right,
	]
	\node (n1)  [box] {$q\colon 2^{[h]} \rightarrow \nnR$\\ initial function};
	\node (n2)  [box, right=of n1]    {$Q\colon [0,1]^h \rightarrow \nnR$ \\ relaxation };
	\node (n3) [box, right=of n2] {$F\colon [0,1]^h \rightarrow \nnR$ \\ perturbation };
	\node (n4) [box, below=of n3] {$G\colon [0,1]^h \rightarrow \nnR$ \\ symmetrization};
	\node (n5) [box, right=of n3] {$f\colon 2^\gnd \rightarrow \nnR$ \\ hardness construction};
	\node (n6) [box, right=of n4] {$g\colon 2^\gnd \rightarrow \nnR$\\ hardness construction};
	\draw[->] (n1) to  (n2);
	\draw[->] (n2) to (n3);
	\draw[->] (n3) to (n4);
	\draw[->] (n3) to (n5);
	\draw[->] (n4) to (n6);
	\end{tikzpicture}
	\caption{Construction of the objective functions $f$ and $g$. \label{fig:hardness-construction}}
\end{figure}

Now with the roadmap complete, we begin our construction of the functions $f$ and $g$.
We define the initial function $q\colon 2^{[h]} \rightarrow \nnR$ as
\[
q(X)
=
\min\{|X|, 1\}
\enspace.
\]

Let $Q: [0, 1]^{h} \to \nnR$ be the mutlilinear extension of $q$. One can observe that $Q(x) = 1 - \prod_{i \in [h]} (1 - x_i)$ for every vector $x \in [0,1]^h$. Furthermore, for every such vector $x$, we define its \emph{symmetrization} as $\bar{x} = (\| x \|_1 / h) \cdot \characteristic_{[h]}$ ($\characteristic_{[h]}$ represents here the all ones vector in $[0, 1]^h$). 
The next lemma is a direct application of Lemma~3.2 in \citep{V13}, but simplified for our setting.
It follows from the symmetry in the initial function $q$, namely that it is invariant under any permutation of the elements of $[h]$.

\begin{lemma}\label{lemma:perturbation-symmetrization}
For every $\eps' > 0$ there exists $\delta_h > 0$ and two functions $F, G : [0,1]^{h} \rightarrow \nnR$ with the following properties.
\begin{itemize}
	\item For all $x \in [0, 1]^{h}$, $|F(x) - Q(x)| \leq \eps'$.
	\item For all $x \in [0, 1]^{h}$: $G(x) = F(\bar{x})$.
	\item Whenever $\norm{x - \bar{x}}^2_2 \leq \delta_h$, $F(x) = G(x)$.
	\item The first partial derivatives of $F$ and $G$ are absolutely continuous,  $\frac{\partial F}{\partial x_u}, \frac{\partial G}{\partial x_u} \geq 0$ everywhere for every $u \in [h]$, and $\frac{\partial^2 F}{\partial x_u \partial x_v}, \frac{\partial^2 G}{\partial x_u \partial x_v} \leq 0$ 
	almost everywhere for every pair $u,v \in [h]$.
\end{itemize}
\end{lemma}
The first property formally states the sense in which $F$ is a perturbation of $Q$; namely, that their values differ only by an $\eps'$ amount for all vectors in the unit cube.
The second property formally states the sense in which the function $G$ is a symmetrization of $F$; namely, that evaluation of $G$ at $x$ is obtained by evaluating $F$ at the symmetrization $\bar{x}$.
The third property states that the two functions are equal on input vectors which are nearly symmetrized.
The final property is used below to show the monotonicity and submodularity of the set functions $f$ and $g$ (which are more formally constructed below).

Recall that the mapping from sets $S \subseteq \gnd$ to vectors $x \in [0,1]^h$ is defined using the partition of $\gnd$ into the sets $H_1, H_2, \dots H_h$ as
\[
x(S) 
= \left( \frac{| S \cap H_1 |}{|H_1|}, \frac{|S \cap H_2|}{|H_2|} , \dots \frac{|S \cap H_h|}{|H_h|}  \right)
\enspace.
\]
The set functions are defined as $f(S) = F(x(S))$ and $g(S) = G(x(S))$. The next lemma uses the final property of Lemma~\ref{lemma:perturbation-symmetrization} to argue that $f$ and $g$ are both monotone and submodular. Its proof is essentially identical to the proof of Lemma~3.1 of \citep{V13}. Nevertheless, we include it here for completeness.
\begin{lemma}
The set functions $f$ and $g$ defined as above are monotone and submodular.
\end{lemma}
\begin{proof}
We only show that $f$ is monotone submodular, as the proof that $g$ is monotone submodular is identical.
We begin by showing the monotonicity of $f$.
To this end, the main technical condition on $F$ that we need is:
\begin{equation}\label{eq:monotone_F}
F(w) \leq F(y) \quad \text{for all } w, y \in [0,1]^h \text{ satisfying } w \preceq y \enspace,
\end{equation}
where $\preceq$ denotes the component-wise partial ordering. To see that this condition holds, 
consider the line segment between $w$ and $y$, given by $v(t) = (1 - t) \cdot w + t \cdot y$, where $t \in [0,1]$ is a parametrization of the line segment.
Since $w \preceq y$, we have that the coordinates of $y - w$ are all non-negative.
Additionally, by Lemma~\ref{lemma:perturbation-symmetrization}, we have that $F$ is differentiable everywhere and $\frac{\partial F}{\partial x_u} \geq 0$ everywhere for every $u \in [h]$.
This means that each coordinate of the gradient $\nabla F(v(t))$ is non-negative for each $t \in [0,1]$.
Thus, the following inner product is non-negative:
\[
\langle \nabla F(v(t)) , y - w \rangle \geq 0 \quad \text{for all } t \in [0,1] \enspace.
\]
Using this and the fundamental theorem of calculus, we obtain
\[
F(y) - F(w) = \int_{t=0}^{1} \langle \nabla F(v(t)) , y - w \rangle dt \geq 0 \enspace.
\]
Monotonicity of $f$ now follows by observing that for sets $A \subseteq B$, the corresponding vectors satisfy $x(A) \preceq x(B)$, and hence, using \eqref{eq:monotone_F}, we have
\[
f(A) = F(x(A)) \leq F(x(B)) = f(B) \enspace.
\]

Next, we show that $f$ is submodular. 
To this end, the main technical condition that we need on $F$ is that for all $w, y \in [0,1]^h$ and $z \succeq 0$ satisfying $w \preceq y$ and $w + z, y + z \in [0,1]^h$, 
\begin{equation}\label{eq:F_submodular}
F(w + z) - F(w) \geq F(y + z) - F(y) \enspace.
\end{equation}
To see that this technical condition holds, consider the line segment between $w$ and $w + z$, which is given by $w(t) = w  + t \cdot z$, where $t \in [0,1]$ is a parameterization of this line segment.
Similarly, $y(t) = y + t \cdot z$ is the line segment between $y$ and $y + z$.
Since $w + z, y + z \in [0,1]^h$, all points on these two line segments are also in the unit cube and so the function $F$, its gradient $\nabla F$ and its Hessian $\nabla^2 F$ are all well-defined on these points.
Additionally, note that for each $t$, $w(t) \preceq y(t)$ so that the vector $y(t) - w(t)$ has non-negative coordinates.
Additionally, Lemma~\ref{lemma:perturbation-symmetrization} states that $\frac{\partial^2 F}{\partial x_u \partial x_v} \leq 0$ almost everywhere for every pair $u,v \in [h]$.
This means that the Hessian matrix $\nabla^2 F$ has non-positive entries, and thus, $$(y(t) - w(t))^T \left[ \nabla^2 F(v) \right] (y(t) - w(t)) \leq 0$$ for almost all $t \in [0,1]$.
Define now $$v_t(s) = (1-s) \cdot w(t) + s \cdot y(t)$$ to be the line segment between $w(t)$ and $y(t)$, and for notation convenience, define the function $H(v) = \langle \nabla F(v), z \rangle$. 
Using the fundamental theorem of calculus along with the chain rule and the properties of the Hessian matrix $\nabla^2 F$ above, the above observations yield
\begin{align*}
\langle \nabla F(y(t)), z \rangle  - \langle \nabla F(w(t)), z \rangle 
&= H(y(t)) - H(w(t)) \\
&= \int_{s=0}^{1} \langle \nabla H( v_t(s) ), y(t) - w(t) \rangle ds 
	&\text{(f.t. of calculus)}\\
&= \int_{s=0}^{1} (y(t)-w(t))^T \left[ \nabla^2 F ( v_t(s) ) \right] (y(t) - w(t))  ds
	&\text{(chain rule)} \\
&\leq 0
\enspace,
\end{align*}
so that $\langle \nabla F(w(t)) , z \rangle \geq \langle \nabla F(y(t)), z \rangle $ for all $t \in [0,1]$.
To prove \eqref{eq:F_submodular}, it only remains to combine the last result with the fundamental theorem of calculus and obtain
\begin{align*}
F(w + z) - F(w) 
= \int_{t=0}^{1} \langle \nabla F(w(t)), t \cdot z \rangle dt  
\geq
 \int_{t=0}^{1} \langle \nabla F(y(t)), t \cdot z \rangle dt 
= F(y + z) - F(y)
\enspace.
\end{align*}
Now we can use Inequality~\eqref{eq:F_submodular} to prove submodularity of $f$.
First, observe that for sets $A \subseteq B$, the corresponding vectors $x(A)$ and $x(B)$ satisfy $x(A) \preceq x(B)$.
Moreover, for any element $e \notin B$, $x(A + e) = x(A) + x(e)$ and $x(B + e) = x(B) = x(e)$.
Thus, setting $w = x(A)$, $y= x(B)$ and $z = x(e)$, we have that
\begin{align*}
f(A + e) - f(A) 
&= F(x(A) + x(e)) - F(x(A)) \\
&\geq F(x(B) + x(e)) - F(x(B)) \\
&= f(B + e) - f(B) \enspace,
\end{align*}
which establishes the submodularity of $f$.
\end{proof}

To define our problem instances, we associate the monotone submodular objective $f$ with the $k$-extendible system $\cM$ and the monotone submodular objective $g$ with the $k$-extendible systems $\cM'$.
Let us now bound the maximum values of the resulting submodular optimization problems.

\begin{lemma} \label{le:submodular_good_opt}
	The maximum value of an independent set in $\cM$ with respect to the objective $f$ is at least $1 - 2k/h - \varepsilon'$, and no more than $1 + \varepsilon'$.
\end{lemma}
\begin{proof}
	Observation~\ref{ob:large_independent_set} guarantees the existence of an independent set $S \subseteq H_1$ in $\cM$ of size $s \geq k(m - 2km/h)$. 
	Using the first property of Lemma~\ref{lemma:perturbation-symmetrization} and evaluating $Q$ at the corresponding vector $x(S)$, we have that the objective value associated with this set is
	\[
	f(S) = 
	F(x(S))
	\geq
	Q(x(S)) - \varepsilon'
	=
	\frac{s}{km} - \varepsilon'
	\geq
	\frac{k(m - 2km/h)}{km} - \varepsilon'\\
	=
	1 - 2k/h - \varepsilon'
	\enspace.
	\]
	This completes the proof of the first part of the lemma. To see that the second part also holds, we observe that $q$ (and therefore, also $Q$) never takes values larger than $1$; and thus, by the first property of Lemma~\ref{lemma:perturbation-symmetrization}, for every set $S' \subseteq \cN$,
	\[
		f(S')
		=
		F(x(S'))
		\leq
		Q(x(S')) + \varepsilon'
		\leq
		1 + \varepsilon'
		\enspace.
		\qedhere
	\]
\end{proof}

\begin{lemma} \label{le:submodular_bad_opt}
	The maximum value of a set in $\cM'$ with respect to the objective $g$ is at most $1 - e^{-1/k} + h^{-1} + \varepsilon'$.
\end{lemma}
\begin{proof}
	The objective $g$ is monotone, and thus, the maximum value set in $\cM'$ must be of size $m$. Using the second and first properties of Lemma~\ref{lemma:perturbation-symmetrization} and evaluating $Q$, we have that for every set $S$ of size $m$, we get that
	\begin{align*}
	g(S)
	={}&
	G(x(S))\\
	={} &
	F(\overline{x(S)})\\
	={}&
	F((kh)^{-1} \cdot \characteristic_{h})
	\leq
	Q((kh)^{-1} \cdot \characteristic_{h}) + \varepsilon'\\
	={} &
	1 - \left(1 - \frac{1}{kh}\right)^h + \varepsilon'\\
	\leq{}&
	1 - e^{-1/k}\left(1 - \frac{1}{k^2h}\right) + \varepsilon'\\
	\leq{}&
	1 - e^{-1/k} + h^{-1} + \varepsilon'
	\enspace.
	\qedhere
	\end{align*}
\end{proof}

As before, our plan is to show that after a random permutation is applied to the ground set it is difficult to distinguish between the problem instances $f$ with $\cM$ and $g$ with $\cM'$. 
This will give us an inapproximability result which is roughly equal to the ratio between the bounds given by the last two lemmata.

Observe that Lemma~\ref{le:independence_different_unlikely} holds regardless of the objective function. Thus, $\cM$ and $\cM'$ are still polynomially indistinguishable. 
Additionally, the next lemma shows that their associated objective functions are also polynomially indistinguishable.

\begin{lemma} \label{le:value_different_unlikely}
	Following the application of the random ground set permutation, the probability that any given set $S$ gets two different values under the two possible objective functions is at most $2h \cdot e^{-2mk\delta_h /h^2}$.
\end{lemma}
\begin{proof}
	We begin by showing that $f(S) = g(S)$ for all sets $S$ which are made up of roughly the same number of elements from each of the partitions $H_1, H_2, \dots H_h$.
	More precisely, define $X_i = |S \cap H_i(k, m)|$. 
	We claim that if a set $S$ satisfies $\left|X_i - \frac{|S|}{h}\right| < mk \cdot \sqrt{\frac{\delta_h}{h}}$ for every $1 \leq i \leq h$, then $f(S) = g(S)$.
	Note that under this condition, the norm of the difference between $x(S)$ and its symmetrization $\overline{x(S)}$ is at most
	\[
	\norm{y(S) - \overline{y(S)}}_2^2
	=
	\sum_{i = 1}^h (y_i(S) - \overline{y_i(S)})^2
	<
	\sum_{i = 1}^h \left(\sqrt{\frac{\delta_h}{h}}\right)^2
	=
	\sum_{i = 1}^h \frac{\delta_h}{h}
	=
	\delta_h
	\enspace,
	\]
	and thus by the third property of Lemma~\ref{lemma:perturbation-symmetrization}, we have that $F(x(S)) = G(x(S))$, which implies that $f(S) = g(S)$ by construction of these set functions.
	
	We now show that following the application of a random permutation to the ground set, any given set $S$ satisfies $\left|X_i - \frac{|S|}{h}\right| < mk \cdot \sqrt{\frac{\delta_h}{h}}$ for every $1 \leq i \leq h$ with high probability, and thus, $f(S) = g(S)$ with high probability.
	Recall that, as long as we consider a single set $S$, applying the permutation to the ground set is equivalent to replacing $S$ with a random set of the same size. 
	Hence, we are interested in the value under the two objective functions of a random set of size $|S|$. 
	Since $X_i$ has the a hypergeometric distribution, the bound of \citep{Skala13} gives us
	\begin{align*}
	\Pr\left[X_i \geq \frac{|S|}{h} + mk \cdot \sqrt{\frac{\delta_h}{h}} \right]
	={} &
	\Pr\left[X_i \geq \mathbb{E}[X_i] + mk \cdot \sqrt{\frac{\delta_h}{h}} \right]\\
	\leq{} &
	e^{-2 \cdot \left(\frac{mk \cdot \sqrt{\delta_h/h}}{|S|}\right)^2 \cdot |S|}\\
	={}&
	e^{-\frac{2\delta_h}{h} \cdot \frac{m^2k^2}{|S|}}\\
	\leq{}&
	e^{-2mk\delta_h/h^2}
	\enspace.
	\end{align*}
	Similarly, we also get
	\[
	\Pr\left[X_i \leq \frac{|S|}{h} - mk \cdot \sqrt{\frac{\delta_h}{h}} \right]
	\leq
	e^{-2mk\delta_h /h^2}
	\enspace.
	\]
	Combining both inequalities using the union bound now yields
	\[
	\Pr\left[\left|X_i - \frac{|S|}{h}\right| \geq mk \cdot \sqrt{\frac{\delta_h}{h}} \right]
	\leq
	2 e^{-2mk\delta_h /h^2}
	\enspace.
	\]
	Using the union bound again, the probability that $\left|X_i - \frac{|S|}{h}\right| \geq mk \cdot \sqrt{\frac{\delta_h}{h}}$ for any $1 \leq i \leq h$ is at most $2h \cdot e^{-2mk\delta_h /h^2}$. 
	It follows now from the first part of the proof that $f(S) \neq g(S)$ with probability at most $2h \cdot e^{-2mk\delta_h /h^2}$.	
%
\end{proof}\vspace{0pt}

Consider a polynomial time deterministic algorithm that gets either $\cM$ with its corresponding objective $f$ or $\cM'$ with its corresponding objective $g$ after a random permutation was applied to the ground set. 
Consider first the case that the algorithm gets $\cM'$ and its corresponding objective $g$. 
In this case, the algorithm checks the independence and value of a polynomial collection of sets.
We may assume, without loss of generality, that the algorithm checks both the value and independence oracles for every set that it checks. 
As before, one can observe that the sets which are queried do not depend on the permutation because the independence of a set in $\cM'$ and its value with respect to $g$ depend only on the set's size, which guarantees that the algorithm takes the same execution path given every permutation. 
If the same algorithm now gets $\cM$ instead, it will start checking the independence and values of the same sets until it will either get a different answer for one of the oracle queries (different than what is expected for $\cM'$) or it will finish all the queries. 
Note that in the later case the algorithm must return the same answer that it would have returned had it been given $\cM'$.

By the union bound, Lemmata~\ref{le:independence_different_unlikely} and~\ref{le:value_different_unlikely} imply that the probability that any of the sets whose value or independence is checked by the algorithm will result in a different answer for the two inputs decreases exponentially in $m$, and thus, with high probability the algorithm fails to distinguish between the inputs, and returns the same output for both. Moreover, note that by Yao's principal this observation extends also to polynomial time randomized algorithms. 

Using these ideas, we are now ready to prove Theorem~\ref{thm:submodular_hardness}.

\begin{proof}[Proof of Theorem~\ref{thm:submodular_hardness}]
	Consider an algorithm that seeks to maximize $f(S)$ over the $k$-extendible system $\cM$ after the random permutation was applied, and let $T$ be its output set. 
	Moreover, the algorithm fails, with high probability, to distinguish between $\cM$ and $\cM'$. 
	Thus, with high probability $T$ is independent in $\cM'$ and has the same value under both objective functions $f$ and $g$ which implies by Lemma~\ref{le:submodular_bad_opt} that $$f(S) = g(S) \leq 1 - e^{-1/k} + h^{-1} + \varepsilon'.$$ Since Lemma~\ref{le:submodular_good_opt} shows that even in the rare case in which the algorithm does mamange to distinguish between the functions, still $f(S) \leq 1 + \eps'$, this implies
	\[
	\bE[f(T)]
	\leq
	1 - e^{-1/k} + h^{-1} + \varepsilon' + o(1)
	\enspace,
	\]
	where the $o(1)$ term represents a value that goes to zero when $m$ goes to infinity assuming $k$ and $h$ are kept constant.
	
	On the other hand, Lemma~\ref{le:submodular_good_opt} shows that $\cM$ contains an independent set $S$ whose objective value $f(S)$ is at least $1 - 2k/h - \varepsilon'$. 
	Thus, the approximation ratio of the algorithm is no better than
	\begin{align*}
	&\frac{1 - 2k/h - \varepsilon'}{1 - e^{-1/k} + h^{-1} + \varepsilon' + o(1)}\\
	\geq{} &
	(1 - e^{-1/k} + h^{-1} + \varepsilon' + o(1))^{-1} - (1 - e^{-1/k})^{-1}(2k/h + \varepsilon')\\
	\geq{} &
	(1 - e^{-1/k})^{-1} - (1 - e^{-1/k})^{-2}(h^{-1} + \varepsilon' + o(1)) - (1 - e^{-1/k})^{-1}(2k/h + \varepsilon')\\
	\geq{} &
	(1 - e^{-1/k})^{-1} - (k + 1)^2(h^{-1} + \varepsilon' + o(1)) - (k + 1)(2k/h + \varepsilon')
	\enspace,
	\end{align*}
	where the last inequality holds since $1 - e^{-1/k} \geq (k + 1)^{-1}$.
	Choosing a large enough $h$ (compared to $k$) and a small enough $\varepsilon'$ (again, compared to $k$), we can make this approximation ratio larger than $(1 - e^{-1/k})^{-1} - \varepsilon$ for any constant $\varepsilon > 0$.
\end{proof}

\section{Practical Considerations and the \package package} \label{sec:practical}

In this section, we discuss practical considerations for practitioners interested in using \mainalg, \repeatedgreedy, and their variants.
We also present \package, an open source Julia package which implements these simultaneous and repeated greedy techniques, their variants, and practical heuristics.

\subsection{Practical Considerations: Simultaneous or Repeated?} \label{sec:considerations}
In this paper, we have proposed simultaneous and repeated greedy techniques.
The most natural question is: which algorithmic technique is better in practice?
Given enough computational resources, executing both algorithms for a variety of parameter settings is most likely to yield the best solution.
Still, it is of practical interest to know which of the two algorithms will generally perform better.
One may be tempted to judge the effectiveness of these algorithms by comparing the approximation ratios which we derived in the preceding sections.
This line of thinking would suggest that \mainalg is clearly the better choice in practice; however, we caution practitioners against making these judgments based solely on the approximation ratios as they are based on worst-case analysis and may not reflect typical problem instances.

We argue that \repeatedgreedy may be more reasonable to use in practice because it requires less parameter tuning and is guaranteed to return a solution which is at least as good as the greedy solution.
Both \mainalg and \repeatedgreedy require setting the main parameter $\ell$, which is roughly the number of candidate solutions produced in both algorithms.
Although our worst-case analysis yields natural choices of $\ell$ based on properties of the independence system, it is ultimately a parameter which is to be set by the user.
The set of solutions produced by \mainalg will generally be quite different for varying $\ell$ and so the resulting approximation performance of the algorithm really does depend on the choice of $\ell$.
In contrast, \repeatedgreedy produces the same sequence of solutions as $\ell$ increases, and so the approximation can only improve as $\ell$ increases.
In this sense, $\ell$ is simpler to tune when using \repeatedgreedy.
In Section~\ref{sec:experiments}, we observe that the quality of the solution returned by \mainalg varies non-monotonically with the number of solutions $\ell$. 
One way to address this is to run \mainalg for a sequence of $\ell = 1, 2, \dotsc, \ell_{\text{max}}$, and return the best solution.

Another reason to use \repeatedgreedy in practice is that it is at least as effective as the greedy algorithm; on the other hand, \mainalg may perform worse than the greedy algorithm if $\ell$ is set too large.
Finally, we conjecture that the approximation ratio of \repeatedgreedy adapts to the curvature of the submodular objective function for all values of $\ell$, while the approximation ratio of \mainalg adapts to curvature only for a restricted set of $\ell$ values; however, this is beyond the scope of the current paper.

The so-called ``lazy greedy'' search \citep{Minoux1978} is a well-known heuristic for running iterative greedy searches which dramatically speeds up any greedy-based algorithm for submodular optimization, including the simultaneous and repeated greedy algorithms presented here.
Rather than computing the marginal gain of each element in the ground set, the lazy greedy approach for greedy search exploits submodularity by maintaining a priority queue of each element along with its previously queried marginal gain.
Because the marginal gain of an element is non-increasing as the solution set grows, previously queried marginal gains are an upper bound for the current marginal gain.
In this way, the lazy greedy approach (typically) results in only a few oracle queries until the element with the top marginal gain is found.
Although the lazy greedy approach does not improve worst-case runtime, it greatly improves the runtime for many practical instances.
Moreover, the lazy greedy approach may be used together with the marginal gain thresholding technique for improved performance gains.
When using a simultaneous greedy algorithm, the lazy greedy priority queue should be modified to include element-solution pairs.

\subsection{The \package package}\label{sec:package}

Our final main contribution in this paper is \package, an open source Julia package which implements the simultaneous and repeated greedy algorithms described here, along with their nearly linear time and knapsack variants.
We have written the package so that it is easy to use ``out-of-the-box'', requiring little to no knowledge of the algorithmic variants such as marginal gain thresholding and density ratio thresholding.
The package is available at \href{https://github.com/crharshaw/SubmodularGreedy.jl}{this URL}\footnote{https://github.com/crharshaw/SubmodularGreedy.jl} and the installation requires only one line of code using the Julia package manager.
Below, we highlight a few of the design decisions in the package:

\begin{itemize}
	\item \textbf{Supported Algorithms}: 
	The \package package supports all algorithms presented in this paper, including \mainalg, \repeatedgreedy, and their nearly linear-time and knapsack variants. 
	Additionally, the \samplegreedy algorithm of \cite{FHK17} is also included. 
	
	\item \textbf{Oracle Models}: 
	Our implementations run in the \emph{oracle model}. 
	That is, the user provides a value oracle which returns $f(S)$ given $S$ and a independence oracle which, given $S$, determines whether or not $S \in \cI$. 
	In this way, faster implementations of these oracles will result in faster run time of our algorithms.
	
	\item \textbf{Default Parameter Settings}: 
	The various parameters of the algorithms are set by default to values suggested by our worst case analysis. 
	The user may specify whether the independence system is $k$-extendible or a $k$-system and whether the objective is monotone, and the number of candidate solutions $\ell$ is set automatically. 
	Additionally, the $\beta$ scaling terms used in the binary density search are also set based on the  analysis in this paper. 
	However, the user may override any of these default parameter settings in favor of their own.
	
	\item \textbf{Lazy Greedy Implementations}: 
	All of our implementations feature a lazy greedy approach (discussed above in Section~\ref{sec:considerations}) for improved practical performance.
	Moreover, the simultaneous greedy algorithm features a lazy greedy priority queue whose keys are element-solution pairs, which is more appropriate for this setting.
\end{itemize}

The following functions are available in the \package package.
A more comprehensive description of these functions is contained in the documentation of the package.
Additionally, we have included a tutorial of the package as a Jupyter notebook.

\begin{itemize}
	\item \texttt{simultaneous\_greedys}: A fast implementation of \mainalg using approximate greedy search and lazy evaluations. If knapsack constraints are given, the density threshold technique is used with default parameter settings.
	\item \texttt{repeated\_greedy}:  A fast implementation of \repeatedgreedy using approximate greedy search and lazy evaluations. If knapsack constraints are given, the density threshold technique is used with default parameter settings.
	\item \texttt{sample\_greedy}: An implementation of the \samplegreedy algorithm of \citet{FHK17} using lazy evaluations.
	\item \texttt{greedy}: A fast implementation of the \greedy algorithm using approximate greedy search and lazy evaluations. If knapsack constraints are given, the density threshold technique is used with default parameter settings.
	\item \texttt{deterministic\_usm}: An implementation of the deterministic linear-time \usm algorithm of \citet{BFNS14}.
\end{itemize}

\section{Experiments} \label{sec:experiments}

In this section, we demonstrate the efficacy of our proposed algorithms on two movie recommendation settings using a real dataset.
The \package package contains all implementations of algorithms used in this experimental section.

\subsection{MovieLens 20M Dataset}\label{sec:dataset}

In our experiments, we use data from the MovieLens 20M Dataset, which features 20 million ratings of 27,000 movies by 138,000 users.
For each movie, we construct a corresponding feature vector $v_i$ by using a low-rank matrix completion technique on the user reviews, as proposed by \cite{Lindgren2015}.
The feature vectors are a low dimensional representation of the movies, based on the available user reviews. 
For a pair of movies $i$ and $j$, we use the feature vectors to construct a similarity score
\[
s_{i,j} = \exp \left( - \sigma^2 (1 - \cos(v_i, v_j)) \right) \enspace,
\]
where $\cos(v_i,v_j) = \langle v_i, v_j \rangle / (\| v_i \| \|v_j \|)$ is the cosine similarity and $\sigma > 0$ is a user-defined bandwidth parameter which controls the decay of this similarity.
In this way, the similarity scores are based on the rating behavior of the users in the MovieLens 20M dataset.
We remark that the similarity scores are in the range $[0,1]$, where $s_{ij} = 1$ only if $v_i$ is a scaled multiple of $v_j$. 

The MovieLens dataset also contains, for each movie, a list of genres that the movie belongs to.
There are 17 total genres, including Action, Drama, Comedy, Thriller, Musical, and Western, to name a few.
We emphasize that each movie belongs to at least one genre, but typically several genres.
By scraping the Internet Movie Database (IMDb), metadata on the movies is collected, including the release year and the average IMDb  user rating.
Metadata is collected for $n=10,473$ movies and so this is the size of the ground set.

In both experiments, we use the following non-monotone submodular objective function.
\[
f(S) = \frac{1}{n} \left[ \sum_{i \in \gnd} \sum_{j \in S} s_{i,j} - \lambda \cdot \sum_{i \in S} \sum_{j \in S} s_{i,j} \right] \enspace,
\]
where $\lambda \in [0,1]$ is a user-defined penalty term.
The first term captures the extent to which the set $S$ summarizes the entirety of movies in the ground set, while the second term penalizes sets $S$ which have a lot of self-similarity.
When $\lambda = 1$, the objective function recovers the graph-cut function on the graph in which edges weights are the normalized similarities, i.e., $s_{i,j} / n$.

\subsection{Experiment 1: Movie Recommendation with Genre Limitations}
In the first experiment, we aim to provide a user with a movie summarization set in which no genre appears too frequently.
This modeling formulation is most suitable for a user who wants a diverse selection of movies from the dataset, in terms of both the MovieLens user ratings and the genres.

Let $\mathcal{G}$ denote the set of movie genres.
For each movie $e \in \gnd$, let $G_e \subseteq \mathcal{G}$ denote the genres that movie $e$ belongs to.
For each genre $g \in \mathcal{G}$, let $d_g$ be a non-negative integer.
We define the \emph{genre-limiting} constraint set $(\cI, \gnd)$, where $S \in \cI$ if 
\[ 
| \{ e \in S : g \in G_e \} | \leq d_g \text{ for each genre } g \in \mathcal{G}.
\]
In other words, the solution set $S$ contains at most $d_g$ movies belonging to genre $g \in \mathcal{G}$.
One can verify that this constraint set is a $k$-extendible system, where $k = | \mathcal{G} |$.
In particular, the genre limiting constraint is the intersection of $| \mathcal{G} |$ partition matroids.

In this experiment, we consider a sequence of problem instances, defined by a sequence of constraint sets.
For each genre $g \in \mathcal{G}$, we define a \emph{genre fraction limit} $q_g \in [0,1]$.
For an integer $t$, we define genre limits according to the genre fraction limits by $d_g = \textrm{Round}(t \cdot q_g)$.
In this way, we define a sequence of growing constraint sets which are indexed by integers $t \in \mathbb{N}$.
The choice of genre fraction limits encode a user's desired fraction of genres in the summary, while the integer $t$ roughly determines the size of the summary.
In our experiment, we choose the genre fraction limits of most genres to reflect the total fraction of movies belonging to the genre, i.e., $q_g = | \{ e \in \gnd : g \in G_e \} | / n$.
The exceptions are Crime, Drama and Thriller which have slightly higher genre fraction limits and Animation, Children, Romance, and Horror which have slightly lower genre fraction limits.
These modified genre fraction limits may be understood to represent a particular user's personal interest.

We compare \mainalg, \repeatedgreedy, \greedy, and \samplegreedy for the sequence of problem instances in this experiment.
We run these algorithms for problem instances with indices $t = 2, 3, \dotsc, 30$.
For each algorithm, we record the objective value of the returned solution and the required number of oracle calls for each of the problem instances.
As recommended in Section~\ref{sec:considerations}, we take the maximum of \mainalg over setting $\ell = 1, 2, \dotsc, 10$.
For comparison, we set the number of solutions to $\ell = 10$ when running \repeatedgreedy.
We ran the linear time implementations with $\eps \in \{ 0.01, 0.1\}$, but the execution paths of the algorithms remained unchanged (compared to the non-linear time implementations); this is likely a result of the lazy greedy implementation.
For this reason, the linear time implementations are not included in these results.
We ran \samplegreedy for $20$ iterations.
Figure~\ref{fig:exp1} contains the results of this experiment.

\begin{figure}
	\centering
	\begin{subfigure}[b]{0.45\textwidth}
		\centering
		\includegraphics[width=\textwidth]{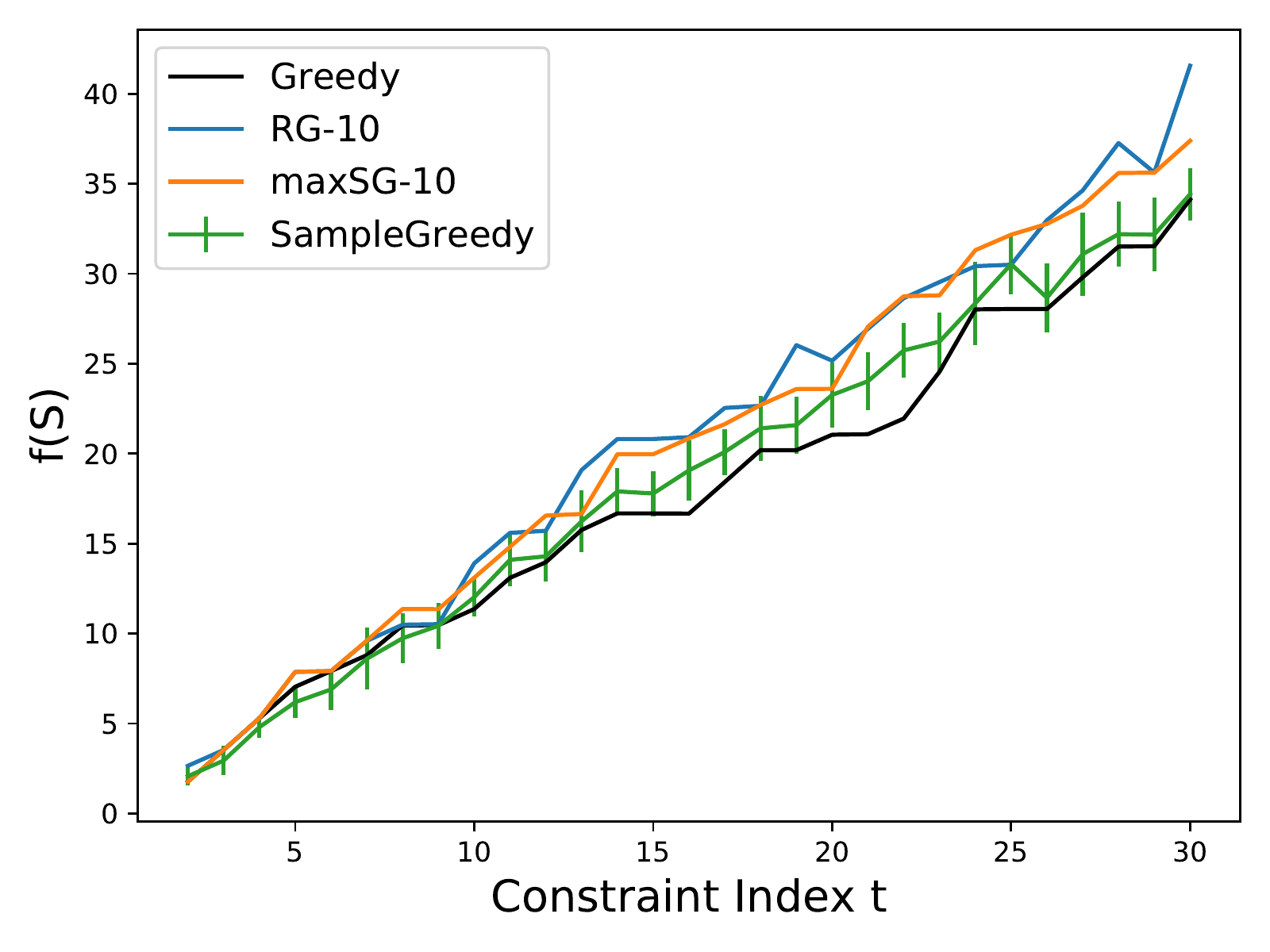}
		\caption{}
		\label{fig:exp1-obj}
	\end{subfigure}
	\hfill
	\begin{subfigure}[b]{0.45\textwidth}
		\centering
		\includegraphics[width=\textwidth]{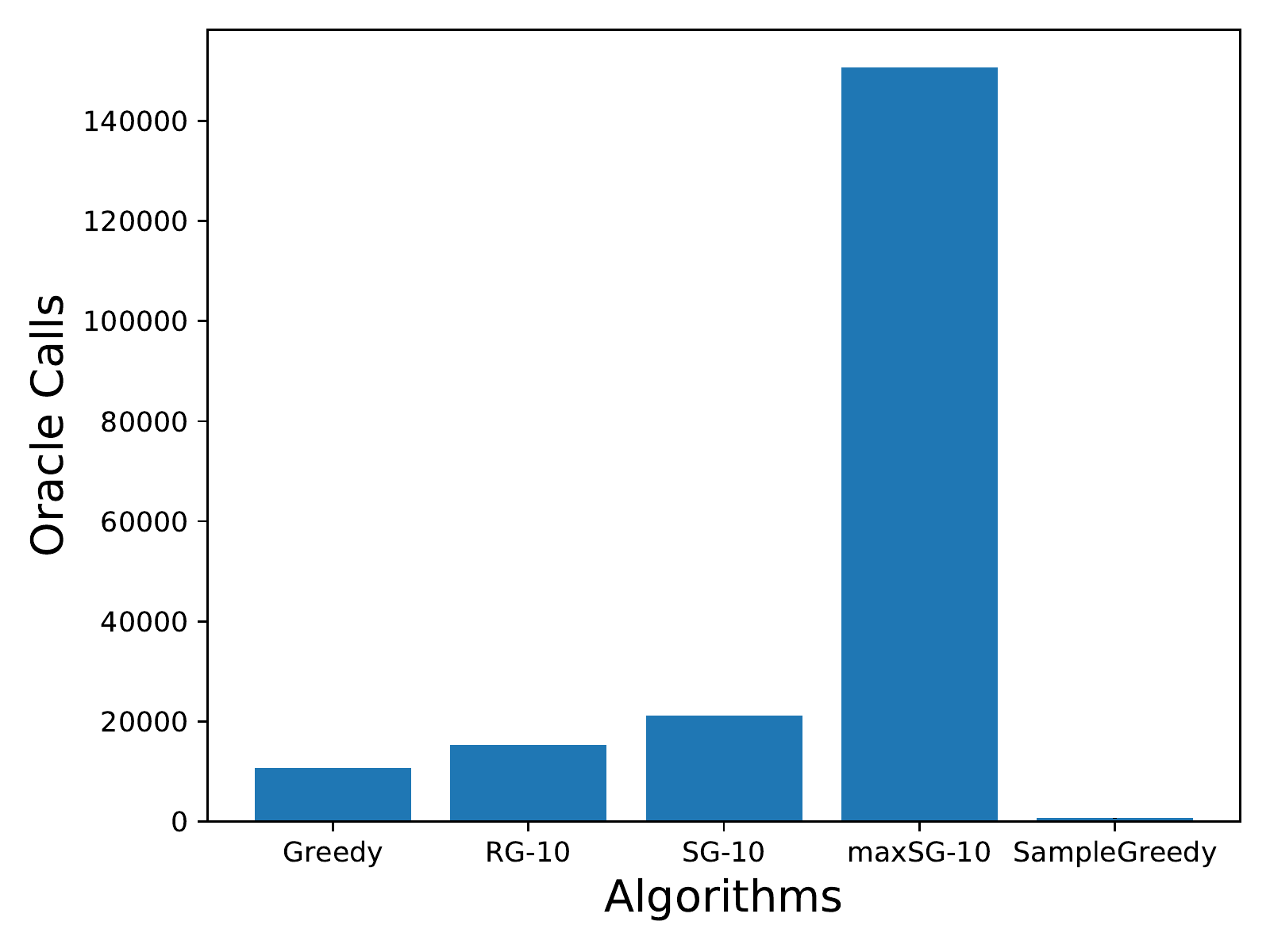}
		\caption{}
		\label{fig:exp1-oracle-calls}
	\end{subfigure}
	\caption{
		A comparison of objective value and runtime of Greedy, Sample Greedy, Repeated Greedy, and Simultaneous Greedys for problem instances in Experiment 1.
		Linear time implementations with $\eps \in \{0.1, 0.01\}$ were run, but not reported here because the execution path did not change.
		Fig~\ref{fig:exp1-obj} plots the objective value attained by the algorithms against the constraint index.
		Fig~\ref{fig:exp1-oracle-calls} plots the number of oracle calls for index $t=30$.
	}
	\label{fig:exp1}
\end{figure}

\begin{figure}
	\centering
	\begin{subfigure}[b]{0.45\textwidth}
		\centering
		\includegraphics[width=\textwidth]{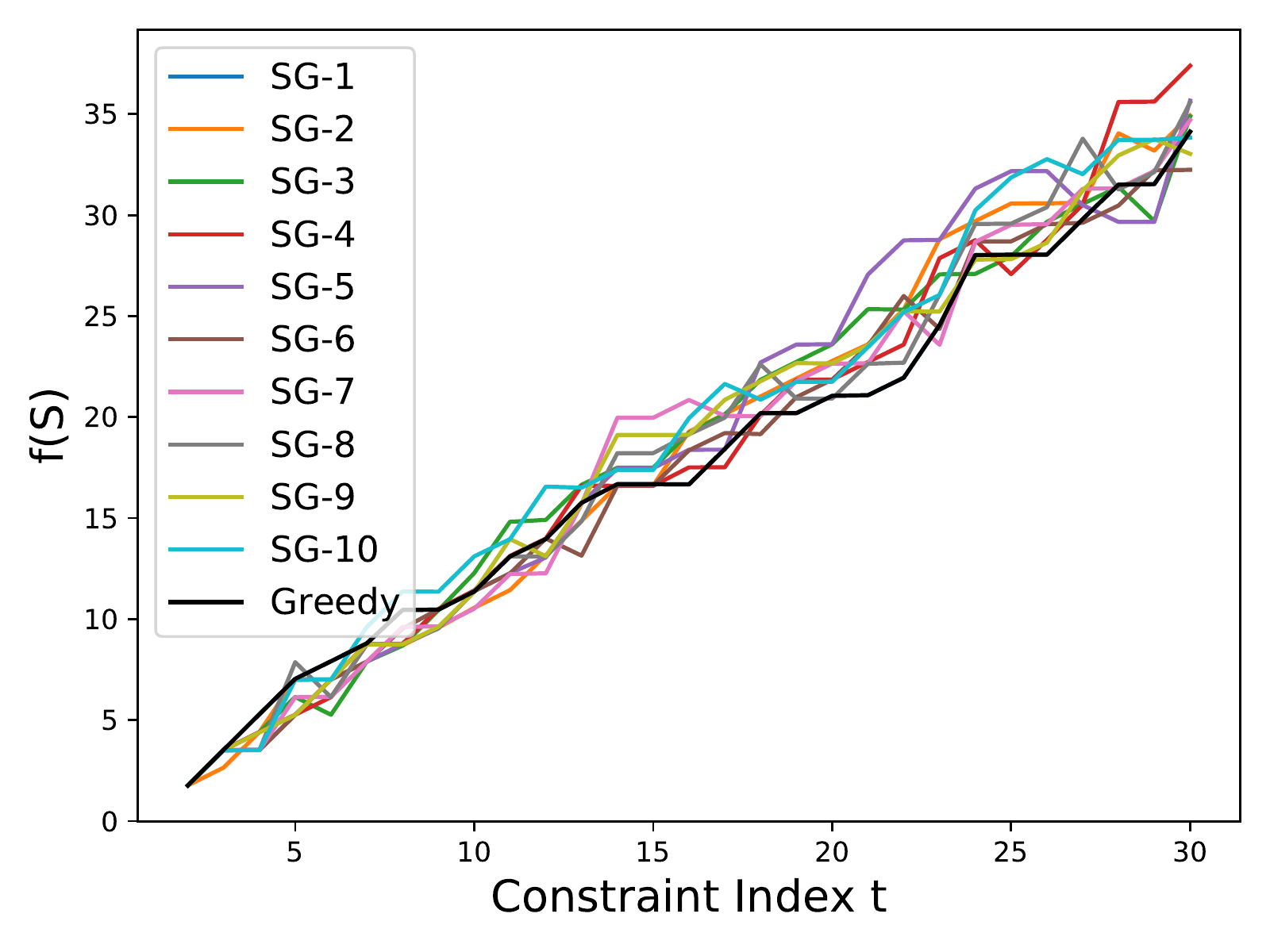}
		\caption{}
		\label{fig:exp1-sg-obj}
	\end{subfigure}
	\hfill
	\begin{subfigure}[b]{0.45\textwidth}
		\centering
		\includegraphics[width=\textwidth]{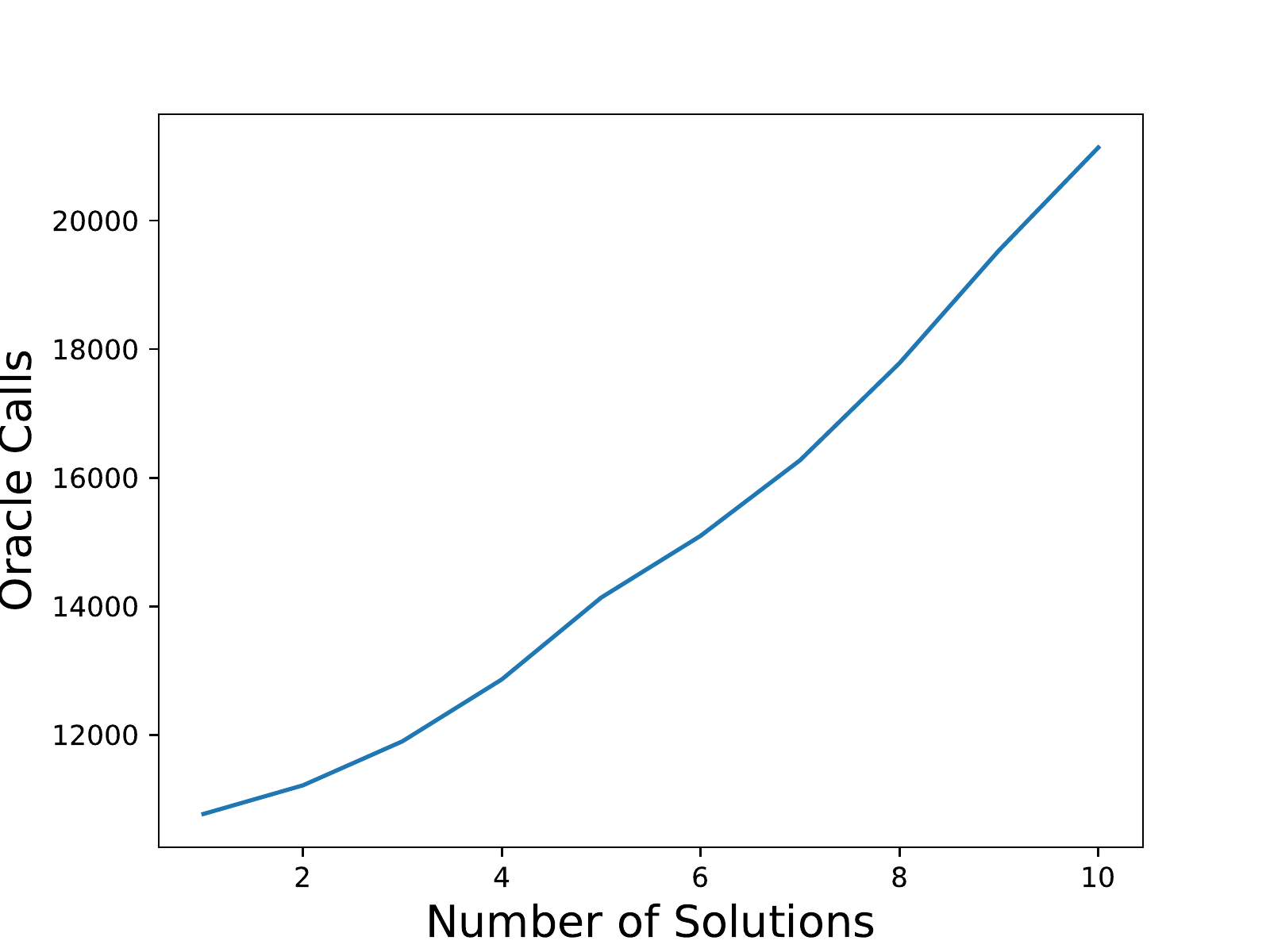}
		\caption{}
		\label{fig:exp1-sg-oracle-calls}
	\end{subfigure}
	\caption{
		A comparison of objective value and runtime of Simultaneous Greedys when the number of solutions $\ell$ is varied, for problem instances in Experiment 1.
		Fig~\ref{fig:exp1-sg-obj} plots the objective values attained by the various Simultaneous Greedys executions against the constraint index.
		Fig~\ref{fig:exp1-sg-oracle-calls} plots the number of oracle calls for index $t=30$ against the number of solutions $\ell$.
	}
	\label{fig:exp1-sg}
\end{figure}

As we see in Figure~\ref{fig:exp1-obj}, \repeatedgreedy and taking the maximum of $\ell = 1, 2, \dots 10$ of \mainalg return higher quality solutions than the greedy algorithm.
In fact, \repeatedgreedy and \mainalg return solutions with larger value than the expected value of the solution returned by \samplegreedy.
Figure~\ref{fig:exp1-oracle-calls} shows the number of oracle calls made by the various algorithms.
\mainalg and \repeatedgreedy require more oracle calls that \greedy and taking the maximum over $\ell = 1,2, \dotsc, 10$ increases this cost.
On the other hand, the expected cost of \samplegreedy is considerably lower than the other algorithms.

Figure~\ref{fig:exp1-sg} demonstrates the behavior of \mainalg as the number of solutions $\ell$ is varied.
Figure~\ref{fig:exp1-sg-obj} shows the objective value attained by \mainalg while varying the number of solutions $\ell = 1, 2, \dotsc, 10$.
For most values of $\ell$, the attained objective value is larger than that of the greedy algorithm; however, we see that there is no value of $\ell$ which consistently returns the highest value of the objective function.
As we see in Figure~\ref{fig:exp1-sg-oracle-calls}, the number of oracle calls increases with the number of solutions, which is to be expected.

\subsection{Experiment 2: Movie Recommendations with Release Dates and Rating Budget}
In the second experiment, we aim to provide a user with a movie summarization set in which movies are far apart in release date and not too many highly rated movies appear.
Our modeling formulation is most suitable for a film watching party based on poorly rated films or ``cult classics' throughout the years'.

In this experiment, we use release date constraint and a rating budget, defined as follows.
For each movie $e \in \gnd$, let $y_e$ denote the release year of the movie.
We define the \emph{release date constraint} $(\cI, \gnd)$, where $S \in \cI$ if
\[
|y_e - y_u| \geq 1 \text{ for all pairs } e,u \in S \enspace. 
\]
In other words, no two movies in a feasible solution set may be released in the same year.
This independence set is $2$-extendible, as adding a movie $e$ to the current solution requires the removal of up to $2$ other movies  that already belong to the set: one that appears up to a year after $y_e$ and one that appears up to a year before $y_e$.

For each movie $e \in \gnd$, we let $r_e$ denote the rating of the movie, according to the IMDb.
The ratings take real values between $1$ and $10$.
The \emph{rating budget constraint} is that
\[
\sum_{e \in S} \max (r_e - 5.0, 0) \leq \beta \enspace,
\]
where $\beta$ is a user-defined \emph{rating budget}.
A set $S$ satisfies the rating budget constraint so long as it does not contain too many highly rated movies; indeed, this constraint does not penalize movies which have a rating less than $5.0$.
Observe that the rating budget is a knapsack constraint with coefficients $c_e =  \max (r_e - 5.0, 0)$.

We compare \densitysearchalg, \densitysearchRG, and \greedy for maximizing the diverse summarization objective over the release date and rating budget constraints.
Recall that \densitysearchalg and \densitysearchRG incorporate the density threshold and density search techniques for handling the knapsack objective, while \greedy incorporates the knapsack constraint into the independence constraint.
For both \densitysearchalg and \densitysearchRG, we set the number of solutions to $\ell = 2$ and we use values $\delta = \eps \in \{ 0.1, 0.01 \}$.

The results of the second experiment are summarized in Figure~\ref{fig:exp2}.
In Figure~\ref{fig:exp2-obj}, we see that \densitysearchalg and \densitysearchRG typically yield solutions with larger objective value than \greedy.
This improvement may be attributed to the density thresholding technique, where an element with high marginal gain may not be chosen if its knapsack cost is relatively larger.
Interestingly, larger values of the error term $\eps$ yields solutions with larger objective values.
This is likely due to increased variability in the execution path of the algorithm, leading to more diverse solutions being constructed. 
In Figure~\ref{fig:exp2-oracle-calls}, we see that the density search techniques are more expensive than the greedy algorithm, especially after the lazy greedy implementation.
However, the cost of the density search techniques decreases as the error term $\eps$ increases, due (in part) to fewer calls to the fixed-density subroutine.

\begin{figure}
	\centering
	\begin{subfigure}[b]{0.45\textwidth}
		\centering
		\includegraphics[width=\textwidth]{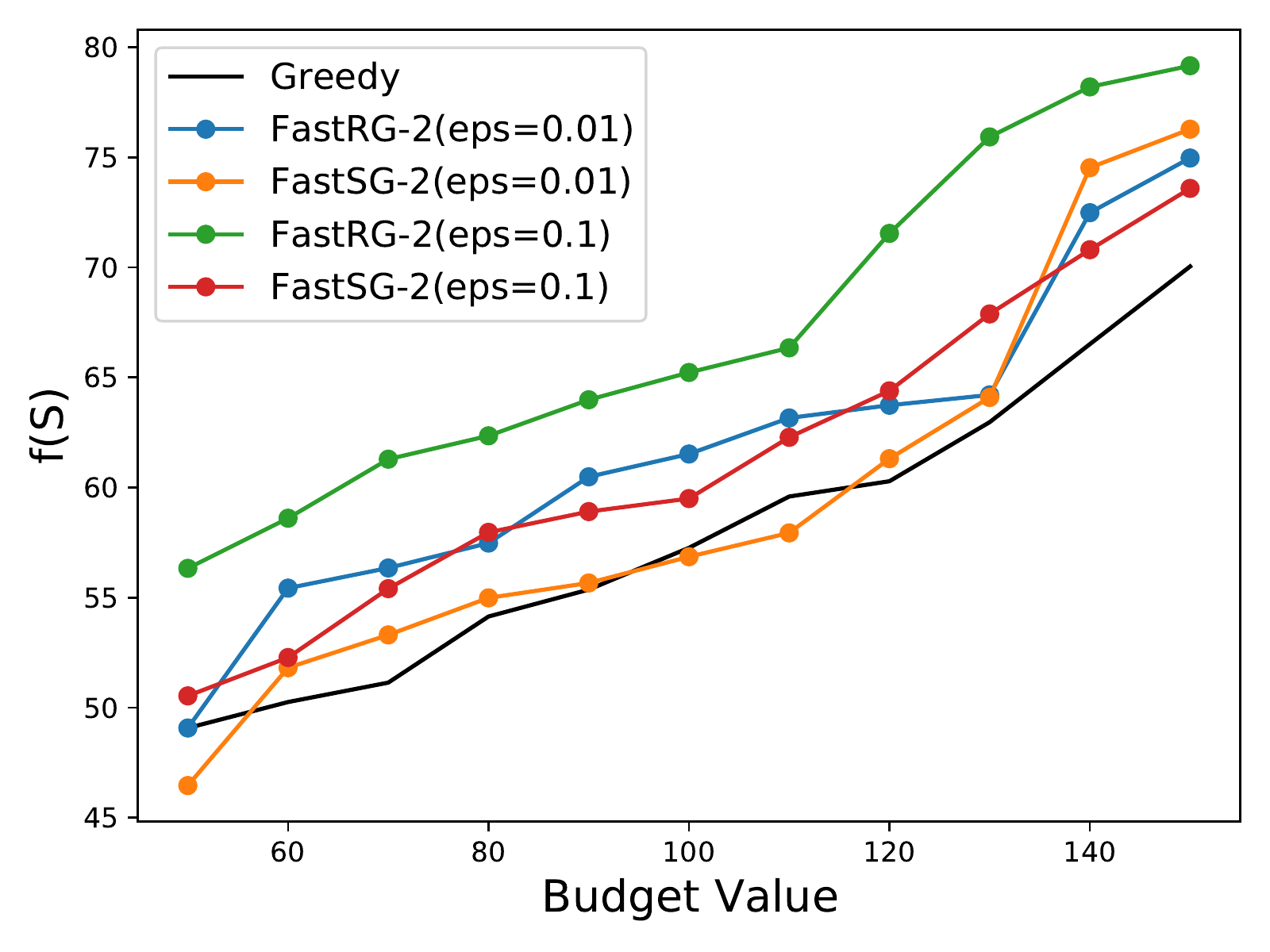}
		\caption{}
		\label{fig:exp2-obj}
	\end{subfigure}
	\hfill
	\begin{subfigure}[b]{0.45\textwidth}
		\centering
		\includegraphics[width=\textwidth]{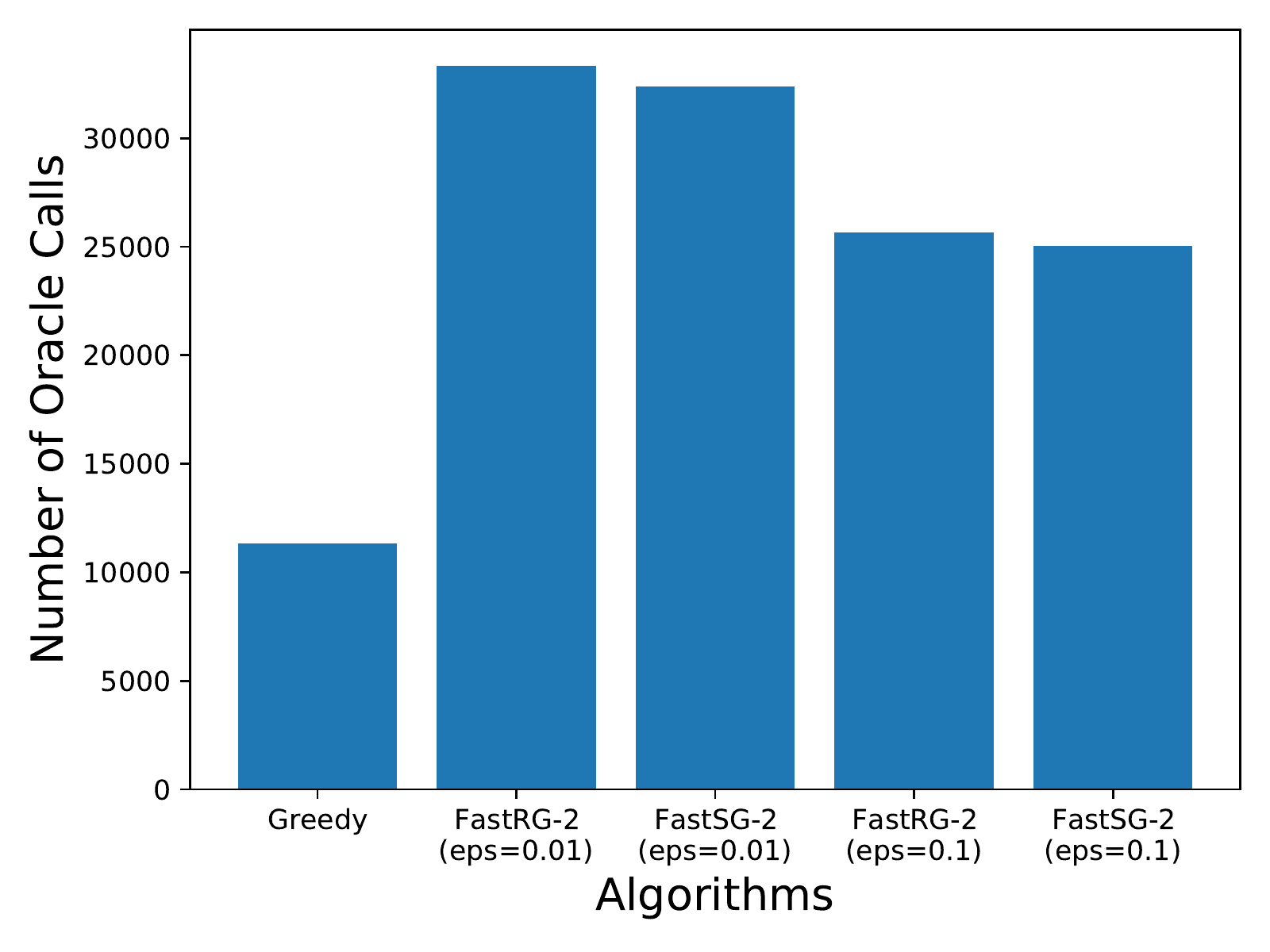}
		\caption{}
		\label{fig:exp2-oracle-calls}
	\end{subfigure}
	\caption{
		A comparison of objective value and runtime of Greedy, Repeated Greedy, and Simultaneous Greedys for problem instances in Experiment 2.
		The linear time and knapsack variants of Repeated Greedy and Simultaneous Greedys are displayed here.
		Fig~\ref{fig:exp2-obj} plots the objective values attained by the algorithms against the budget value.
		Fig~\ref{fig:exp2-oracle-calls}  plots the number of oracle calls for budget value $\beta = 150$
	}
	\label{fig:exp2}
\end{figure}

\section{Conclusion} \label{sec:conclusion}
In this paper, we have presented \mainalg, a new algorithmic technique for constrained submodular maximization.
In addition, we have improved the analysis of \repeatedgreedy, showing that fewer repeated iterations yield a better approximation than what was previously known to be possible.
We have shown that both greedy-based techniques can accommodate several variants, including a nearly-linear implementation and the handling of additional knapsack constraints.
Perhaps most surprisingly, the simple \mainalg algorithmic technique provides the tightest known approximation guarantees across a mix and match of many settings: $k$-system constraints, $k$-extendible constraints, $m$ additional knapsack constraints, non-monotone objectives, and monotone objectives.
We have provided two kinds of negative results: the first is hardness results demonstrating that, for several of these settings, no efficient algorithm can achieve a significantly better approximation ratio.
The second is a result which shows that our analysis of \repeatedgreedy is tight in the sense that it cannot be improved for the subclasses considered here.

We also provided practical insights, arguing that although \mainalg has better worst-case approximation guarantees, \repeatedgreedy is often better suited for practical applications.
Implementations of all the algorithms considered in this paper appear in \package, an open source Julia package which is available for download at \href{https://github.com/crharshaw/SubmodularGreedy.jl}{this URL}\footnote{https://github.com/crharshaw/SubmodularGreedy.jl}.
We hope that these simple, yet theoretically sound, techniques becomes a standard in the toolbox of practitioners across a variety of disciplines.
In a larger sense, we hope that this technique may aid the flexibility of the submodular optimization framework as more exciting applications continue to emerge.

\paragraph{Acknowledgments}
We graciously thank Erik Lindgren and Ehsan Kazemi for sharing the feature vectors and the scraped IMDb meta-data used in our experiments. 
The work of Moran Feldman was supported in part by ISF grants no. 1357/16 and 459/20.
This work was supported in part by an NSF Graduate Research Fellowship (DGE1122492) awarded to Christopher Harshaw. Amin Karbasi is partially supported by NSF (IIS- 1845032), ONR (N00014-19-1-2406), and TATA  Sons Private Limited. 

\newpage
\bibliography{tex/derandomized}

\appendix
\renewcommand\theHsection{appendix.\arabic{section}}
\section{Proof of Proposition~\ref{prop:knapsack_general_result} (\knapsackalg) } \label{sec:knapsack-proof}
In this section, we prove Proposition~\ref{prop:knapsack_general_result}, which is the main technical lemma behind the \mainalg variants \fastalg and \knapsackalg.
The proposition is a meta-analysis that reduces the conditions of approximation to simple combinatorial statements relating the constructed solutions to $\OPT$.
Furthermore, the proposition and its proof mirror Proposition~\ref{prop:general_result}, which provided a similar meta-analysis for approximation guarantees of \mainalg.
We remind the reader that the constructions in Section~\ref{sec:analysis-extendible} and Section~\ref{sec:analysis-k-system} demonstrate how these conditions are satisfied for $k$-extendible systems and $k$-systems, respectively.

We begin by restating the proposition.
\begin{manualtheorem}{\ref{prop:knapsack_general_result}}
	\knapsackgeneralresult[*]
\end{manualtheorem}

As before, at each iteration $1 \leq t \leq \numiter$, the set $\SetF{O}{i}{j}$ contains the elements of $\OPT$ which maintain feasbility in the independence system when added to solution $\SetF{S}{i}{j}$.
It may be the case, however, that some of these elements of $\SetF{O}{i}{j}$ are infeasible to  add to the corresponding solution with respect to the knapsack constraints.
We note also that there are several differences between the proof of Proposition~\ref{prop:knapsack_general_result} and the proof of the earlier Proposition~\ref{prop:general_result}.
The most significant difference is that there is now a case analysis depending on whether or not Line~\ref{line:knapsack-check} of \knapsackalg ever evaluates to false, which is denoted by the indicator variable $\knapevent$.
If $\knapevent = 1$, then a simple argument lower bounds the quality of the returned solution; and if $\knapevent = 0$, then we obtain an approximation guarantee using similar techniques to those used in the proof of Proposition~\ref{prop:general_result}.

In the case of $\knapevent = 0$, the proof techniques differ in a few ways:
first, the elements of $\SetF{O}{i}{j}$ are typically broken up into two groups: those with high density with respect to the current solution and those with low density.
In the analysis, the two groups of elements are considered separately.
Second, the greedy search is now approximate (up to a factor $(1 - \eps)$), and so this factor carries through the analysis.
Finally, the remaining elements of $\SetF{O}{T}{j}$ after termination may have positive marginal gain when added to the constructed solution, but the gain is sufficiently small so that it does not greatly decrease the quality of the constructed solution.

Before continuing, let us set up some notation to split the elements of $\OPT$ that we throw away into high and low density.
Recall that at each iteration $i$, $\SetF{O}{i-1}{j} \setminus (\SetF{O}{i}{j} \cup \SetF{U}{i}{j})$ are the elements of $\OPT$ which must be removed so that element $u_i$ may be added to solution $\SetF{S}{i}{j_i}$.
Of these elements we must throw away, we will distinguish between those with high density and those with low density.
In particular, we will define $\SetF{\highdensity}{i}{j}$ to be those elements of high density with respect to solution $\SetF{S}{i}{j}$ and $\SetF{\lowdensity}{i}{j}$ to be those elements of low density with respect to solution $\SetF{S}{i}{j}$.
More formally, for any solution $1 \leq j \leq \ell$ and iteration $1 \leq i \leq \numiter$, we define the sets
\begin{align*}
\SetF{\highdensity}{i}{j} &= 
\left\{ u \in  \SetF{O}{i-1}{j} \setminus (\SetF{O}{i}{j} \cup \SetF{U}{i}{j}) : 
f( u \mid \SetF{S}{i}{j} ) \geq \density \cdot \sum_{r=1}^m c_r(u) \right\} \\
\SetF{\lowdensity}{i}{j} &= \left[ \SetF{O}{i-1}{j} \setminus (\SetF{O}{i}{j} \cup \SetF{U}{i}{j}) \right] \setminus \SetF{\highdensity}{i}{j}
\enspace.
\end{align*}

The following lemma is the first step towards proving Proposition~\ref{prop:knapsack_general_result}. Intuitively, this lemma shows that as the iteration $i$ increases, the decrease in the value of $f(\SetF{O}{i}{j} \mid \SetF{S}{i}{j})$ is transferred, at least to some extent, to $\SetF{S}{i}{j}$.

\begin{lemma} \label{lem:ks_greedy_invariant}
	Given the conditions of Proposition~\ref{prop:knapsack_general_result}, if $\knapevent = 0$, then for every iteration $0 \leq i \leq \numiter$,
	\[
	\frac{(p + 1)}{(1 - \eps)} \cdot \sum_{j = 1}^\ell f(\SetF{S}{i}{j}) + \sum_{j = 1}^\ell f(\SetF{O}{i}{j} \mid \SetF{S}{i}{j})
	\geq
	\sum_{j = 1}^\ell f(\OPT \cup \SetF{S}{i}{j})
	- \density \sum_{t=1}^i \sum_{j=1}^\ell \sum_{r=1}^m c_r( \SetF{\lowdensity}{i}{j})
	\enspace.
	\]
\end{lemma}
\begin{proof}
	We prove the lemma by induction on the iterations $i=0, 1, \dotsc, \numiter$. 
	The base case is the case of $i = 0$, corresponding to the initialization of the algorithm.
	Recall that  the solutions are initialized to be empty, i.e.,  $\SetF{S}{0}{j} = \varnothing$ for every $j \in [\ell]$.
	This, together with non-negativity of $f$, implies that
	{\allowdisplaybreaks
		\begin{align*}
		\sum_{j = 1}^\ell f(\OPT \cup \SetF{S}{0}{j})
		&= \sum_{j = 1}^\ell f(\OPT \cup \varnothing) 
		&\text{(by the initialization $\SetF{S}{0}{j} = \varnothing$)} \\
		&= \sum_{j = 1}^\ell f(\varnothing) + \sum_{j = 1}^\ell f(\OPT \mid \varnothing) 
		&\text{(rearranging terms)}\\
		&\leq \frac{(p + 1)}{(1 - \eps)} \cdot \sum_{j = 1}^\ell f(\varnothing) + \sum_{j = 1}^\ell f(\OPT \mid \varnothing)
		&\mspace{-9mu}\text{($f(\varnothing) \geq 0$ by the non-negativity)} \\
		&\leq  \frac{(p + 1)}{(1 - \eps)} \cdot \sum_{j = 1}^\ell f(\SetF{S}{0}{j}) + \sum_{j = 1}^\ell f(\SetF{O}{0}{j} \mid \SetF{S}{0}{j}) \enspace.
		&\text{(by the initialization $\SetF{S}{0}{j} = \varnothing$)}
		\end{align*}
	}%
	This establishes the base case as the right term appearing on the right hand side of the lemma's inequality is zero when $i = 0$.
	
	Assume now that the lemma holds for all iterations $i - 1 \geq 0$, and let us prove it for iteration $i$.
	Recall that only the solution $\SetF{S}{i}{j_i}$ is modified during iteration $i$.
	Thus, we have that the change in iteration $i$ in the first sum in the guarantee of the lemma is
	\begin{equation} \label{eq:ks_S_difference}
	\frac{(p + 1)}{(1 - \eps)} \cdot \sum_{j = 1}^\ell f(\SetF{S}{i}{j})
	-
	\frac{(p + 1)}{(1 - \eps)} \cdot \sum_{j = 1}^\ell f(\SetF{S}{i - 1}{j})
	=
	\frac{(p + 1)}{(1 - \eps)} \cdot f(u_i \mid \SetF{S}{i - 1}{j_i})
	\enspace.
	\end{equation}
	
	Bounding the change in the second sum in the guarantee is more involved, and is done in three steps. 
	The first step is the following inequality.
	\begin{align} \label{eq:ks_S_difference_condition_on}
	\sum_{j = 1}^\ell f(\SetF{O}{i-1}{j} \mid \SetF{S}{i - 1}{j}) 
	&- \sum_{j = 1}^\ell f(\SetF{O}{i-1}{j} \mid \SetF{S}{i}{j}) \\ \nonumber
	&= f(\SetF{O}{i-1}{j_i} \mid \SetF{S}{i - 1}{j_i}) - f(\SetF{O}{i-1}{j_i} \mid \SetF{S}{i}{j_i})
	&\text{(only $\SetF{S}{i}{j_i}$ is modified)}\\ \nonumber
	&= f(u_i \mid \SetF{S}{i-1}{j_i}) - f(u_i \mid \SetF{O}{i-1}{j_i} \cup \SetF{S}{i - 1}{j_i})
	&\text{(rearranging terms)} \\ \nonumber 
	&\leq f(u_i \mid \SetF{S}{i-1}{j_i}) - f(u_i \mid\OPT \cup \SetF{S}{i - 1}{j_i})
	\enspace,
	\end{align}
	where the inequality may be proved by considering two cases.
	First, suppose that $u_i \in \SetF{O}{i-1}{j_i} \cup \SetF{S}{i-1}{j_i}$.
	In this case, the inequality holds with equality, because $\SetF{O}{i-1}{j_i} \subseteq \OPT$ by assumption.
	Consider now the case in which $u_i \not \in \SetF{O}{i-1}{j_i} \cup \SetF{S}{i-1}{j_i}$. 
	In this case, our assumption that $(\SetF{S}{\numiter}{j_i} \setminus \SetF{S}{i-1}{j_i}) \cap \OPT \subseteq \SetF{O}{i-1}{j_i}$ implies $u_i \not \in (\SetF{S}{\numiter}{j_i} \setminus \SetF{S}{i-1}{j_i}) \cap \OPT$, which implies in its turn $u_i \not \in \OPT$ since $u_i \in \SetF{S}{i}{j_i} \subseteq \SetF{S}{\numiter}{j_i}$ and $u_i \in \gnd_{i - 1} \subseteq \gnd \setminus \SetF{S}{i-1}{j_i}$. 
	Therefore, we get that in this case that Inequality~\eqref{eq:ks_S_difference_condition_on} holds due to the submodularity of $f$ (recall that $\SetF{O}{i-1}{j_i} \subseteq \OPT$ by our assumption).
	
	For the second step in the proof of the above mentioned bound, we use submodularity to bound the marginal gain $f(\SetF{O}{i - 1}{j} \mid \SetF{S}{i}{j})$ using sums of marginal gains of single elements.
	Observe that
	\begin{align*}
	\sum_{j = 1}^\ell f(\SetF{O}{i - 1}{j} \mid \SetF{S}{i}{j})
	&\leq \sum_{j = 1}^\ell f(\SetF{O}{i}{j} \mid \SetF{S}{i}{j}) 
	+ \sum_{j = 1}^\ell \sum_{u \in \SetF{O}{i-1}{j} \setminus \SetF{O}{i}{j}} \mspace{-18mu} f(u \mid \SetF{S}{i}{j})
	&\text{(submodularity, $\SetF{O}{i}{j} \subseteq \SetF{O}{i - 1}{j}$)} \\
	&= \sum_{j = 1}^\ell f(\SetF{O}{i}{j} \mid \SetF{S}{i}{j}) + \sum_{j = 1}^\ell \sum_{u \in \SetF{O}{i-1}{j} \setminus (\SetF{O}{i}{j} \cup \SetF{U}{i}{j})} \mspace{-36mu} f(u \mid \SetF{S}{i}{j}) \enspace. \mspace{-18mu}
	&\text{($\SetF{U}{i}{j} \subseteq \SetF{S}{i}{j}$)}
	\end{align*}
	
	The third step is to analyze the inner sum above by partitioning the elements $u \in \SetF{O}{i-1}{j} \setminus (\SetF{O}{i}{j} \cup \SetF{U}{i}{j})$ based on their density, i.e. into the two sets $\SetF{\highdensity}{i}{j}$ and $\SetF{\lowdensity}{i}{j}$ (recall that these two sets are indeed a partition of $\SetF{O}{i-1}{j} \setminus (\SetF{O}{i}{j} \cup \SetF{U}{i}{j})$).
	\[
	\sum_{u \in \SetF{O}{i-1}{j} \setminus (\SetF{O}{i}{j} \cup \SetF{U}{i}{j})} \mspace{-36mu} f(u \mid \SetF{S}{i}{j})
	= \sum_{u \in \SetF{\highdensity}{i}{j}} f(u \mid \SetF{S}{i}{j}) + \sum_{u \in \SetF{\lowdensity}{i}{j}} f(u \mid \SetF{S}{i}{j})
	\]
	The second sum may be bounded by virtue of the low density of its elements, as
	\[
	\sum_{u \in \SetF{\lowdensity}{i}{j}} f(u \mid \SetF{S}{i}{j}) 
	\leq \sum_{u \in \SetF{\lowdensity}{i}{j}}  \density \cdot \sum_{r=1}^m c_r(u)
	= \density \sum_{r=1}^m c_r(\SetF{\lowdensity}{i}{j} )
	\enspace.
	\]
	Recall now that by the approximate greedy search, the element-solution pair $(u_i, \SetF{S}{i}{j_i})$ has the property that $ f( u_i \mid \SetF{S}{i}{j_i}) \geq (1 - \eps) f(u \mid \SetF{S}{i}{j})$ for all element-solution pairs $(u, \SetF{S}{i}{j})$ where $\SetF{S}{i}{j} + u$ is feasible with respect to independence system, and $u$ has high density with respect to $\SetF{S}{i}{j}$.
	In particular, we have that $ f( u_i \mid \SetF{S}{i}{j_i}) \geq (1 - \eps) f(u \mid \SetF{S}{i}{j})$ for all $u \in \SetF{\highdensity}{i}{j}$.
	This yields an upper bound on the first sum,
	\[
	\sum_{u \in \SetF{\highdensity}{i}{j}} f(u \mid \SetF{S}{i}{j}) 
	\leq \sum_{u \in \SetF{\highdensity}{i}{j}} ( 1 - \eps)^{-1} f( u_i \mid \SetF{S}{i}{j_i})
	= ( 1 - \eps)^{-1} f( u_i \mid \SetF{S}{i}{j_i}) \cdot | \SetF{\highdensity}{i}{j} | 
	\enspace.
	\]
	Combining the upper bounds we have obtained on the sums corresponding to $\SetF{\lowdensity}{i}{j}$ and $\SetF{\highdensity}{i}{j}$ yields
	\begin{align*}
	\sum_{j = 1}^\ell \sum_{u \in \SetF{O}{i-1}{j} \setminus (\SetF{O}{i}{j} \cup \SetF{U}{i}{j})} \mspace{-36mu} f(u \mid \SetF{S}{i}{j})
	& \leq{}\density \sum_{j=1}^\ell \sum_{r=1}^m c_r(\SetF{\lowdensity}{i}{j} ) 
	+ \sum_{j=1}^\ell ( 1 - \eps)^{-1} f( u_i \mid \SetF{S}{i}{j_i}) \cdot | \SetF{\highdensity}{i}{j} |  
		\mspace{-36mu}\\
& 	={}\density \sum_{j=1}^\ell \sum_{r=1}^m c_r(\SetF{\lowdensity}{i}{j} ) 
	+ ( 1 - \eps)^{-1} f( u_i \mid \SetF{S}{i}{j_i}) \cdot \sum_{j=1}^\ell | \SetF{\highdensity}{i}{j} |  
		\text{~~~~~(rearranging)}\\
&	\leq{} \density \sum_{j=1}^\ell \sum_{r=1}^m c_r(\SetF{\lowdensity}{i}{j} ) 
	+ ( 1 - \eps)^{-1} f( u_i \mid \SetF{S}{i}{j_i}) \cdot \sum_{j=1}^\ell | \SetF{O}{i-1}{j} \setminus (\SetF{O}{i}{j} \cup \SetF{U}{i}{j}) | \\
& 	\leq{}\density \sum_{j=1}^\ell \sum_{r=1}^m c_r(\SetF{\lowdensity}{i}{j} ) 
	+ \frac{p}{ 1 - \eps}  \cdot f( u_i \mid \SetF{S}{i}{j_i}) 
	\enspace,
	\end{align*}
	where the cardinality bound $| \SetF{\highdensity}{i}{j} |  \leq | \SetF{O}{i-1}{j} \setminus (\SetF{O}{i}{j} \cup \SetF{U}{i}{j}) | $ in second inequality follows from the containment $\SetF{\highdensity}{i}{j} \subseteq \SetF{O}{i-1}{j} \setminus (\SetF{O}{i}{j} \cup \SetF{U}{i}{j})$
	and the last inequality follows from the final condition of the proposition which states that $\sum_{j=1}^\ell | \SetF{O}{i-1}{j} \setminus (\SetF{O}{i}{j} \cup \SetF{U}{i}{j}) |  \leq p$ .
	Together with the inequality from this second step, this yields
	\begin{equation} \label{eq:ks_O_difference}
	\sum_{j = 1}^\ell f(\SetF{O}{i - 1}{j} \mid \SetF{S}{i}{j})
	\leq \sum_{j = 1}^\ell f(\SetF{O}{i}{j} \mid \SetF{S}{i}{j}) 
	+ \density \sum_{j=1}^\ell \sum_{r=1}^m c_r(\SetF{\lowdensity}{i}{j} ) 
	+ \frac{p}{ 1 - \eps}  \cdot f( u_i \mid \SetF{S}{i}{j_i})
	\enspace.
	\end{equation}
	The remainder of the proof consists of combining the three inequalities \eqref{eq:ks_S_difference}, \eqref{eq:ks_S_difference_condition_on} and \eqref{eq:ks_O_difference} with the induction hypothesis, as follows.
	\begin{align*}
		\frac{(p+1)}{(1-\eps)}& \cdot \sum_{j = 1}^\ell f(\SetF{S}{i}{j}) 
		+ \sum_{j = 1}^\ell f(\SetF{O}{i}{j} \mid \SetF{S}{i}{j})\\
	\geq{} &
		\left[
			\frac{(p + 1)}{(1-\eps)} \cdot \sum_{j = 1}^\ell f(\SetF{S}{i - 1}{j}) 
			+ \frac{(p + 1)}{(1-\eps)} \cdot f(u_i \mid \SetF{S}{i - 1}{j_i})
		\right] 
	\\ & 
	+ \left[
		\sum_{j = 1}^\ell f(\SetF{O}{i-1}{j} \mid \SetF{S}{i}{j}) 
		- \frac{p}{(1-\eps)} \cdot f(u_i \mid \SetF{S}{i-1}{j_i}) 
		- \density \sum_{j=1}^\ell \sum_{r=1}^m c_r(\SetF{\lowdensity}{i}{j} ) 
		\right] \\
	\geq{} &
		\frac{(p + 1)}{(1-\eps)} \cdot \sum_{j = 1}^\ell f(\SetF{S}{i - 1}{j})
		+ \sum_{j = 1}^\ell f(\SetF{O}{i - 1}{j} \mid \SetF{S}{i}{j}) 
		+  f(u_i \mid \SetF{S}{i-1}{j_i}) 
		- \density \sum_{j=1}^\ell \sum_{r=1}^m c_r(\SetF{\lowdensity}{i}{j} ) \\
	\geq{} &
		\frac{(p + 1)}{(1-\eps)} \cdot \sum_{j = 1}^\ell f(\SetF{S}{i - 1}{j})
		+ \sum_{j = 1}^\ell f(\SetF{O}{i - 1}{j} \mid \SetF{S}{i - 1}{j})
		+  f(u_i \mid\OPT \cup \SetF{S}{i - 1}{j_i})
		- \density \sum_{j=1}^\ell \sum_{r=1}^m c_r(\SetF{\lowdensity}{i}{j} ) \\
	\geq{} &
		\sum_{j = 1}^\ell f(\OPT \cup \SetF{S}{i - 1}{j}) 
		- \density \sum_{t=1}^{i-1}\sum_{j=1}^\ell \sum_{r=1}^m c_r( \SetF{\lowdensity}{i}{j})
		+  f(u_i \mid\OPT \cup \SetF{S}{i - 1}{j_i})
		- \density \sum_{j=1}^\ell \sum_{r=1}^m c_r(\SetF{\lowdensity}{i}{j} ) \\
	={} &
		\sum_{j = 1}^\ell f(\OPT \cup \SetF{S}{i - 1}{j}) 
		+  f(u_i \mid\OPT \cup \SetF{S}{i - 1}{j_i})
		- \density \sum_{t=1}^{i}\sum_{j=1}^\ell \sum_{r=1}^m c_r( \SetF{\lowdensity}{i}{j}) \\
	={}&
		\sum_{j = 1}^\ell f(\OPT \cup \SetF{S}{i}{j}) 
		- \density \sum_{t=1}^{i}\sum_{j=1}^\ell \sum_{r=1}^m c_r( \SetF{\lowdensity}{i}{j}) 
	\end{align*}
	where the first inequality follows from \eqref{eq:ks_S_difference} and \eqref{eq:ks_O_difference}, the second inequality holds since $f(u_i \mid \SetF{S}{i - 1}{j_i})$ is guaranteed to be non-negative, the third inequality follows from \eqref{eq:ks_S_difference_condition_on}, and the fourth inequality follows by induction.
\end{proof}

\begin{corollary} \label{cor:ks_approximation_to_union}
	Given the conditions of Proposition~\ref{prop:knapsack_general_result}, if $E = 0$, then the solutions constructed by \knapsackalg satisfy the lower bound
	\[
	\frac{(p+1)}{(1 - \eps)} \sum_{j=1}^\ell f( \SetF{S}{\numiter}{j}) 
	\geq \sum_{j=1}^\ell f( \OPT \cup \SetF{S}{\numiter}{j}) - \eps \ell \maxgain - \density \ell m
	\enspace.
	\]
\end{corollary}
\begin{proof}
	Our first step is to show that $f(\SetF{O}{\numiter}{j} \mid \SetF{S}{\numiter}{j})$ is negligable for every solution $1 \leq j \leq \ell$.
	To this end, consider any fixed solution $\SetF{S}{\numiter}{j}$ for $1 \leq j \leq \ell$.
	By the termination conditions of \knapsackalg, each element $u \in \gnd_{\numiter}$ satisfies
	\begin{equation} \label{eq:mg_termination}
	f(u \mid \SetF{S}{\numiter}{j}) < \max \left((\eps / n) \cdot \maxgain, \density \cdot \sum_{r=1}^m c_r(u) \right) 
	\enspace.
	\end{equation}
	In particular, this holds for each $u \in \SetF{O}{\numiter}{j}$, as the set $\SetF{O}{\numiter}{j}$ is contained in $\gnd_{\numiter}$.
	We now partition the set $\SetF{O}{\numiter}{j}$ into two groups: the elements with high density and the elements of low density.
	More formally, let $\SetF{\lowdensity}{\numiter+1}{j}$ be the elements in $\SetF{O}{\numiter}{j}$ with low density, 
	\[
	\SetF{\lowdensity}{\numiter+1}{j} 
	= \left\{ u \in  \SetF{O}{\numiter}{j} : f(u \mid \SetF{S}{\numiter}{j}) < \density \cdot \sum_{r=1}^m c_r(u)\right\},
	\] 
	and define $\SetF{\highdensity}{\numiter+1}{j} = \SetF{O}{\numiter}{j} \setminus \SetF{\lowdensity}{\numiter+1}{j}$ to be the high density elements.
	We claim that adding any high density element in  $\SetF{\highdensity}{\numiter+1}{j}$ to the solution $\SetF{S}{\numiter}{j}$ has a marginal gain of at most $(\eps / n) \cdot \maxgain$.
	To see this, observe that because the element $u$ has high density, \eqref{eq:mg_termination} implies that $f(u \mid \SetF{S}{\numiter}{j}) < (\eps / n) \cdot \maxgain$.
	
	Using the above observations, we can now bound the marginal gain of adding $\SetF{O}{\numiter}{j}$ to $\SetF{S}{\numiter}{j}$ as follows.
	\begin{align*}
	f( \SetF{O}{\numiter}{j} \mid \SetF{S}{\numiter}{j})
	&\leq \sum_{u \in \SetF{O}{\numiter}{j}} f( u \mid \SetF{S}{\numiter}{j})
		&\text{(submodularity)} \\
	&= 
		\sum_{u \in \SetF{\highdensity}{\numiter+1}{j}} f( u \mid \SetF{S}{\numiter}{j}) 
		+ \sum_{u \in \SetF{\lowdensity}{\numiter+1}{j}} f( u \mid \SetF{S}{\numiter}{j})
		&\text{(partitioning the sum)} \\
	&\leq 
		\sum_{u \in \SetF{\highdensity}{\numiter+1}{j}} \frac{\eps}{n} \cdot \maxgain 
		+ \sum_{u \in \SetF{\lowdensity}{\numiter+1}{j}} \density \cdot \sum_{r=1}^m c_r(u)
		&\text{(above bound)} \\
	&= 
		\frac{|\SetF{\highdensity}{\numiter+1}{j}|}{n} \cdot \eps \maxgain 
		+  \density \sum_{r=1}^m c_r(\SetF{\lowdensity}{\numiter+1}{j}) \\
	&\leq 
		\eps \maxgain 
		+  \density \sum_{r=1}^m c_r(\SetF{\lowdensity}{\numiter+1}{j}).
	\end{align*}
	
	Substituting the above bound into the guarantee of Lemma~\ref{lem:ks_greedy_invariant} for the final iteration $i=\numiter$ implies
	\begin{align*}
	\frac{(p + 1)}{(1 - \eps)} \cdot \sum_{j = 1}^\ell f(\SetF{S}{\numiter}{j}) 
	&\geq 
		\sum_{j = 1}^\ell f(\OPT \cup \SetF{S}{i}{j})
	- \sum_{j = 1}^\ell f(\SetF{O}{\numiter}{j} \mid \SetF{S}{\numiter}{j})
	- \density \sum_{t=1}^\numiter \sum_{j=1}^\ell \sum_{r=1}^m c_r( \SetF{\lowdensity}{t}{j}) \\
	&\geq 
		\sum_{j = 1}^\ell f(\OPT \cup \SetF{S}{i}{j})
		- \sum_{j = 1}^\ell \left[ \eps \maxgain 
		+  \density \sum_{r=1}^m c_r(\SetF{\lowdensity}{\numiter+1}{j}) \right]  
		- \density \sum_{t=1}^\numiter \sum_{j=1}^\ell \sum_{r=1}^m c_r( \SetF{\lowdensity}{t}{j}) \\
	&=
		\sum_{j = 1}^\ell f(\OPT \cup \SetF{S}{i}{j})
		- \eps \ell \maxgain 
		- \density \sum_{t=1}^{\numiter+1} \sum_{j=1}^\ell \sum_{r=1}^m c_r( \SetF{\lowdensity}{t}{j})
	\end{align*}
	To complete the proof of the corollary, we need to show that $\sum_{t=1}^{\numiter+1} \sum_{j=1}^\ell \sum_{r=1}^m c_r( \SetF{\lowdensity}{t}{j}) \leq \ell m$.
	To this end, observe that for each solution $1 \leq j \leq \ell$, the sets $\SetF{\lowdensity}{1}{j}, \dots \SetF{\lowdensity}{\numiter+1}{j}$ are disjoint subsets of $\OPT$.
	Also observe that $\OPT$ is a feasible solution so that it satisfies all knapsack constraints, $c_r(\OPT) \leq 1$ for all $1 \leq r \leq m$.
	Using these facts and the modularity of the knapsack functions, we have that
	\begin{align*}
	\sum_{t=1}^{\numiter+1} \sum_{j=1}^\ell \sum_{r=1}^m c_r( \SetF{\lowdensity}{t}{j}) 
	&= \sum_{j=1}^\ell \sum_{r=1}^m \sum_{t=1}^{\numiter+1} c_r( \SetF{\lowdensity}{t}{j}) 
		&\text{(rearranging terms)} \\
	&= \sum_{j=1}^\ell \sum_{r=1}^m c_r\left( \cup_{t=1}^{\numiter+1} \SetF{\lowdensity}{t}{j} \right) 
		&\text{(disjointedness, modularity)} \\
	&\leq \sum_{j=1}^\ell \sum_{r=1}^m c_r \left( \OPT \right)
		&\text{($\cup_{t=1}^{\numiter+1} \SetF{\lowdensity}{t}{j} \subseteq \OPT$)} \\
	&\leq \sum_{j=1}^\ell \sum_{r=1}^m 1 
		&\text{(feasibility of $\OPT$)}\\
	&= \ell m 
	\enspace.
	\qedhere
	\end{align*}
\end{proof}

\begin{proof}[Proof of Proposition~\ref{prop:knapsack_general_result}]
	The analysis proceeds with two cases, depending on whether $\knapevent = 1$ or $\knapevent = 0$.
	
	First, suppose that $\knapevent = 1$, which is to say that Line~\ref{line:knapsack-check} evaluates to \texttt{false} at some point during the execution of the algorithm.
	This happens when, at some iteration $i$ there exists a solution $\SetF{S}{i}{j}$ and a high density element $u$ such that adding the element to this set is feasible in the independence system, but the knapsack constraint is violated. 
	More precisely, the set $A \triangleq \SetF{S}{i}{j} + u$ is independent (i.e., $A \in \cI$) but  $c_r(\SetF{S}{i}{j} + u) > 1$ for some knapsack function $1 \leq r \leq m$.
	Although $A$ itself is not feasible, we claim that $f(A) > \density$.
	To this end, let us order the elements of $A$ according to the order in which they were added to $\SetF{S}{i}{j}$, with $u$ appearing last, i.e., $A = \{ u_1, u_2, \dots u_k \}$ with $u_k = u$.
	For $1 \leq i \leq k$, define the sets $A_i = \{ u_1, u_2, \dots u_i \}$ and $A_0 = \varnothing$.
	Then, we obtain the lower bound
	\[
	f(A) 
	= \sum_{i=1}^k f( u_i \mid A_{i-1})
	\geq \sum_{i=1}^k \density \cdot \sum_{r=1}^m c_r(u_i) 
	= \density \sum_{r=1}^m c_r(A)
	> \density \cdot 1
	= \density
	\enspace,
	\]
	where the first inequality follows from the fact that each of the elements has high density when it is added to the solution and the second inequality follows from the fact that $A$ violates at least one of the knapsack constraints.
	
	The next step is to show that between $\SetF{S}{i}{j}$ and $\{ u \}$, at least one of these has value larger than $\density / 2$.
	In particular, observe that
	\[
	\max\left\{ f(\SetF{S}{i}{j}), f(\{u\}) \right\}
	\geq \frac{1}{2} \left( f(\SetF{S}{i}{j}) + f(\{u\}) \right) 
	\geq \frac{1}{2} \left( f(\SetF{S}{i}{j} + u ) + f(\varnothing) \right) 
	\geq \frac{1}{2} f(A) 
	> \frac{\density}{2} 
	\enspace,
	\]
	where the first inequality bounds the maximum by the average, the second inequalities follows by submodularity, the third inequality follows by non-negativity, and the final inequality follows from the bound above.
	
	Recall now that the algorithm returns the set $S$ among the sets $\SetF{S}{\numiter}{1}, \dots \SetF{S}{\numiter}{\ell}$, and $\{ e \}  = \argmax_{u \in \gnd} f(u)$ which maximizes the objective value.
	One can note that the final solutions have larger objective values than the solutions at iteration $i$ (i.e., $f(\SetF{S}{\numiter}{j}) \geq f(\SetF{S}{i}{j})$) because only elements with positive marginal gains are added to the solutions by the algorithm.
	We also note that by construction of $e$, we have that $f(\{e\}) \geq f(\{u\})$ because $u$ is a feasible element.
	Together, these facts imply that
	\[
	f(S) 
	\geq \max \left\{ \SetF{S}{\numiter}{1}, \dots \SetF{S}{\numiter}{\ell}, \{e\} \right\}
	\geq \max \left\{ f(\SetF{S}{\numiter}{j}), f(e) \right\}
	\geq \max \left\{  f(\SetF{S}{i}{j}), f(u) \right\}
	> \frac{\density}{2}
	\enspace,
	\]
	which completes our proof for the case of $\knapevent = 1$.
	
	Next, we turn our attention to the case of $\knapevent = 0$.
	Recall that the algorithm returns the set $S$ among the sets $\SetF{S}{\numiter}{1}, \dots \SetF{S}{\numiter}{\ell}$, and $\{ e \}  = \argmax_{u \in \gnd} f(u)$ which maximizes the objective value.
	Therefore, to lower bound $f(S)$, it suffices to only consider the maximum over the sets $\SetF{S}{\numiter}{1}, \dots \SetF{S}{\numiter}{\ell}$.
	Applying an averaging argument to the guarantee of Corollary~\ref{cor:ks_approximation_to_union} yields
	\begin{equation} \label{eq:ks_final_lower_bound}
	f(S) 
	\geq \max \left\{ \SetF{S}{\numiter}{1}, \dots \SetF{S}{\numiter}{\ell}\right\}
	\geq \frac{1}{\ell} \sum_{j=1}^\ell f(\SetF{S}{\numiter}{j})
	\geq \frac{(1-\eps)}{(p+1)} \left( \frac{1}{\ell}\sum_{j=1}^\ell f( \OPT \cup \SetF{S}{\numiter}{j}) - \eps \maxgain - \density m \right)
	\enspace.
	\end{equation}
	Consider now a random set $\bar{S}$ chosen uniformly at random from the $\ell$ constructed solutions $\SetF{S}{\numiter}{1}, \SetF{S}{\numiter}{2}, \dotsc, \SetF{S}{\numiter}{\ell}$. 
	Since these solutions are disjoint by construction, an element can belong to $\bar{S}$ with probability at most $\ell^{-1}$. 
	Hence, by applying Lemma~\ref{lem:distribution} to the submodular function $g(S) = f(\OPT \cup S)$, we get
	\[
	\frac{1}{\ell} \cdot \sum_{j=1}^\ell f( \OPT \cup \SetF{S}{\numiter}{\ell})
	= \bE[ f(OPT \cup \bar{S}) ]
	= \bE[ g(\bar{S})]
	\geq (1 - \ell^{-1}) \cdot g( \varnothing)
	= (1 - \ell^{-1}) \cdot f(\OPT) \enspace.
	\]
	Plugging this inequality into \eqref{eq:ks_final_lower_bound}, and using the fact that $\maxgain \leq \OPT$, we obtain the lower bound
	\begin{align*}
	f(S) 
	&\geq 
		\frac{(1-\eps)}{(p+1)} \Big( 
			(1 - \ell^{-1}) \cdot f( \OPT ) - \eps \maxgain - \density m 
		\Big) \\
	&\geq 
		\frac{(1-\eps)}{(p+1)} \Big( 
		(1 - \ell^{-1}) \cdot f( \OPT ) - \eps f(\OPT) - \density m 
		\Big) \\
	&=
		\frac{(1-\eps)}{(p+1)} \Big( 
		(1 - \ell^{-1} - \eps) \cdot f( \OPT ) - \density m 
		\Big) \enspace.
	\end{align*}
	
	Suppose further that $f$ is monotone. 
	In this case, relating $f(\OPT \cup S)$ to $f(\OPT)$ is more straightforward and does not require a loss of approximation. 
	In particular, applying monotonicity directly to \eqref{eq:ks_final_lower_bound}, we get
	{\allowdisplaybreaks
	\begin{align*}
	f(S) 
	&\geq \frac{(1-\eps)}{(p+1)} \left( \frac{1}{\ell}\sum_{j=1}^\ell f( \OPT \cup \SetF{S}{\numiter}{j}) - \eps \maxgain - \density m \right) 
		&\text{(Inequality~\eqref{eq:ks_final_lower_bound})}\\
	&\geq \frac{(1-\eps)}{(p+1)} \left( \frac{1}{\ell}\sum_{j=1}^\ell f( \OPT) - \eps \maxgain - \density m \right) 
		&\text{(monotonicity)}\\
	&= \frac{(1-\eps)}{(p+1)} \Big( f( \OPT) - \eps \maxgain - \density m \Big) \\
	&\geq \frac{(1-\eps)}{(p+1)} \Big( f( \OPT) - \eps f(\OPT) - \density m \Big) 
		&\text{($\maxgain \leq f(\OPT)$)}\\
	&= \frac{(1-\eps)}{(p+1)} \Big( (1 - \eps) f( \OPT) - \density m \Big) \enspace.
	\qedhere
	\end{align*}
	}
\end{proof}
\section{Proof of Proposition~\ref{prop:mod-repeated-greedy-guarantees} (\modrepeatedgreedy)} \label{sec:mod-rg-proof}

In this section, we present a proof of Proposition~\ref{prop:mod-repeated-greedy-guarantees} which provides approximation guarantees for \modrepeatedgreedy when the density parameter $\density$ is fixed.
We begin by restating the proposition.

\begin{manualtheorem}{\ref{prop:mod-repeated-greedy-guarantees}}
	\modRGguarantees[*]
\end{manualtheorem}

The main technical aspect is to prove an approximation guarantee for \modgreedy when $(\gnd, \cI)$ is a $k$-system.
Roughly speaking, this will be similar to the analysis of the vanilla greedy algorithm for $k$-systems (Lemma~3.2 of \cite{GRST10}), but we will need to account for the marginal gain thresholding and the knapsack density technique.

In order to analyze \modgreedy, we now introduce the following lemma, which is a structural result about $k$-systems.
This lemma is implicit in the proof of Lemma~3.2 of \cite{GRST10}, but we choose to state it separately since our use of it is slightly more involved.
We remark that a nearly identical construction appears in Section~\ref{sec:analysis-k-system}.

\begin{lemma}\label{lemma:greedy_partition_opt}
	Consider in an arbitrary $X \in \cI$ and let $T$ be the number of iterations of \modgreedy.
	There exists sets $C_1, C_2,\dots C_{T+1}$ with the following properties:
	\begin{itemize}
		\item The sets $C_1, C_2,\dots C_{T+1}$ form a disjoint partition of $X$.
		\item For every integer $1 \leq t \leq T$, $|C_t| \leq k$.
		\item For every integer $1 \leq t \leq T+1$, $C_t \subseteq \{ u \mid S_{t-1} + u \in \cI  \}$.
	\end{itemize}
\end{lemma}
\begin{proof}
	We construct the sets $C_1, C_2, \dots C_{T+1}$ recursively, with knowledge of the algorithm's execution path.
	We begin by defining the last set, 
	\[
	C_{T+1} = \{ u \in X \setminus S_T \mid S_T + u \in \cI \} \enspace.
	\]
	We construct the remaining sets recursively.
	For an integer $1 \leq t \leq T$, define the set $B_t$ to be the elements in $X$ not contained in $C_{t+1} \cup \dots \cup C_{T+1}$ which are feasible to add to solution $S_{t-1}$, i.e.,
	\[
	B_t = \{ u \in (X \setminus S_{t - 1}) \setminus ( \cup_{s=t+1}^{T+1} C_s ) \mid S_{t-1} + u \in \cI \} \enspace.
	\]
	We define $C_t$ to be an arbitrary subset of $B_t$ of size $\max ( |B_t|, k)$.
	At this point, the second and third properties in the lemma follow by construction of the sets $C_1, C_2, \dots C_{T+1}$.
	In the remainder of the proof, we show that the sets $C_1, C_2, \dotsc, C_{T+1}$ satisfy the first property; that is, they form a disjoint partition of $X$. 
	
	By construction, it is clear that the sets $C_1, C_2, \dotsc, C_{T+1}$ are disjoint and that $\cup_{t=1}^{T+1} C_t \subseteq X$.
	Thus, we seek to show that $X \subseteq \cup_{t=1}^{T+1} C_t$.
	To do this, we prove the stronger guarantee that for each integer $1 \leq t \leq T+1$,
	\[
	| X \setminus (\cup_{s=t}^{T+1} C_s ) | \leq k \cdot | S_{t-1} | \enspace.
	\]
	Note that $X \subseteq \cup_{t=1}^{T+1} C_t$ follows as $S_0 = \varnothing$.
	We prove this inequality by induction, starting at $t=T+1$ as the base case and working backwards.
	By definition of $C_{T+1}$, no element of $X \setminus (C_{T+1} \cup S_T)$ can be added to $S_T$ without violating independence, and thus, $S_T$ is a base of $(X \setminus C_{T+1}) \cup S_T$.
	In contrast, $X \setminus C_{T+1}$ is an independent subset of $(X \setminus C_{T+1}) \cup S_T$ because it is a subset of the independent set $X$.
	Thus, since $(\gnd, \cI)$ is a $k$-system, 
	\[
	| X \setminus C_{T+1} | \leq k \cdot |S_T| \enspace,
	\]
	which establishes the claim for $t = T+1$.
	Assume that the claim holds for all integers $t+1, t+2, \dotsc, T+1$, and let us prove it for $t$.
	There are two cases to consider.
	First, suppose that $| C_{t} | = k$.
	In this case,
	\begin{align*}
		| X \setminus \cup_{s=t}^{T+1} C_s |
		&= | X \setminus \cup_{s=t+1}^{T+1} C_s | - |C_t| \\
		&= | X \setminus \cup_{s=t+1}^{T+1} C_s| - k \\
		&\leq k \cdot | S_t | - k \\
		&= k \cdot |S_{t-1}| \enspace,
	\end{align*}
	where the inequality follows by induction hypothesis and the first equality holds because $C_t$ is disjoint from all $C_{t+1}, \dotsc, C_{T+1}$ and $C_t \subseteq X$.
	The second case is that $|C_t| < k$.
	In this case, $C_t = B_t$ and so no element of $X \setminus ( \cup_{s=t}^{T+1}  C_s \cup S_{t-1})$ can be added to $S_{t-1}$ without violating independence, and thus $S_{t-1}$ is a base of $(X \setminus  \cup_{s=t}^{T+1}  C_s) \cup S_{t-1}$. 
	This allows us to prove the claim in the same way as we did for the base case.
	In particular, observe that $X \setminus \cup_{s=t}^{T+1}  C_s$ is an independent subset of $(X \setminus \cup_{s=t}^{T+1} C_s) \cup S_{t-1}$ because it is also a subset of the independent set $X$.
	Thus, because $(\gnd, \cI)$ is a $k$-system,
	\[
	| X \setminus \cup_{s=t}^{T+1} C_s | \leq k \cdot |S_{t-1}| \enspace,
	\]
	which completes the proof by induction.
\end{proof}

Now we are ready to prove the approximation guarantee of \modgreedy.

\begin{lemma} \label{lemma:modified-greedy-analysis}
	Suppose that $\cI$ is a $k$-system and that $S$ is the set returned by \modgreedy.
	Then, 
	\[
	f(S) \geq 
	\left\{
	\begin{array}{lr}
		\left( \frac{1-\eps}{k+1} \right) \cdot \left[ f( \OPT \cup S) - \eps \cdot \maxgain - \rho m \right] & \text{if } \knapevent = 0 \\
		\rho / 2 & \text{if } \knapevent = 1 \\
	\end{array}
	\right.
	\enspace.
	\]	
\end{lemma}
\begin{proof}
	Let $T$ denote the number of iterations in \modgreedy so that the sequence of solutions it produces is $S_0, S_1, \dotsc, S_T$, where $S_T = S$ is the solution that is returned.
	
	In the first case, suppose that $E = 1$, which is to say that Line~\ref{line:mg-knapsack-check} evaluates to \texttt{false} at some point during the execution of the algorithm.
	This happens when, at some iteration $t$ there exists a solution $S_t$ and a high density element $u$ such that adding the element to this set is feasible in the independence system, but the knapsack constraint is violated. 
	More precisely, the set $A \triangleq S_t+ u$ is independent (i.e., $A \in \cI$) but  $c_r(A) > 1$ for some knapsack function $1 \leq r \leq m$.
	Although $A$ itself is not feasible, we claim that $f(A) > \rho$.
	To this end, let us order the elements of $A$ according to the order in which they were added to $S_t$, with $u$ appearing last, i.e., $A = \{ u_1, u_2, \dotsc, u_k \}$ with $u_k = u$.
	For $1 \leq t \leq k$, define the sets $A_t = \{ u_1, u_2, \dotsc, u_t \}$ and $A_0 = \varnothing$.
	Then, we obtain the lower bound
	\[
	f(A) 
	= \sum_{t=1}^k f( u_t \mid A_{t-1})
	\geq \sum_{t=1}^k \rho \cdot \sum_{r=1}^m c_r(u_t) 
	= \rho \sum_{r=1}^m c_r(A)
	> \rho \cdot 1
	= \rho
	\enspace,
	\]
	where the first inequality follows from the fact that each of the elements has high density when it is added to the solution and the second inequality follows from the fact that $A$ violates at least one of the knapsack constraints.
	
	The next step is to show that between $S_t$ and $\{ u \}$, at least one of these has value larger than $\rho / 2$.
	In particular, observe that
	\[
	\max\left\{ f(S_t), f(\{u\}) \right\}
	\geq \frac{1}{2} \left( f(S_t)) + f(\{u\}) \right) 
	\geq \frac{1}{2} \left( f(S_t + u ) + f(\varnothing) \right) 
	\geq \frac{1}{2} f(A) 
	> \frac{\rho}{2} 
	\enspace,
	\]
	where the first inequality bounds the maximum by the average, the second inequalities follows by submodularity, the third inequality follows by non-negativity, and the final inequality follows from the bound above.
	
	Recall now that the algorithm returns the set $S$ which has the larger objective value among $S_T$ and $\{ u^*\}$. 
	One can note that the final solution has larger objective value than the solution at iteration $t$ (i.e., $f(S_T) \geq f(S_t)$) because only elements with positive marginal gains are added to the solutions by the algorithm.
	We also note that by construction of $u^*$, we have that $f(\{u^*\}) \geq f(\{u\})$ because $u$ is a feasible element.
	Together, these facts imply that
	\[
	f(S) 
	\geq \max \left\{ f(S_T), f(u^*) \right\}
	\geq \max \left\{ f(S_t), f(u) \right\}
	> \frac{\rho}{2}
	\enspace,
	\]
	which completes our proof for the case of $E = 1$.
	
	In the second case, suppose that $E = 0$ so that the algorithm never considers an element which might violate the knapsack constraints.
	We seek to upper bound the marginal gain of adding  $\OPT$ to the returned solution $S_T$.
	To this end, we begin by splitting the elements of  $\OPT$ into two sets: those elements with high density with respect to $S_T$ and those with low density.
	More precisely, define the set of low density elements to be
	\[
	\mathcal{L} = \left\{ u \in \OPT \mid f(u \mid S_T) < \rho \cdot \sum_{r=1}^m c_r(u) \right\} \enspace,
	\]
	and define the set of high density elements to be the remaining elements of \OPT,
	\[
	\mathcal{H} = \OPT \setminus \mathcal{L} 
	= \left\{ u \in \OPT \mid f(u \mid S_T) \geq \rho \cdot \sum_{r=1}^m c_r(u)  \right\}
	\enspace.
	\]
	By submodularity of $f$, we may now bound the marginal gain of adding  $\OPT$ to $S_T$ in terms of adding the high and low density elements separately as
	\begin{equation}\label{eq:greedy-split-gain-of-opt}
		f (\OPT \cup S_T) - f(S_T) \leq f(\mathcal{H} \mid S_T) + f(\mathcal{L} \mid S_T)
		\enspace.
	\end{equation}
	
	We now upper bound the marginal gain of adding the low density elements $\mathcal{L}$ to the solution $S_T$.
	Observe that
	\begin{align}
		f(\mathcal{L} \mid S_T)
		&\leq \sum_{u \in \mathcal{L}} f(u \mid S_T)
		&\text{(by submodularity of $f$)} \nonumber \\
		&\leq \sum_{u \in \mathcal{L}} \rho \cdot \sum_{r=1}^m c_r(u) 
		&\text{(definition of $\mathcal{L}$)} \nonumber \\
		&= \rho \cdot \sum_{r=1}^m  c_r(\mathcal{L})
		&\text{(by modularity)} \nonumber\\
		&\leq \rho \cdot \sum_{r=1}^m  c_r (\OPT)
		&\text{($\mathcal{L} \subset \OPT$)} \nonumber \\
		&\leq \rho m \label{eq:greedy-low-density-bound}
		\enspace,
	\end{align}
	where the last line follows because  $\OPT$ is feasible and so it satisfies the cardinality constraints $c_r (\OPT) \leq 1$ for all $1 \leq r \leq m$.
	
	We now seek to upper bound the marginal gain of adding the high density elements $\mathcal{H}$ to the solution $S_T$.
	However, this direction is more involved and it is simpler to work backwards by lower bounding the objective value of the returned solution in terms of the high density elements of  $\OPT$.
	Taking $X = \mathcal{H}$, define a partition of its elements into sets $C_1, \dotsc, C_T, C_{T+1}$ as in the statement of Lemma~\ref{lemma:greedy_partition_opt}.
	By non-negativity of $f$ and a telescoping sum, we have
	\begin{align*}
		k \cdot  f(S_T) 
		&\geq k \cdot \left( f(S_T) - f(S_0) \right) 
		&\text{(non-negativity of $f$)} \\
		&= k \sum_{t=1}^T \left[  f(S_t) - f(S_{t-1})  \right]
		&\text{(telescoping sum)} \\
		&= \sum_{t=1}^T  k \cdot  f(u_t \mid S_{t-1})
		&\text{(distributing)} \\
		&\geq \sum_{t=1}^T  | C_t| \cdot  f(u_t \mid S_{t-1})
		&\text{(by Lemma~\ref{lemma:greedy_partition_opt}, $|C_t| \leq k$)} 
	\end{align*}
	
	Note that at each iteration, the chosen element $u_t$ is a feasible high density element which has a marginal gain within a $(1 - \eps)$ multiplicative factor of the largest marginal gain among all such elements. 
	We may now use the greedy selection of the element $u_t$ and submodularity of $f$ to establish the following lower bound:
	
	\begin{align*}
		\sum_{t=1}^T |C_t| \cdot  f(u_t \mid S_{t-1}) 
		&\geq \sum_{t=1}^T |C_t| \cdot (1 - \eps) \max_{u \in C_t}  f(u \mid S_{t-1}) 
		&\text{(approx. greedy selection)} \\
		&\geq (1 - \eps) \sum_{t=1}^T |C_t| \cdot \frac{1}{|C_t|} \sum_{u \in C_t}  f(u \mid S_{t-1}) 
		&\text{(max $\geq$ average)} \\
		&\geq (1 - \eps) \sum_{t=1}^T f(C_t \mid S_{t-1})
		&\text{(submodularity of $f$)} \\
		&\geq (1 - \eps) \sum_{t=1}^T f(C_t \mid S_{T})
		&\text{(submodularity of $f$)} \\
		&= (1 - \eps)  \left[ \sum_{t=1}^{T+1} f(C_t \mid S_{T})
		- f(C_{T+1} \mid S_{T}) \right]
		&\text{(adding and subtracting term)} \\
		&\geq  (1 - \eps)  \left[ f( \cup_{t=1}^{T+1} C_t \mid S_T) - f(C_{T+1} \mid S_{T}) \right]
		&\text{(subadditivity of $f$)}  \\
		&= (1 - \eps)  \left[ f( \mathcal{H} \mid S_T ) - f(C_{T+1} \mid S_{T}) \right]
		&\text{(Lemma~\ref{lemma:greedy_partition_opt})}
		\enspace,
	\end{align*}
	where subadditivity of $f$ follows from submodularity and non-negativity.
	
	Our final goal now is to bound the value $f(C_{T+1} \mid S_{T})$, which is the marginal gain of all the elements of $\mathcal{H}$ that were not added to the final solution $S_T$, but could maintain feasibility in $\cI$ if added.
	Consider an element $e \in C_{T+1}$.
	Because $u \notin S_T$ and $E = 0$, it must be the case that the marginal gain of this element to the final solution is bounded by $f(u \mid S_{T}) < \max \left( \tau, \rho \cdot \sum_{r=1}^m c_r(e) \right)$.
	However, this element $e$ is in $\mathcal{H}$ so it has high density with respect to the solution $S_T$.
	Thus, it must be the case that $f(u \mid S_{T}) < \tau$.
	By the termination condition, we have that $\tau < (\eps / n) \cdot \maxgain$, which implies a bound on the marginal gain $f(u \mid S_{T}) < (\eps / n) \cdot \maxgain$.
	This upper bound on the marginal gain, together with submodularity of $f$ and the (trivial) cardinality bound  $|C_{T+1}| \leq n$, yields
	\[
	f(C_{T+1} \mid S_{T})
	\leq \sum_{u \in C_{T+1}} f(u \mid S_{T})
	\leq \sum_{u \in C_{T+1}} ( \eps / n ) \cdot \maxgain
	\leq \eps \cdot \left( \frac{|C_{T+1}|}{n} \right)  \maxgain
	\leq \eps \cdot \maxgain
	\enspace.
	\]
	Using these inequalities together yields an upper bound on the marginal gain of adding the high density elements to the returned solution,
	\begin{equation} \label{eq:greedy-high-density-bound}
		f( \mathcal{H} \mid S_T) 
		\leq 
		\frac{k}{1 - \eps} \cdot f(S_T)  + \eps \cdot \maxgain
		\enspace.
	\end{equation}
	Thus, we may now bound the marginal gain of adding  $\OPT$ to the final solution $S_T$ by combining the above upper bounds on adding the high and low density elements.
	More precisely, substituting inequalities~\eqref{eq:greedy-low-density-bound} and \eqref{eq:greedy-high-density-bound} into inequality~\eqref{eq:greedy-split-gain-of-opt} yields
	\[
	f (\OPT \cup S_T) - f(S_T)
	\leq f( \mathcal{H} \mid S_T) + f(\mathcal{L} \mid S_T)
	\leq \frac{k}{1-\epsilon} \cdot f(S_T) + \eps \maxgain + \rho m
	\]
	Rearranging this inequality and using the inequality $1 \leq (1-\eps)^{-1}$, we obtain
	\[
	f(S_T) \geq \frac{1-\eps}{k+1} \Big( f (\OPT \cup S_T) - \eps \maxgain - \rho m \Big) \enspace.
	\qedhere
	\] 
\end{proof}

We are now ready to prove the approximation guarantees of \modrepeatedgreedy as stated in Proposition~\ref{prop:mod-repeated-greedy-guarantees}.
The approximation analysis of \modrepeatedgreedy is similar to the approximation analysis of \repeatedgreedy in the main paper.
The main difference is that we apply Lemma~\ref{lemma:modified-greedy-analysis} when considering the \modgreedy subroutine rather than applying Lemma~3.2 of \cite{GRST10}, which holds only for the vanilla greedy algorithm (which is slower than \modgreedy and does not handle knapsack constraints). 

\begin{proof}[Proof of Proposition~\ref{prop:mod-repeated-greedy-guarantees}]
	Observe that, for every $1 \leq i \leq \ell$, we have
	\begin{equation} \label{eq:mrg-set_equalities}
		\OPT \setminus \gnd_i = \OPT \cap ( \gnd \setminus \gnd_i) = \OPT \cap \left( \cup_{j=1}^{i-1} S_i \right) = \cup_{j=1}^{i-1} \left( \OPT \cap S_j \right)
	\end{equation}
	where the first equality holds because $\OPT \subseteq \gnd$, and the second equality follows from the removal of $S_i$ from the ground set in each iteration of \modrepeatedgreedy.
	Using the previous lemmata and this observation, we can obtain a lower bound on the objective value of the returned solution $S$ in terms of the  average value of $f(S_i \cup \OPT)$ as
	{\allowdisplaybreaks
		\begin{align*}
			\frac{1}{\ell} \sum \limits_{i=1}^\ell f(S_i \cup \OPT) 
			&\leq \frac{1}{\ell} \sum \limits_{i=1}^\ell f(S_i \cup (\OPT \cap \gnd_i)) + \frac{1}{\ell} \sum \limits_{i=1}^\ell f(\OPT \setminus \gnd_i) &\text{(Lemma~\ref{lem:submodular_fact1})} \\
			&= \frac{1}{\ell} \sum \limits_{i=1}^\ell f(S_i \cup (\OPT \cap \gnd_i)) + \frac{1}{\ell} \sum \limits_{i=1}^\ell f \left(\cup_{j=1}^{i-1}(\OPT \cap S_j) \right) &\text{(Equality~\eqref{eq:mrg-set_equalities})} \\
			&\leq \frac{1}{\ell} \sum \limits_{i=1}^\ell f(S_i \cup (\OPT \cap \gnd_i)) + \frac{1}{\ell} \sum \limits_{i=1}^\ell \sum \limits_{j=1}^{i-1} f(\OPT \cap S_j) &\text{(submodularity)} \\
			&\leq \frac{1}{\ell} \sum \limits_{i=1}^\ell \left[ \frac{k+1}{1-\eps} f(S_i) + \eps \maxgain + \density m \right]  + \frac{\alpha}{\ell} \sum \limits_{i=1}^\ell \sum \limits_{j=1}^{i-1} f(S'_j) &\text{(Lemmas~\ref{lemma:modified-greedy-analysis} and \ref{lem:known_approx_results})} \\
			&\leq \frac{1}{\ell} \sum \limits_{i=1}^\ell \left[ \frac{k+1}{1-\eps} f(S) + \eps \maxgain + \density m \right]  + \frac{\alpha}{\ell} \sum \limits_{i=1}^\ell \sum \limits_{j=1}^{i-1} f(S)&\text{(definition of $S$)} \\
			&= \frac{k+1}{1-\eps} f(S) + \eps \maxgain + \density m + \frac{\alpha (\ell-1)}{2} f(S) \\
			&\leq (1-\eps)^{-1} \left( k + 1 + \alpha(\ell - 1)/2 \right) f(S) + \eps f(\OPT) + \density m
			\enspace.
		\end{align*}
	}
	Rearranging this inequality yields the following lower bound on the value of the returned solution:
	\begin{equation}\label{eq:mod-rg-lower-bound-average}
		f(S) \geq
		\left( \frac{1-\eps}{ k + 1 + \alpha (\ell - 1)/2 } \right)
		\left[ \frac{1}{\ell} \sum_{i=1}^\ell f(S_i \cup \OPT) - \eps f(\OPT) - \density m \right]
		 \enspace.
	\end{equation}
	In order to remove the dependence of the right hand side on the solutions $S_i$, we again use Lemma~\ref{lem:distribution} [Lemma~2.2 of \cite{BFNS14}].
	In particular, consider a set $\bar{S}$ chosen uniformly at random from the $\ell$ constructed solutions $S_1, S_2, \dots S_\ell$. 
	Because the solutions are disjoint by construction, an element can belong to $\bar{S}$ with probability at most $\ell^{-1}$. 
	Hence, applying Lemma~\ref{lem:distribution} to the submodular function $g(S) = f(\OPT \cup S)$, we get
	\begin{equation}\label{eq:modrg-buchbinder-lb-on-average}
		\frac{1}{\ell} \sum_{i=1}^\ell f(S_i \cup \OPT)
		= \Exp{f(\OPT \cup \bar{S})}
		= \Exp{g(\bar{S})} 
		\geq (1 - \ell^{-1}) \cdot g(\varnothing)
		= (1 - \ell^{-1}) \cdot f(\OPT) \enspace.
	\end{equation}
	Substituting \eqref{eq:modrg-buchbinder-lb-on-average} into the lower bound of \eqref{eq:mod-rg-lower-bound-average} yields the desired result.
	
	When $f$ is monotone submodular, we may obtain an improved approximation ratio by applying monotonicity directly to the lower bound \eqref{eq:mod-rg-lower-bound-average}.
	In particular, applying monotonicity yields
	\[
	\frac{1}{\ell} \sum_{i=1}^\ell f(S_i \cup \OPT)
	\geq \frac{1}{\ell} \sum_{i=1}^\ell f( \OPT)
	= f(\OPT) \enspace,
	\]
	which yields the desired approximation in the monotone setting.
\end{proof}

\end{document}